\documentclass[a4paper,11pt]{article}

\usepackage[T1]{fontenc}
\usepackage[top=1in,bottom=1in,left=1in,right=1in]{geometry}
\usepackage{pdflscape}
\usepackage{amsmath,amsthm}
\usepackage{amssymb,latexsym}
\usepackage[mathscr]{eucal}
\usepackage{dsfont}
\usepackage{graphicx}
\usepackage[usenames,dvipsnames]{color}
\usepackage{xspace}
\usepackage{comment}
\usepackage[unicode,final=true]{hyperref} 
\usepackage{multirow,bigdelim}

\newtheorem{theorem}{Theorem}

\newtheorem{lemma}[theorem]{Lemma}

\newcounter{pln}
\newenvironment{smallenum}
{\begin{list} {(\roman{pln})} { \usecounter{pln}%
     \settowidth{\labelwidth}{(iii)}%
     \setlength{\leftmargin}{\labelwidth}%
     \addtolength{\leftmargin}{1.5 \labelsep}%
     \setlength{\itemsep}{0cm}%
     \setlength{\parsep}{0cm}%
     \setlength{\topsep}{0mm} }
}
{\end{list}}

\newcommand{\reals}{\mathbb{R}}
\newcommand{\naturals}{\mathbb{N}}

\newcommand{\cC}{\mathcal{C}}

\newcommand{\eE}{\mathcal{E}}
\newcommand{\rR}{\mathcal{R}}

\newcommand{\fF}{\mathcal{F}}
\newcommand{\gG}{\mathcal{G}}

\newcommand{\qQ}{\mathcal{Q}}
\newcommand{\yY}{\mathcal{Y}}
\newcommand{\kK}{\mathcal{K}}
\newcommand{\zZ}{\mathcal{Z}}
\newcommand{\xX}{\mathcal{X}}

\newcommand{\st}{\mathrm{st}}  
\newcommand{\str}{\mathrm{star}}  
\newcommand{\y}{y_{ \{ 1 \} } }  

\newcommand{\Wavg}{\overline{W}}

\newcommand{\UU}{\mathscr{U}}
\newcommand{\VV}{\mathscr{V}}

\newcommand{\SSS}{\mathscr{S}}  
\newcommand{\Sl}[1]{\mathsf{S}(#1)}  
\newcommand{\CC}{\mathscr{C}}  
\newcommand{\restr}[1]{\mathsf{R}(#1)}  
\newcommand{\sset}{\mathfrak{S}}  

\newcommand{\sub}{\leftarrow}

\newcommand{\conv}{{\rm conv}}
\newcommand{\sm}{\setminus}
\newcommand{\nin}{\not\in}

\newcommand{\mb}[1]{{\boldsymbol{#1}}}
\newcommand{\me}{\mb{e}}
\newcommand{\mbv}{\mb{v}}

\newcommand{\mc}{\mb{\gamma}}
\newcommand{\mcc}{\mb{\beta}}
\newcommand{\z}{\zeta}

\newcommand{\sep}{,\xspace}

\newcommand{\ie}{i.e.,\xspace}
\newcommand{\eg}{e.g.,\xspace}

\newcommand{\lexp}[1]{\lfloor\frac{#1}{2}\rfloor}
\newcommand{\ltexp}[1]{\lfloor\tfrac{#1}{2}\rfloor}

\newcommand{\bx}[1]{\partial{#1}}

\newcommand{\range}[2][r]{{#2}_1,{#2}_2,\ldots,{#2}_{#1}}

\newcommand{\ub}[1]{\Phi_{#1}(\range{n})}
\newcommand{\MSsymbol}{+}
\newcommand{\MS}[1][P]{{#1}_1\MSsymbol{#1}_2\MSsymbol\cdots\MSsymbol{#1}_r}

\newcommand{\rr}{{\bar{r}}}

\newcommand{\compl}{\bar{S}}
\newcommand{\cb}[1]{\mb{\bar{#1}}}

\newcommand{\GVD}{\text{\rm{}GVD}}

\newcounter{ctrcopy}

\begin{document}

\title{The maximum number of faces of the Minkowski sum\\ of three
  convex polytopes}

\author{Menelaos I. Karavelas$^{1,2}$
\hspace*{1cm}
Christos Konaxis$^{3}$
\hspace*{1cm}
Eleni Tzanaki$^{1,2}$\\[5pt]
\it{}$^1$Department of Applied Mathematics,
\it{}University of Crete\\
\it{}GR-714 09 Heraklion, Greece\\
{\small\texttt{\{mkaravel,etzanaki\}@tem.uoc.gr}}\\[5pt]
\it{}$^2$Institute of Applied and Computational Mathematics,\\
\it{}Foundation for Research and Technology - Hellas,\\
\it{}P.O. Box 1385, GR-711 10 Heraklion, Greece\\[5pt]
\it{}$^3$Archimedes Center for Modeling, Analysis \& Computation,\\
\it{}University of Crete,
\it{}GR-710 03 Heraklion, Greece\\
{\small\texttt{ckonaxis@acmac.uoc.gr}}}

\phantomsection
\addcontentsline{toc}{section}{Title}

\maketitle

\phantomsection
\addcontentsline{toc}{section}{Abstract}

\begin{abstract}
  We derive tight expressions for the maximum
  number of $k$-faces, $0\le{}k\le{}d-1$, of the
  Minkowski sum, $P_1+P_2+P_3$, of three $d$-dimensional convex
  polytopes $P_1$, $P_2$ and $P_3$, as a function of the number of
  vertices of the polytopes, for any $d\ge{}2$.
  Expressing the Minkowski sum of the three polytopes
  as a section of their Cayley polytope $\cC$,
  the problem of counting the number of $k$-faces of $P_1+P_2+P_3$,
  reduces to counting the number of $(k+2)$-faces of the subset  
  of $\cC$ comprising of the faces that contain at least one vertex from
  each $P_i$.
  In two dimensions our expressions reduce to known results, 
  while in three dimensions, the tightness of our bounds follows
  by exploiting known tight bounds for the number of faces of $r$
  $d$-polytopes, where $r\ge d$.
  For $d\ge{}4$, the maximum values are attained when
  $P_1$, $P_2$ and $P_3$ are $d$-polytopes, 
  whose vertex sets are chosen appropriately from three
  distinct $d$-dimensional moment-like curves.

  \bigskip\noindent
  \textit{Key\;words:}\/
  high-dimensional geometry\sep discrete geometry\sep
  combinatorial geometry\sep combinatorial complexity\sep
  Cayley trick\sep tight bounds\sep
  Minkowski sum\sep convex polytopes
  \medskip\par\noindent
  \textit{2010 MSC:}\/ 52B05, 52B11, 52C45, 68U05
\end{abstract}

\clearpage

\section{Introduction}
\label{sec:intro}

We study the Minkowski sum of three $d$-dimensional convex polytopes, or, simply $d$-polytopes,
and derive tight upper bounds for the number of its $k$-faces, for $0\leq k \leq d-1$,
with respect to the number of vertices of the summands.
Given two  convex polytopes $P_1$ and $P_2$, their \emph{Minkowski sum} $P_1 + P_2$ is the set
$\{ p_1+p_2~|~p_1\in P_1, p_2 \in P_2\}$. This definition extends to any number of summands
and also, to non-convex sets of points.
The Minkowski sum of convex polytopes is itself a convex polytope, namely, the convex hull of the 
Minkowski sum of the vertices of its summands.

Minkowski sums are widespread operations in
Computational Geometry and find applications in a wide range of areas such as
robot motion planning \cite{Lat91},
pattern recognition \cite{trhh-smcpma-00}, collision detection \cite{lm-cpq-03},
Computer-Aided Design, and, very recently, Game Theory.
They reflect geometrically some algebraic operations, and capture
important properties of algebraic objects, such as polynomial
systems. This makes them especially useful in Computational Algebra,
see \eg \cite{gs-mapcc-93,s-gbcp-96,CLO2}.

The geometry of the Minkowski sum can be derived from that 
of its summands: its \emph{normal fan}  is the \emph{common refinement}
of the normal fans of the summands (see \cite{z-lp-95} for definitions and details). 
However, its combinatorial structure is not fully understood,
partially due to the fact that most algorithms for computing Minkowski sums
have focused on low dimensions (see, \eg \cite{f-mscaa-08}
for algorithms computing Minkowski sums in three dimensions).
The recent development of algorithms
that target high dimensions \cite{f-zcmacp-04}, has led to a more
extensive study of their properties (see, \eg \cite{w-mspcc-07}).

A natural and fundamental question regarding the combinatorial properties 
of Minkowski sums, concerns their complexity measured as a function
of the vertices, or the facets of the summands. 
A complete answer, in terms of the number of vertices or facets of the summands,
does not yet exist although for certain classes of
polytopes the question has been resolved (see Section~\ref{ssec:prev}). 
Most of the known results offer tight bounds with respect to
the number of vertices of the summands;
deriving tight upper bounds with respect to the number of facets seems much harder. 
Knowing the complexity of Minkowski sums is crucial
in developing algorithms for their computation, since it allows 
to quantify their efficiency. 

\paragraph{Preliminaries.}\label{ssec:prelim} 
Let $P$ be a $d$-polytope; its dimension is the dimension of its affine span.
The faces of $P$ are $\emptyset,P$, and the intersections of $P$ with its supporting hyperplanes.
The former faces are called improper while the latter faces are called proper. Each
face of $P$ is itself a polytope, and a face of dimension $k$ is called a $k$-face.
Faces of $P$ of dimension $0,1,d-2$ and $d-1$ are called vertices, edges,
ridges, and facets, respectively.

A $d$-dimensional \emph{polytopal complex, or simply $d$-complex},
$\CC$ is a finite collection of polytopes
in $\reals^d$ such that
(i) $\emptyset\in\CC$,
(ii) if $P\in\CC$ then all the faces of $P$ are also in $\CC$ and
(iii) the intersection $P\cap{}Q$ for two polytopes $P$ and $Q$ in
$\CC$ is a face of both. 
The dimension $\dim(\CC)$ of $\CC$ is the largest dimension of a
polytope in $\CC$.
A polytopal complex is called \emph{pure} if all its maximal (with
respect to inclusion) faces have the same dimension. 
In this case the maximal faces are called the \emph{facets} of
$\CC$. 
A polytopal complex is simplicial if all its faces are simplices. 
Finally, a polytopal complex $\CC'$ is called a \emph{subcomplex}
of a polytopal complex $\CC$ if all faces of $\CC'$ are also faces
of $\CC$.
For a polytopal complex $\CC$, the \emph{star} of $v$ in $\CC$,
denoted $\str(v,\CC)$, is the subcomplex of $\CC$ consisting of
all faces that contain $v$, and their faces.
The \emph{link} of $v$, denoted by $\CC/v$, is the
subcomplex of $\str(v,\CC)$ consisting of all the
faces of $\str(v,\CC)$ that do not contain $v$.

One important class of polytopal complexes arises from polytopes.
More precisely, a $d$-polytope $P$, together with all its
faces and the empty set, form a $d$-complex, denoted
by $\CC(P)$. The only maximal face of $\CC(P)$, which is clearly
the only facet of $\CC(P)$, is the polytope $P$ itself.
Moreover, all proper faces of $P$ form a pure $(d-1)$-complex,
called the \emph{boundary complex} $\CC(\bx{}P)$, or simply
$\bx{}P$, of $P$. The facets of $\bx{}P$ are just the facets of $P$.

For a $d$-polytope $P$, or its boundary complex $\bx{}P$,
we can define its $f$-vector  as
$\mb{f}(P) = (f_{-1},f_0, f_1,\dots, f_{d-1})$, where
$f_k = f_k(P)$ denotes the number of $k$-faces of $P$ and
$f_{-1}(P):=1$ corresponds to the empty face of $P$.
From the $f$-vector of $P$ we define its $h$-vector as the vector
$\mb{h}(P)=(h_0,h_1,\ldots,h_d)$, where
$h_k=h_k(P)=\sum_{i=0}^k(-1)^{k-i}\tbinom{d-i}{d-k}f_{i-1}(P)$,
$0\leq k \leq d$.

Let $\CC$ be a pure simplicial polytopal $d$-complex. A
shelling $\Sl{\CC}$ of $\CC$ is a linear ordering
$F_1,F_2,\ldots,F_s$
of the facets of $\CC$ such that for all $1<j\le{} s$ the
intersection, $F_j\cap\left(\bigcup_{i=1}^{j-1}F_i\right)$, of the
facet $F_j$ with the previous facets is non-empty and pure
$(d-1)$-dimensional.
In other words, for every $i<j$ there exists some $\ell <j$ such that
the intersection $F_i\cap{}F_j$ is contained in $F_\ell\cap{}F_j$, 
and such that $F_\ell\cap{}F_j $ is a facet of $F_j$.

Every pure polytopal complex that has a shelling is called \emph{shellable}.
In particular, the boundary complex of a polytope is always shellable
(cf. \cite{bm-sdcs-71}).
Consider a pure shellable simplicial polytopal complex $\CC$ and let
$\Sl{\CC}=\{F_1,\ldots,F_s\}$ be a shelling order of its facets.
The \emph{restriction} $\restr{F_j}$ of a facet $F_j$ is the set
of all vertices $v\in{}F_j$ such that $F_j\sm\{v\}$ is contained in
one of the earlier facets.\footnote{For simplicial faces, we
identify the face with its defining vertex set.}
The main observation here is that when we construct $\CC$ according to
the shelling $\Sl{\CC}$, the new faces at the $j$-th step of the
shelling are exactly the vertex sets $G$ with
$\restr{F_j}\subseteq{}G\subseteq{}F_j$ (cf. \cite[Section 8.3]{z-lp-95}).
Moreover, notice that $\restr{F_1}=\emptyset$ and
$\restr{F_i}\ne{}\restr{F_j}$ for all $i\ne{}j$.

\paragraph{Previous work.}\label{ssec:prev}
The complexity of Minkowski sums depends on the geometry of their summands.
Worst-case tight upper bounds offer the best possible alternative
when the geometric characteristics of a specific instance
of the problem are not accounted for. 
Gritzman and Sturmfels \cite{gs-mapcc-93} have been the first to
derive tight upper bounds for the number of $k$-faces of
$P_1+\cdots+P_r$, namely:
\begin{equation*}
f_k(P_1+\cdots+P_r)\leq 2\binom{m}{k}\sum_{j=0}^{d-k-1}\binom{m-k-1}{j},
\qquad 0\le{}k\le{}d-1,
\qquad d,r\ge{}2,
\end{equation*}
where $m$ denotes the number of non-parallel edges of $P_1,\ldots,P_r$.
Equality occurs when $P_i$ are generic zonotopes, \ie when each $P_i$
is a Minkowski sum of edges, and the generating edges of all polytopes
are in general position.

Our knowledge of tight upper bounds for $f_k(P_1+\cdots+P_r)$ as a
function of the number of vertices or facets of the summands is much
more limited, while the problem of finding such tight bounds is far from
being fully understood and resolved.
Given two polygons $P_1,P_2$ in two dimensions, with $n_1,n_2$
vertices (or edges) respectively, their Minkowski sum can have at most
$n_1+n_2$ vertices; clearly, this bound holds also for the number of
edges of $P_1+P_2$, and generalizes in the obvious way for any number
of summands (cf.~\cite{bkos-cgaa-00}).

In three or more dimensions,
Fukuda and Weibel \cite{fw-fmacp-07} have shown what they call the
\emph{trivial upper bound}: given $r$ $d$-polytopes
$\range{P}$ in $\reals^d$, where $d\ge{}3$ and $r\ge{}2$, we have
\begin{equation}\label{equ:trivial-ub}
  f_k(\MS)\le{}\ub{k+r},
\end{equation}
where $n_i$ is the number of vertices of $P_i$, $1\le{}i\le{}r$, and
\begin{equation*}
  \ub{\ell}=\sum_{\substack{1\le{}s_i\le{}n_i\\s_1+\ldots+s_r=\ell}}
  \prod_{i=1}^r\binom{n_i}{s_i},\qquad \ell\ge{}r,\qquad s_i\in\naturals.
\end{equation*}
In the same paper, Fukuda and Weibel have shown that the trivial upper
bound is tight for: (i) $d\ge{}4$, $2\le{}r\le\lexp{d}$ and for all
$0\le{}k\le\lexp{d}-r$, and (ii) for the number of vertices,
$f_0(\MS)$, of $\MS$, when $d\ge{}3$ and $2\le{}r\le{}d-1$.
For $r\ge{}d$, Sanyal \cite{s-tovnm-09} has shown that the trivial
bound for $f_0(\MS)$ cannot be attained, since in this case:
\begin{equation*}
  f_0(\MS)\le\left(1-\frac{1}{(d+1)^d}\right)\prod_{i=1}^{r}n_{i}
  <\prod_{i=1}^{r}n_{i}.
\end{equation*}
Karavelas and Tzanaki \cite{kt-tlbnf-11} recently extended 
the range of of $d$, $r$ and $k$ for which the trivial upper bound
\eqref{equ:trivial-ub} is attained. More
precisely, they showed that for any $d\ge{}3$, $2\le{}r\le{}d-1$ and for
all $0\le{}k\le{}\lexp{d+r-1}-r$, there exist $r$ neighborly
$d$-polytopes $\range{P}$ in $\reals^d$, for which the number of
$k$-faces of their Minkowski sum attains the trivial upper bound. 
Recall that a $d$-polytope $P$ is neighborly if any subset of $\lexp{d}$ 
or less vertices is the vertex set of a face of $P$.
Tight bounds for $f_0(\MS)$, where $r\ge{}d$, have very recently been
shown by Weibel \cite{w-mfmsl-12}, namely:
\begin{equation*}
  f_0(\MS)\le\alpha+\sum_{j=1}^{d-1}(-1)^{d-1-j}\binom{r-1-j}{d-1-j}
  \sum_{S\in\sset_j^r}\left(\prod_{i\in{}S}f_0(P_i)-\alpha\right),
\end{equation*}
where $\sset_j^r$ is the family of subsets of $\{1,2,\ldots,r\}$ of
cardinality $j$, and $\alpha=2(d-2\lexp{d})$.

Tight bounds for \emph{all} face numbers,
\ie for all $0\le{}k\le{}d-1$, expressed as a function of the number of
vertices or facets of the summands, are known only for \emph{two}
$d$-polytopes when $d\ge{}3$.
Fukuda and Weibel \cite{fw-fmacp-07} have shown that, given two
3-polytopes $P_1$ and $P_2$ in $\reals^3$, the number of
$k$-faces of $P_1\MSsymbol{}P_2$, $0\le{}k\le{}2$, is bounded from
above as follows:
\begin{equation}\label{equ:ub3vertices}
  \begin{aligned}
    f_0(P_1\MSsymbol{}P_2)&\le{}n_1 n_2,\\
    f_1(P_1\MSsymbol{}P_2)&\le{}2n_1 n_2+n_1+n_2-8,\\
    f_2(P_1\MSsymbol{}P_2)&\le{}n_1 n_2+n_1+n_2-6,
  \end{aligned}
\end{equation}
where $n_i$ is the number of vertices of $P_i$, $i=1,2$. These bounds
are tight.
Weibel \cite{w-mspcc-07} has derived analogous tight expressions in
terms of the number of facets $m_i$ of $P_i$, $i=1,2$:
\begin{equation}\label{equ:ub3facets}
  \begin{aligned}
    f_0(P_1\MSsymbol{}P_2)&\le{}4m_1 m_2-8m_1-8m_2+16,\\
    f_1(P_1\MSsymbol{}P_2)&\le{}8m_1 m_2-17m_1-17m_2+40,\\
    f_2(P_1\MSsymbol{}P_2)&\le{}4m_1 m_2-9m_1-9m_2+26.
  \end{aligned}
\end{equation}
Weibel's expression for $f_2(P_1\MSsymbol{}P_2)$
(cf. rel. \eqref{equ:ub3facets}) has been generalized to the number of
facets of the Minkowski sum of any number of 3-polytopes
by Fogel, Halperin and Weibel \cite{fhw-emcms-09}; they have shown
that, for $r\ge{}2$, the following tight bound holds:
\begin{equation*}
  f_2(\MS)\le\sum_{1\le{}i<j\le{}r}(2m_i-5)(2m_j-5)
  +\sum_{i=1}^r{}m_i+\binom{r}{2},
\end{equation*}
where $m_i=f_2(P_i)$, $1\le{}i\le{}r$. 
Finally, Karavelas and Tzanaki
\cite{kt-mnfms-12} have shown that for any two $d$-polytopes $P_1$
and $P_2$ in $\reals^d$, where $d\ge{}4$, and for all $1\le{}k\le{}d$,
we have:
\begin{equation}\label{equ:2p-any-d}
  f_{k-1}(P_1\MSsymbol{}P_2)\le{}f_k(C_{d+1}(n_1+n_2))
  -\sum_{i=0}^{\lexp{d+1}}\tbinom{d+1-i}{k+1-i}
  \left(\tbinom{n_1-d-2+i}{i}+\tbinom{n_2-d-2+i}{i}\right),
\end{equation}
where $n_i=f_0(P_i)$, $i=1,2$,
and $C_d(n)$ stands for the cyclic
$d$-polytope with $n$ vertices. The bounds in
\eqref{equ:2p-any-d} have been shown to be tight,
and match the corresponding, previously known,
bounds for 2- and 3-polytopes (cf. rel. \eqref{equ:ub3vertices}).

\paragraph{Overview.}\label{ssec:overview}
In this work we continue the line of research in \cite{kt-mnfms-12},
extending the methods to the case of three $d$-polytopes in $\reals^d$.
This turns out to be far from trivial.
Allowing just one more summand 
significantly raises the problem's intricacy.
On the other hand, the case of three $d$-polytopes provides
a valuable insight towards our ultimate goal,
the general case of $r$ $d$-polytopes, for any $d,r\ge{}2$.
Using the tools and methodology applied in this paper, 
some of the results obtained here can 
be generalized to the case $d,r\geq 2$ (see Section~\ref{sec:concl}),
while others still remain elusive.

We state our main result, also presented in Theorem~\ref{thm:final-result}.
Let $P_1$, $P_2$ and $P_3$ be three $d$-polytopes in $\reals^d$,
$d\ge{}2$, with $n_i\ge{}d+1$ vertices, $1\le{}i\le{}3$. Then, for all
$1\le{}k\le{}d$, we have:
\begin{equation*}
  \begin{aligned}
    f_{k-1}(P_1+P_2+P_3)&\le
    f_{k+1}(C_{d+2}(n_{[3]}))-\sum_{i=0}^{\lexp{d+2}}\binom{d+2-i}{k+2-i}
    \sum_{\emptyset\subset{}S\subset[3]}(-1)^{|S|}\binom{n_S-d-3+i}{i}\\
    &\quad-\delta\binom{\lexp{d}+1}{k-\lexp{d}}
    \sum_{i=1}^3\binom{n_i-\lexp{d}-2}{\lexp{d}+1},
  \end{aligned}
\end{equation*}
where $[3]=\{1,2,3\}$, $\delta=d-2\lexp{d}$, and $n_S=\sum_{i\in{}S}n_i$,
$\emptyset \subset S\subseteq [3]$.
Moreover, for any $d\ge{}2$, there exist three $d$-polytopes in
$\reals^d$ for which the bounds above are attained.

To establish the upper bounds we first lift the three $d$-polytopes
in $\reals^{d+2}$ using an affine basis of $\reals^2$,
and form the convex hull $\cC$ of the embedded polytopes in $\reals^{d+2}$.
$\cC$ is known as the \emph{Cayley polytope} of the $P_i$'s
(see Section \ref{sec:cayley}).
Exploiting the bijection between  the set $\fF_{[3]}$, consisting of
the $k$-faces of $\cC$ that contain vertices from each $P_i$, and the
$(k-2)$-faces of $P_1+P_2+P_3$, we reduce the derivation of upper
bounds for $f_{k-2}(P_1+P_2+P_3)$ to deriving upper bounds for
$f_k(\fF_{[3]})$, $2\le{}k\le{}d+1$.

The rest of our proof follows the main steps of McMullen's proof of the 
Upper bound Theorem for polytopes \cite{m-mnfcp-70}. 
In Section \ref{sec:ds} we add auxiliary vertices to appropriate 
faces of the Cayley polytope $\cC$, resulting in a simplicial
polytope $\qQ$ whose face set contains $\fF_{[3]}$.
We then consider the $f$-vector $\mb{f}(\bx\qQ)$
and the $h$-vector $\mb{h}(\bx\qQ)$ of $\bx\qQ$
and derive expressions for their entries via the corresponding vectors
for $\fF_{[3]}$. Using these expressions, we
continue by deriving Dehn-Sommerville-like equations for $\fF_{[3]}$.
As an intermediate step we define the subcomplex
$\kK_{[3]}$ of $\cC$ as the closure under subface inclusion of $\fF_{[3]}$, 
and derive expressions for its $f$- and $h$-vectors 
(cf. relations \eqref{eq:fk_KR} and \eqref{equ:hkKR} with $R=[3]$).
This allows us to write the Dehn-Sommerville-like equations for $\fF_{[3]}$
in the very concise form:
\begin{equation*}
  h_{d+2-k}(\fF_{[3]})=h_k(\kK_{[3]}),\qquad 0\le{}k\le{}d+2.
\end{equation*}

In Section \ref{sec:recur} we establish a recurrence relation for
the elements of $\mb{h}(\fF_{[3]})$ (see Lemma~\ref{lem:hkWrecur}).
Our starting point is a well known 
relation by McMullen (cf. rel. \eqref{equ:links}),
and the expressions for the $h$-vector of $\bx{}\qQ$ already 
established in the previous section.
The recurrence relation for the elements of $\mb{h}(\fF_{[3]})$ is 
then used in Section \ref{sec:ub} to 
prove upper bounds on the elements of
$\mb{h}(\fF_{[3]})$ and $\mb{h}(\kK_{[3]})$.
These upper bounds combined with the Dehn-Sommerville-like
equations for $\fF_{[3]}$, yield refined upper bounds  for the values
$h_k(\fF_{[3]})$ when $k>\lexp{d+2}$.
We end by establishing our upper bounds on
the number of $k$-faces, $0\le{}k\le{}d-1$, of the Minkowski sum of
three $d$-polytopes by computing $\mb{f}(\fF_{[3]})$ from
$\mb{h}(\fF_{[3]})$. At the same time we establish
conditions on a subset of the elements of the vectors
$\mb{f}(\fF_R)$, $\emptyset\subset{}R\subseteq{}[3]$, that are
sufficient and necessary in order for the upper bounds in the number
of $k$-faces of $P_1+P_2+P_3$ to be tight for all $k$ ($\fF_R$
stands for the set of faces of $\cC$ that have at least one vertex
from each $P_i$ for all $i\in{}R$).

In Section \ref{sec:lb} we describe the constructions that
establish the tightness of our upper bounds.
For $d=2$ and $d=3$ we rely on previous results. For $d\geq 4$
we define three convex $d$-polytopes, whose vertices lie on three distinct
moment-like $d$-curves, and show that the sets $\fF_R$,
$\emptyset{}\subset{}R\subseteq[3]$, associated with them satisfy the
sufficient and necessary conditions mentioned above.
We conclude with Section \ref{sec:concl}, where we discuss the case
of four or more summands and directions for future work.

\section{The Cayley trick}
\label{sec:cayley}

Recall that $[3]$ stands for the set $\{1,2,3\}$, and denote by
$\sset_j^3:=\{ R\subseteq [3] \mid |R|=j\}$,
the set of all subsets of $[3]$ of cardinality $j$,  for $1\le j \le 3$.
To keep the notation lean, in the rest of this paper we shall denote
$\sset_j^3$ as $\sset_j$.
Consider three $d$-polytopes $P_1$, $P_2$ and $P_3$ in $\reals^d$, and choose
the basis $\me_{2,1}=(0,0)$, $\me_{2,2}=(1,0)$, $\me_{2,3}=(0,1)$,
as the preferred affine basis of $\reals^2$.
The \emph{Cayley embedding} of the $P_i$'s is defined via the maps
$\mu_i(\mb{x})=(\me_{2,i},\mb{x})$, and we denote by $\cC$ the
$(d+2)$-polytope we get by taking the convex hull of the sets
$\VV_i=\{\mu_i(\mbv)\mid{}v\in{}V_i\}$, where $V_i$ is the vertex set
of $P_i$. This is known as the \emph{Cayley polytope} of the $P_i$.
Similarly,  by taking appropriate affine bases we define the Cayley 
polytope $\cC_R$ of all polytopes $P_i$, $i\in R$,
where $R\in \sset_j$, $j=1,2$. These are the Cayley polytopes of all
pairs of $P_i$'s and, trivially,  the $P_i$'s themselves.
Clearly, $\cC_R \equiv P_i$, for $R\in \sset_1$.
Moreover, $\cC\equiv \cC_{[3]}$.

For any $\emptyset \subset R \subseteq [3]$, let $\VV_R$ denote the 
union of the sets $\VV_i,\, i\in R$.
In the sequel we shall identify $\cC_R \subset \reals^{d+|R|-1}$, for 
all $R\in \sset_j,\, j=1,2$, with the affinely
isomorphic and combinatorially equivalent polytope 
$\conv(\VV_R)\subset \cC \subset \reals^{d+2}$. 
This will allow us to study properties of these subsets 
of $\cC$ by examining the 
corresponding Cayley polytopes which lie in lower dimensional spaces.

We shall denote by $\fF_R$, $\emptyset\subset{}R\subseteq[3]$, the set
of proper faces of $\cC$, with the property that $F\in\fF_R$ if
$F\cap{}\VV_i\ne\emptyset$, for all $i\in{}R$. In other words,
$\fF_R$ consists of all the faces of $\cC$ that have at least one vertex
from each $\VV_i$, for all $i\in{}R$. 
Clearly, if $|R|\geq 2$, then $f_0(\fF_R)=0$.
Moreover, if $R\in \sset_1$ then $\fF_R\equiv{}\bx{}P_i$.
The dimension of $\fF_R$ is the maximum
dimension of the faces in $\fF_R$, \ie
$\dim(\fF_R)=\max_{F\in\fF_R}\dim(F)=d+|R|-2$.

We call $\Wavg$ the $d$-flat of $\reals^{d+2}$:
\begin{equation*}
  \Wavg=
  \{\tfrac{1}{3}\me_{2,1}+\tfrac{1}{3}\me_{2,2}+\tfrac{1}{3}\me_{2,3}\}\times\reals^d,
\end{equation*}
and consider the \emph{weighted} Minkowski sum $\tfrac{1}{3}P_1+\tfrac{1}{3}P_2+\tfrac{1}{3}P_3$.
Note that this is nothing more than $P_1+P_2+P_3$, scaled down by $\tfrac{1}{3}$,
hence these two sums are combinatorially equivalent.
The \emph{Cayley trick} \cite{hrs-ctlsb-00} says that the intersection of $\Wavg$ with
$\cC$ is combinatorially equivalent (isomorphic) to the weighted Minkowski sum
$\tfrac{1}{3}P_1+\tfrac{1}{3}P_2+\tfrac{1}{3}P_3$, hence also to the unweighted
Minkowski sum $P_1+P_2+P_3$ (see also Fig. \ref{fig:cayley3}).
Moreover, every face of $P_1+P_2+P_3$ is the intersection of a face
of $\fF_{[3]}$ with $\Wavg$. This implies that:
\begin{equation*}
 f_{k-1}(P_1+P_2+P_3)=f_{k+1}(\fF_{[3]}), \qquad 1\leq k \leq d.
\end{equation*}

\begin{figure}[ht]
  \begin{center}
    \includegraphics[width=0.6\textwidth,page=1]{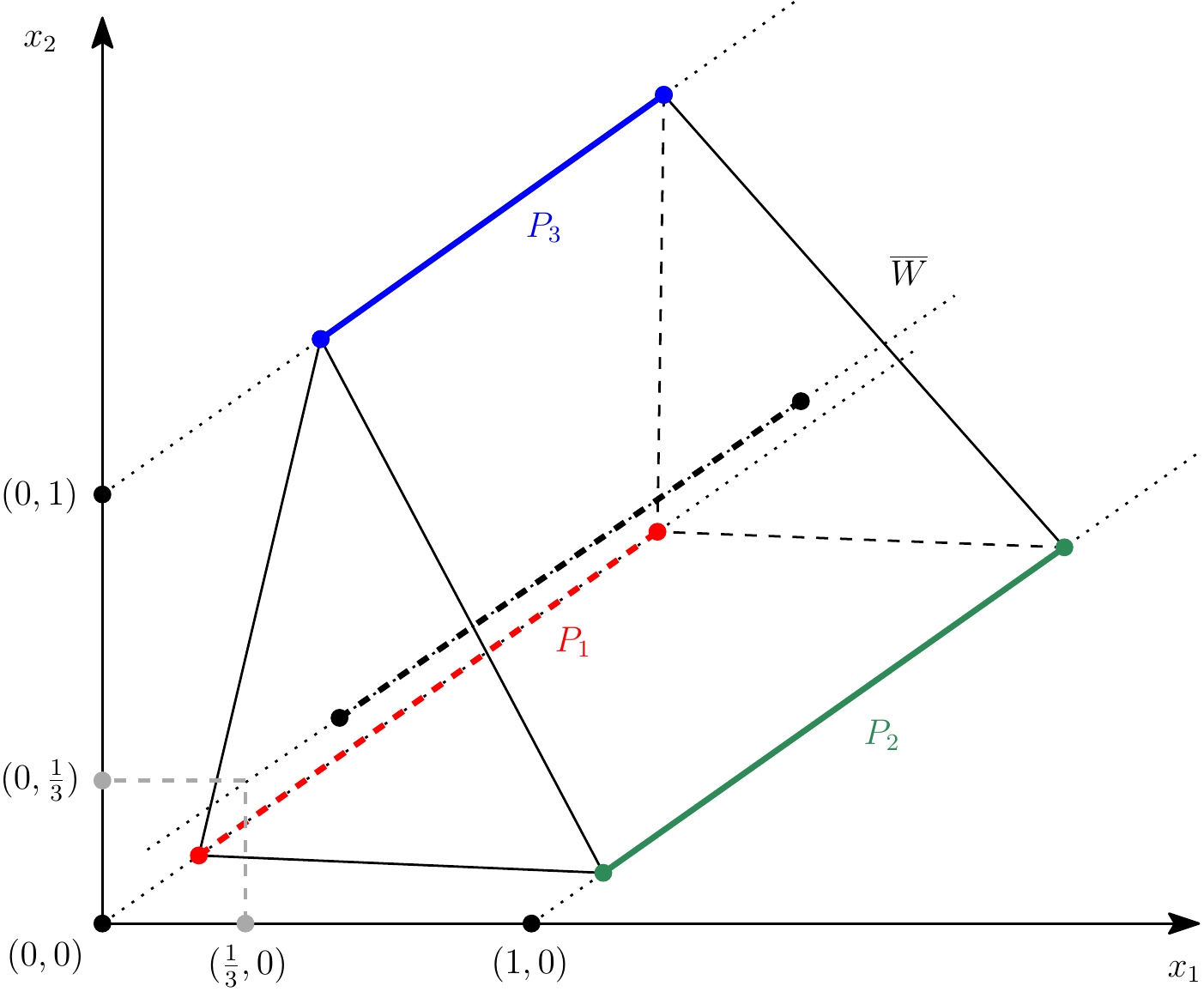}
    \caption{Schematic of the Cayley trick for three polytopes. The
      three polytopes $P_1$, $P_2$ and $P_3$ are shown in red, green
      and blue, respectively. The polytope
      $\frac{1}{3}P_1+\frac{1}{3}P_2+\frac{1}{3}P_3$ is shown in
      black.}\label{fig:cayley3}
  \end{center}
\end{figure}

To compute the upper bounds for the number of $k$-faces of $P_1+P_2+P_3$, 
in the rest of the paper we assume that $\cC$ is 
\emph{``as simplicial as possible''},
\ie all faces of $\cC$ are simplicial except for the trivial faces of
$\cC_R$, for all $\emptyset\subset{}R\subseteq{}[3]$.
Otherwise, we can employ the so called \emph{bottom-vertex triangulation} \cite{Mat02},
where we triangulate every face of $\cC$ except the trivial faces of $\cC_R$ for all $\emptyset\subset{}R\subseteq{}[3]$.
The resulting complex is polytopal and all of its faces are
simplicial, except from the seven trivial faces above.
Moreover, it has the same number of vertices as $\cC$, while the number
of its $k$-faces is never less than the number of $k$-faces of $\cC$. 

Under the \emph{``as simplicial as possible''} assumption above, the
faces in $\fF_R$ are simplicial. We shall denote by $\kK_R$ the
\emph{closure}, under subface inclusion, of $\fF_R$, \ie $\kK_R$
contains all the faces in $\fF_R$ and all the faces that are subfaces
of faces in $\fF_R$. It is easy to see that $\kK_R$ does not contain
any of the trivial faces of $\cC_S$, $S\subseteq{}R$, and, thus,
$\kK_R$ is a pure simplicial $(d+|R|-2$)-complex, whose facets are
precisely the facets in $\fF_{R}$.
It is also clear that  $\fF_R \equiv \kK_R \equiv \bx{P_R}$, for $R\in \sset_1$.
Moreover, $\kK_{[3]}$ is the boundary complex
$\bx{}\cC$ of the Cayley polytope $\cC$, except for its three facets (\ie
$(d+1)$-faces) $\cC_{R}$, $R\in\sset_2$, and its three ridges (\ie $d$-faces)
$P_i$, $1\le{}i\le{}3$. 

Consider a $k$-face $F$ of $\kK_R,\,\emptyset \subset R \subseteq [3]$.
By the definition of $\kK_R$, $F$ is either a $k$-face of $\fF_R$, 
or a $k$-face of $\fF_S$ for some nonempty 
subset $S$ of $R$. Hence
\begin{equation}\label{eq:fk_KR}
  f_k(\kK_R)=\sum_{\emptyset\subset{}S\subseteq{}R}f_k(\fF_S),
  \qquad -1\le{}k\le{}d+|R|-2,
\end{equation}
where, in order for the above equation to hold for $k=-1$, we set
$f_{-1}(\fF_R)=(-1)^{|R|-1}$. In what follows we use the convention that
$f_k(\fF_R)=0$, for any $k<-1$ or $k>d+|R|-2$.
\section{\texorpdfstring{$f$}{f}-vectors,
  \texorpdfstring{$h$}{h}-vectors and Dehn-Sommerville-like equations}
\label{sec:ds}

We are going to define auxiliary vertices in $\reals^{d+2}$ not contained in $\VV_i, i=1,2,3$.
For every $\emptyset \subset R \subset [3]$ we add a vertex $y_R$ in the relative interior of $\cC_R$
and, following~\cite{EwSh74}, we consider the complex arising by taking successive 
stellar subdivisions of $\bx{}\cC$ as follows:
\begin{enumerate}
 \item we form the complex arising from $\bx{\cC}$ by taking the
   stellar subdivisions $\st(y_{\{i\}},\cC_{\{i\}})$ for all
   $1\le{}i\le{}3$, then
 \item we form the complex arising from the one constructed in the
   previous step by taking the stellar subdivisions
   $\st(y_{R},\cC'_{R})$ for every  $R\in \sset_2$, 
   where $\cC'_{R}$ is the complex obtained by taking, 
  for every $S\subset R$, the stellar subdivision of $y_S$ 
  over the boundary complex of $\cC_S$.
\end{enumerate}
This complex is polytopal and  isomorphic to the boundary complex of a
$(d+2)$-polytope which we shall denote as $\qQ$  (see also Fig. \ref{fig:Q}).
The boundary complex $\bx{}\qQ$ is a simplicial $(d+1)$-sphere.
The simpliciality of $\bx\qQ$ will allow us to utilize its
Denh-Sommerville equations in order to prove Dehn-Sommerville-like
equations for $\fF_{[3]}$ in the upcoming
Lemma~\ref{lem:hkK-hkQ-dehn-sommervile}. We denote by
$\VV:=\VV_1\cup\VV_2\cup\VV_3\cup\{y_R\mid\emptyset\subset{}R\subset[3]\}$
the vertex set of $\qQ$.

\begin{figure}[t]
  \begin{center}
    \includegraphics[width=0.6\textwidth,page=2]{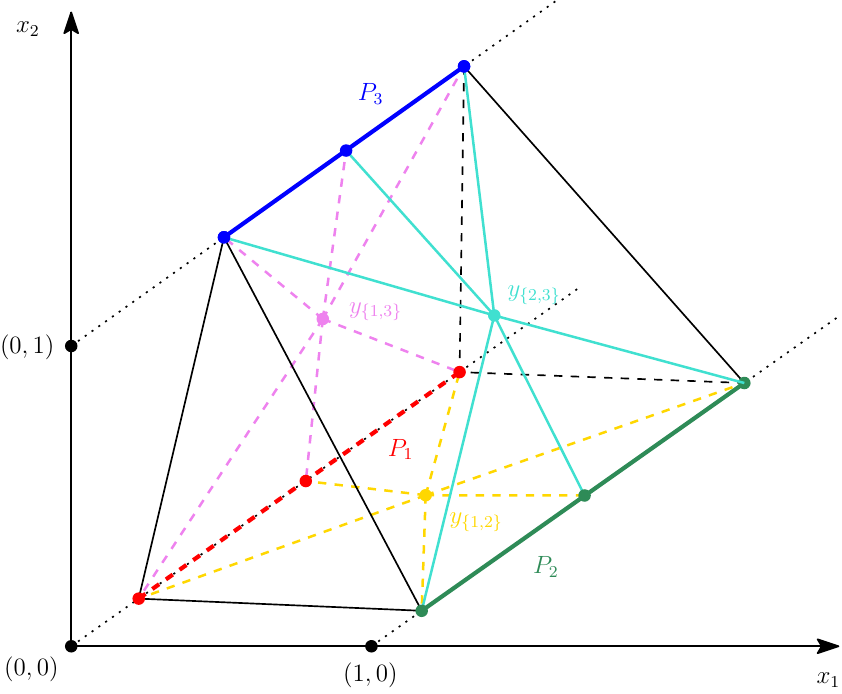}
    \caption{The $(d+2)$-polytope $\qQ$.}
    \label{fig:Q}
  \end{center}
\end{figure}

Let us count the $k$-faces of $\bx{}\qQ$. 
Suppose that $F$ is a $k$-face of $\bx{}\qQ$.
We distinguish between the following cases depending on 
the number of auxiliary vertices, $y_R$, that $F$ contains:
\begin{enumerate}
\item $F$ does not contain any additional auxiliary vertices. Then, it
  can be a $k$-face of any $\fF_R,R\in\sset_1$,
  or it can be a $k$-face  of any of the $\fF_R,R\in\sset_2$, 
  or it can be a $k$-face  of $\fF_{[3]}$. This gives a total of
  $f_k(\fF_{[3]})+\sum_{R\in\sset_1}f_k(\fF_R)+\sum_{R\in\sset_2}f_k(\fF_R)$
  $k$-faces of $\bx\qQ$.

\item $F$ contains one auxiliary vertex. Then, it can consist of  a $(k-1)$-face of:
  \begin{enumerate}
    \item
      $\fF_{R}, R\in \sset_1$ and vertex $y_R$, 
      (\eg a $(k-1)$-face of $\bx{}P_1$ and vertex $\y$),  or
    \item
      $\fF_{R}, R\in \sset_2$ and vertex $y_R$, 
     (\eg a $(k-1)$-face of $\fF_{\{1,2\}}$ and vertex $y_{\{1,2\}}$),  or
    \item 
     $\fF_{S}, S\in \sset_1$ and vertex $y_R$, where $S\subset R\in\sset_2$,
     (\eg a $(k-1)$-face of $\bx{}P_1$ and vertex $y_{\{1,2\}}$
     or vertex $y_{\{1,3\}}$),
  \end{enumerate} 
  for a total of faces equal to:
      \begin{equation*}
            \overbrace{\sum_{R\in\sset_1}\!f_{k-1}(\fF_R)}^{\text{case (a)}}+
            \overbrace{\sum_{R\in\sset_2}\!f_{k-1}(\fF_R)}^{\text{case (b)}}+
            \overbrace{\sum_{R\in\sset_2}
            \sum_{\emptyset \subset S \subset R}\!\!f_{k-1}(\fF_S)}^{\text{case (c)}}= 
             \sum_{R\in\sset_2}\!f_{k-1}(\fF_R)+3\!\!\sum_{R\in\sset_1}\!f_{k-1}(\fF_R).
      \end{equation*}

\item $F$ contains two auxiliary vertices. Then, it can consist of  
  a $(k-2)$-face of $\fF_R, R\in\sset_1$ and 
  vertices $y_R$ and $y_S$, where $S\in\sset_2$ such that $R\subset S$, 
  (\eg a $(k-2)$-face of $\bx{}P_1$ and vertices $\y$ and either
  $y_{\{1,2\}}$ or $y_{\{1,3\}}$), for a total of
  $2\sum_{R\in\sset_1}f_{k-2}(\fF_R)$ faces.
\end{enumerate}

Summing over all previous cases we obtain the following relation, for
all $0\le{}k\le{}d+1$:
\begin{equation}\label{r1}
  f_k(\bx{}\qQ)=f_k(\fF_{[3]})+\sum_{R \in \sset_2}[f_k(\fF_R)+f_{k-1}(\fF_R)]
  +\sum_{R \in \sset_1}[f_k(\fF_R)+3f_{k-1}(\fF_R)+2f_{k-2}(\fF_R)].
\end{equation}
Relation \eqref{r1} also holds for $k\in\{-1,0\}$,
since, by convention, we have set $f_l(\fF_S)=0$ for all $l<-1$ 
and $\emptyset \subset S \subseteq [3]$.

Denote by $\yY$ a generic subset of faces of $\cC$. $\yY$ will either
be a subcomplex of the boundary complex $\bx\cC$ of $\cC$, or one of
the $\fF_R$'s. Let $\delta$ be the dimension of $\yY$. Then
we can define the $h$-vector of $\yY$ as
\begin{equation}\label{eq:hvectorY}
  h_k(\yY)=\sum_{i=0}^{\delta+1}(-1)^{k-i}
  \binom{\delta+1-i}{\delta+1-k}f_{i-1}(\yY).
\end{equation}

Another quantity that will be heavily used in the rest of the paper is
that we call the $m$-order $g$-vector of $\yY$, the $k$-th element of
which is given by the following recursive formula:
\begin{equation}\label{eq:gvec}
  g_k^{(m)}(\yY)=\begin{cases}
    h_k(\yY),&m=0,\\
    g_k^{(m-1)}(\yY)-g_{k-1}^{(m-1)}(\yY),&m>0.\\
  \end{cases}
\end{equation}
Observe that for $m=0$ we get the $h$-vector of $\yY$, while for $m=1$
we get what is typically known as the $g$-vector of $\yY$. Clearly, 
$\mb{g}^{(m)}(\yY)$ is the $m$-order backward finite difference of
$\mb{h}(\yY)$, which suggests the following lemma (see
Section~\ref{app:ds} of Appendix~\ref{app:omitted} for
the proof):
\begin{lemma}\label{lem:gvec_as_sum}
  For any $k,m\ge{}0$, we have:
  \begin{equation}
    g_k^{(m)}(\yY)=\sum_{i=0}^m(-1)^i\binom{m}{i}h_{k-i}(\yY).
  \end{equation}
\end{lemma}

We next define the summation operator $\SSS_k(\cdot;D,\nu)$
whose action on $\yY$ is as follows:
\begin{equation}\label{eq:sumoperator}
\SSS_k(\yY;D,\nu)= \sum_{i=0}^{D+1} (-1)^{k-i} \binom{D+1-i}{D+1-k}
f_{i-1-\nu}(\yY).
\end{equation}
Regarding the action of $\SSS_k(\cdot;D,\nu)$ on $\yY$, it is easy to
verify the following (see Section~\ref{app:ds} of
Appendix~\ref{app:omitted} for the proof):

\begin{lemma}\label{lem:sumoperator_gvec}
Let $\delta$ be the dimension of $\yY$, $\nu\ge{}0$,
$\delta\le{}D$, and $D-\delta-\nu\ge{}0$. Then for any $k\ge{}0$ we
have:
\begin{equation}\label{eq:sumoperator_gvec}
  \SSS_k(\yY;D,\nu)=g_{k-\nu}^{(D-\delta-\nu)}(\yY).
\end{equation}
\end{lemma}

In the following lemma we relate the $h$-vectors of $\fF_R$ and
$\kK_R$ with each other, and with the $h$-vector of $\bx{}\qQ$. The last
among the relations proved in the following lemma can be thought of as
the analogue of the Dehn-Sommerville equations for $\fF_{[3]}$ and $\kK_{[3]}$.

\begin{lemma}\label{lem:hkK-hkQ-dehn-sommervile}
The following relations hold:
\begin{equation}\label{equ:hkKR}
  h_k(\kK_R)=\sum_{\emptyset\subset{}S\subseteq{}R}g_{k}^{(|R|-|S|)}(\fF_S),
  \qquad 0\le{}k\le{}d+|R|-1,
  \qquad \emptyset \subset R\subseteq [3].
\end{equation}
\begin{equation}\label{r4}
  h_k(\bx{}\qQ)=h_k(\fF_{[3]})+\sum_{R\in \sset_2}h_k(\fF_R)
  +\sum_{R\in \sset_1}[h_k(\fF_R)+h_{k-1}(\fF_R)],
  \qquad 0 \leq k \leq d+2.
\end{equation}
\begin{equation}\label{equ:DSW}
  h_{d+2-k}(\fF_{[3]})=h_k(\kK_{[3]}),\qquad 0\le{}k\le{}d+2.
\end{equation}
\end{lemma}
\begin{proof}
  Relation \eqref{equ:hkKR} follows directly from the application of
  the summation operator $\SSS_k(\cdot;d+|R|-2,0)$ to relation
  \eqref{eq:fk_KR}. More precisely, from \eqref{eq:fk_KR} we get,
  for all $0\le{}k\le{}d+|R|-1$,
  \begin{equation}\label{equ:sumop-on-fkKR}
    \SSS_k(\kK_R;d+|R|-2,0)=
    \sum_{\emptyset\subset{}S\subseteq{}R}\SSS_k(\fF_S;d+|R|-2,0).
  \end{equation}
  Relation \eqref{equ:hkKR} now immediately follows by noticing that:
  \begin{smallenum}
  \item[$\bullet$]
    By applying Lemma \ref{lem:sumoperator_gvec} on the right-hand-side
    of \eqref{equ:sumop-on-fkKR}, with $\delta\sub{}d+|R|-2$,
    $D\sub{}d+|R|-2$ and $\nu\sub{}0$, we get
    \begin{equation*}
      \SSS_k(\kK_R;d+|R|-2,0)=g_{k-0}^{((d+|R|-2)-(d+|R|-2)-0)}(\kK_R)=h_k(\kK_R).
    \end{equation*}
  \item[$\bullet$]
    Similarly, by applying Lemma \ref{lem:sumoperator_gvec} on the left
    hand side of \eqref{equ:sumop-on-fkKR}, with
    $\delta\sub{}d+|S|-2$, $D\sub{}d+|R|-2$, $\nu\sub{}0$, we get:
    \begin{equation*}
      \SSS_k(\fF_S;d+|R|-2,0)=g_{k-0}^{((d+|R|-2)-(d+|S|-2)-0)}(\fF_S)
      =g_k^{(|R|-|S|)}(\fF_S).
    \end{equation*}
  \end{smallenum}

To prove \eqref{r4}, we apply the summation operator
$\SSS_k(\cdot;d+1;0)$ to the $(d+1)$-complex $\bx{}\qQ$.
Using relation \eqref{r1}, we get, for all $0\le{}k\le{}d+2$:
\begin{align*}
  \SSS_k(\bx{}\qQ;d+1;0)&=
  \SSS_k(\fF_{[3]};d+1;0)
  +\sum_{R\in \sset_2}[\SSS_k(\fF_R;d+1;0)+\SSS_k(\fF_R;d+1;1)]\\
  &\quad+\sum_{R\in \sset_1}[\SSS_k(\fF_R;d+1;0)+3\SSS_k(\fF_R;d+1;1)
  +2\SSS_k(\fF_R;d+1;2)],
\end{align*}
which, using Lemma~\ref{lem:sumoperator_gvec}, gives, for all
$0\le{}k\le{}d+2$:
\begin{equation*}
g_k^{(0)}(\bx{}\qQ)=
g_k^{(0)}(\fF_{[3]}) + \sum_{R\in \sset_2}[g_k^{(1)}(\fF_R) + g_{k-1}^{(0)}(\fF_R)]+
\sum_{R\in \sset_1}[g_k^{(2)}(\fF_R)+3 g_{k-1}^{(1)}(\fF_R)+2 g_{k-2}^{(0)}(\fF_R)].
\end{equation*}
Relation \eqref{r4} follows by expanding $\mb{g}^{(m)}(\cdot)$, $1\le m\le 2$, 
according to Lemma~\ref{lem:gvec_as_sum}, and gathering common terms.
\bigskip

To prove what we named the Dehn-Sommerville-like equations for
$\fF_{[3]}$ (cf.~\eqref{equ:DSW}),
we replace $k$ by $d+2-k$ in \eqref{r4}, to get, for all $0\le{}k\le{}d+2$:
\begin{equation}\label{r4inv}
  h_{d+2-k}(\bx{}\qQ)=h_{d+2-k}(\fF_{[3]})+\sum_{R\in\sset_2}h_{d+2-k}(\fF_{R})
  +\sum_{R\in \sset_1}[h_{d+2-k}(\fF_R)+h_{d+1-k}(\fF_R)].
\end{equation}
Using the above relation, in conjunction with \eqref{r4}, the
Dehn-Sommerville equations for $\bx{}\qQ$ become:
\begin{align*}
  h_{d+2-k}(\fF_{[3]})&+\sum_{R\in\sset_2}h_{d+2-k}(\fF_{R})
  +\sum_{R\in \sset_1}[h_{d+2-k}(\fF_R)+h_{d+1-k}(\fF_R)]\\
  &=h_k(\fF_{[3]})+\sum_{R\in\sset_2}h_k(\fF_{R})
  +\sum_{R\in \sset_1}[h_k(\fF_R)+h_{k-1}(\fF_R)].
\end{align*}
Using the Dehn-Sommerville equations for $\fF_R$, $R\in\sset_1$,
as well as the Dehn-Sommerville-like equations for $\fF_R$, $R\in\sset_2$
(cf. \cite[rel. (3.10)]{kt-mnfms-12}), we get:
\begin{align*}
  h_{d+2-k}(\fF_{[3]})&+\sum_{R\in\sset_2}[h_{k-1}(\fF_{R})
    +\sum_{\emptyset\subset{}S\subset{}R}g_{k-1}(\fF_S)]
  +\sum_{R\in \sset_1}[h_{k-2}(\fF_R)+h_{k-1}(\fF_R)]\\
  &=h_k(\fF_{[3]})+\sum_{R\in\sset_2}h_k(\fF_{R})
  +\sum_{R\in \sset_1}[h_k(\fF_R)+h_{k-1}(\fF_R)].
\end{align*}
Finally, solving in terms of $h_{d+2-k}(\fF_{[3]})$, we arrive at the
following:
{\allowdisplaybreaks
\begin{align*}
  h_{d+2-k}(\fF_{[3]})
  &=h_k(\fF_{[3]})+\sum_{R\in\sset_2}h_k(\fF_{R})
  +\sum_{R\in \sset_1}[h_k(\fF_R)+h_{k-1}(\fF_R)]\\
  &\quad-\sum_{R\in\sset_2}[h_{k-1}(\fF_{R})
  +\sum_{\emptyset\subset{}S\subset{}R}g_{k-1}(\fF_S)]
  -\sum_{R\in \sset_1}[h_{k-2}(\fF_R)+h_{k-1}(\fF_R)]\\
  &=h_k(\fF_{[3]})+\sum_{R\in\sset_2}h_k(\fF_{R})
  +\sum_{R\in \sset_1}[h_k(\fF_R)+h_{k-1}(\fF_R)]\\
  &\quad-\sum_{R\in\sset_2}h_{k-1}(\fF_{R})
  -2\sum_{R\in\sset_1}g_{k-1}(\fF_S)
  -\sum_{R\in \sset_1}[h_{k-2}(\fF_R)+h_{k-1}(\fF_R)]\\
  &=h_k(\fF_{[3]})+\sum_{R\in\sset_2}[h_k(\fF_{R})-h_{k-1}(\fF_{R})]\\
  &\quad+\sum_{R\in\sset_1}[h_k(\fF_R)+h_{k-1}(\fF_R)-2g_{k-1}(\fF_S)-
  h_{k-2}(\fF_R)-h_{k-1}(\fF_R)]\\
  &=h_k(\fF_{[3]})+\sum_{R\in\sset_2}g_k(\fF_{R})
  +\sum_{R\in\sset_1}[h_k(\fF_R)-2h_{k-1}(\fF_R)+h_{k-2}(\fF_R)]\\
  &=h_k(\fF_{[3]})+\sum_{R\in\sset_2}g_k(\fF_{R})+\sum_{R\in\sset_1}g^{(2)}_k(\fF_R)\\
  &=h_k(\kK_{[3]}),
\end{align*}}
where for the last equality we used relation~\eqref{equ:hkKR} for
$R\equiv{}[3]$.
\end{proof}

\section{Recurrence relations}
\label{sec:recur}

Recall that we denote by $\VV$ the vertex set of $\bx{}\qQ$ and by $\VV_i$ the 
(Cayley embedding of the) vertex set of  $\bx{}P_i$, $1\le i\le 3$.
Let $\yY / v$ denote the link of vertex $v$ of $\yY$ in the simplicial complex $\yY$.
McMullen \cite{m-mnfcp-70} showed  that for any $d$-dimensional polytope $P$ the following relation holds:
\begin{equation}\label{equ:links}
(k+1)h_{k+1}(\bx{}P) + (d-k)h_k(\bx{}P)= \sum_{v\in \text{vert}(\bx{}P)} h_k (\bx{}P / v), \qquad 0 \leq k \leq d-1.
\end{equation}

Applying relation \eqref{equ:links} to the $(d+2)$-dimensional
polytope $\qQ$, we have, for all $0\le{}k\le{}d+1$:
\begin{equation}\label{equ:linksQ}
  (k+1)h_{k+1}(\bx{}\qQ)+ (d+2-k)h_k(\bx{}\qQ)
  =\sum_{v\in \VV} h_k(\bx{}\qQ/v)
  =\sum_{v\in \VV_{[3]}} h_k(\bx{}\qQ / v)
  +\sum_{\emptyset{}\subset{}R\subset[3]} h_k(\bx{}\qQ / y_R),
\end{equation}
where we used the fact that $\VV$ is the disjoint union of the vertex sets
$\VV_{[3]}=\VV_1\cup\VV_2\cup\VV_3$ and
$\{y_R\mid{}\emptyset{}\subset{}R\subset[3]\}$.

\begin{lemma}\label{lem:h-vectors-Q-links}
The $h$-vectors of the complexes $\bx{\qQ}/v,\,v\in \VV_i,\,i=1,2,3$, $\bx{\qQ}/y_R,~R\in \sset_1$, 
and $\bx{\qQ}/y_R,~R\in \sset_2$
are given by the following relations:
\begin{equation}
  \label{equ:hk-Q-link-v}
  h_k(\bx{}Q / v) = h_k(\kK_{[3]} / v)
  + \sum_{\{i\} \subseteq R \subset[3]} h_{k-1}(\kK_{R} / v)
  + h_{k-2}(\kK_{\{i\}} / v),
  \quad v\in \VV_i,\quad i\in[3],
\end{equation}
\begin{equation}
  \label{equ:hk-Q-link-S1}
  h_k(\bx{}\qQ / y_R) =  h_k(\fF_R) + h_{k-1}(\fF_R), \qquad R\in\sset_1,
\end{equation}
\begin{equation} 
  \label{equ:hk-Q-link-S2}
  h_k(\bx{}\qQ / y_R) = \sum_{\emptyset \subset S \subseteq{} R}h_{k}(\fF_S),
  \qquad R\in\sset_2.
\end{equation}
\end{lemma}
\begin{proof} 
We start by proving relation \eqref{equ:hk-Q-link-v}.
Without loss of generality we assume that $v\in \VV_1$; the cases
$v\in\VV_2$ and $v\in\VV_3$ are entirely analogous.

Let $F$ be a $k$-face of $\bx{}Q/v$. We have the following cases
depending on the number of additional points $y_R$,
$\emptyset\subset{}R\subset [3]$, that $F$ contains: 
\begin{enumerate}
\item $F$ does contain any additional points. Then, it is a $k$-face of $\kK_{[3]}$.

\item $F$ contains one additional point. Then, it can consist of  a $(k-1)$-face of:
  \begin{enumerate}
   \item 
      $\kK_{\{1\}} (\equiv \bx{}P_1)$ and point $\y$,  or
   \item 
      $\kK_{\{1,2\}}$, and point $y_{\{1,2\}}$, or
   \item 
      $\kK_{\{1,3\}}$, and point $y_{\{1,3\}}$.
  \end{enumerate} 

\item $F$ contains two additional points. Then, it can consist of  a $(k-2)$-face of $\kK_{\{1\}}$ and 
points $\y$ and $y_{\{1,2\}}$, or points  $\y$ and $y_{\{1,3\}}$.
\end{enumerate}

Summing over all previous cases we obtain the following relation:
\begin{equation}
 \label{equ:k-faces-Q-link-v}
 f_k(\bx{}Q / v) = \overbrace{f_k(\kK_{[3]} / v)}^{\text{case (i)}}
 + \overbrace{\sum_{\{1\} \subseteq R \subset [3]} f_{k-1}(\kK_{R} /
   v)}^{\text{case (ii)}}
 + \overbrace{2f_{k-2}(\kK_{\{1\}} / v)}^{\text{case (iii)}}, \quad v\in \VV_1. 
\end{equation}
We apply the summation operator $\SSS_k(\cdot;d,0)$ to the $d$-complex $\bx{}Q/v$ and obtain:
\begin{equation*}
 g^{(0)}_k(\bx{}Q / v) = g^{(0)}_k(\kK_{[3]} / v)
 + \sum_{\{1\} \subseteq R \subset [3]} g^{(2-|R|)}_{k-1}(\kK_{R} / v)
 + 2g^{(0)}_{k-2}(\kK_{\{1\}} / v), 
\end{equation*}
which finally gives, for any $v\in\VV_1$:
\begin{align*}
  h_k(\bx{}Q/v)
  &= h_k(\kK_{[3]}/v) +\bigg(g_{k-1}(\kK_{\{1\}}/v)
  + \!\!\sum_{\{1\}\subset{}R\subset[3]}\!\! h_{k-1}(\kK_{R}/v)\bigg)
  + 2h_{k-2}(\kK_{\{1\}}/v)\\ 
  &= h_k(\kK_{[3]}/v) + h_{k-1}(\kK_{\{1\}}/v)-h_{k-2}(\kK_{\{1\}}/v)
  + \!\!\sum_{\{1\}\subset{}R\subset[3]}\!\! h_{k-1}(\kK_{R}/v)
  + 2h_{k-2}(\kK_{\{1\}}/v)\\ 
  &= h_k(\kK_{[3]}/v)
  + \sum_{\{1\}\subseteq{}R\subset [3]}h_{k-1}(\kK_{R}/v)
  + h_{k-2}(\kK_{\{1\}}/v).
\end{align*}

To prove \eqref{equ:hk-Q-link-S1} consider a $k$-face of $\bx{}\qQ /
y_R,~R\in\sset_1$. Such a face is either a $k$-face of $\fF_R$,
or consists of a $(k-1)$-face of $\fF_R$ and point $y_S$ for any $ S
\in \sset_2$ such that $S\supset R$. Note that there exactly two
such points $y_S$. Hence:
\begin{equation}\label{r5right2}   
  f_k(\bx{}\qQ / y_R) = f_k(\fF_R) + 2f_{k-1}(\fF_R),\qquad R\in\sset_1 .
\end{equation}
Applying the summation operator $\SSS_k(\cdot;d,0)$
to the simplicial $d$-complex $\bx{}\qQ / y_R,~R\in\sset_1$,
and using relation \eqref{r5right2} and
Lemma~\ref{lem:sumoperator_gvec}, we get, for any $R\in\sset_1$:
\begin{equation*}
  \begin{array}{lclcl}
  h_k(\bx{}\qQ/y_R)&=&g_k^{(0)}(\bx{}\qQ / y_R)
  &=&\SSS_k(\bx{}\qQ / y_R;d,0)\\[2pt]
  &=&\SSS_k(\fF_R;d,0) + 2\SSS_k(\fF_R;d,1)
  &=&g_k^{(1)}(\fF_R) + 2g_{k-1}^{(0)}(\fF_R)\\[2pt]
  &=&h_k(\fF_R) - h_{k-1}(\fF_R) + 2h_{k-1}(\fF_R)
  &=&h_k(\fF_R) + h_{k-1}(\fF_R).
  \end{array}
\end{equation*}

To prove \eqref{equ:hk-Q-link-S2} consider a $k$-face of $\bx{}\qQ /
y_R,~R\in\sset_2$. This is either a $k$-face of $\fF_S,$ for any
$\emptyset\subset S\subseteq R$, or consists of a $(k-1)$-face of
$\fF_S$ and point $y_S$ for any $\emptyset\subset S\subset R$. Hence,
for any $R\in\sset_2$, we have:
\begin{equation}  \label{r5right3}
  f_k(\fF_R / y_R)
  =\sum_{\emptyset\subset S\subseteq R}f_k(\fF_S)
  +\sum_{\emptyset\subset S\subset R}f_{k-1}(\fF_S)
  =f_k(\fF_R) + \sum_{\emptyset\subset S\subset R}[f_k(\fF_S) + f_{k-1}(\fF_S)].
\end{equation}
Applying the summation operator $\SSS_k(\cdot;d,0)$ 
to the $d$-dimensional complex $\bx{}\qQ / y_R$, $R\in\sset_2$, and
using relation \eqref{r5right3}, along with
Lemma~\ref{lem:sumoperator_gvec}, we get, for any $R\in\sset_2$:
\begin{align*}
  h_k(\bx{}\qQ / y_R)=\SSS_k(\bx{}\qQ/y_R;d,0)
  &=\SSS_k(\fF_R;d,0)
  + \sum_{\emptyset\subset S\subset R}[\SSS_k(\fF_S;d,0) + \SSS_k(\fF_S;d,1)]\\[2pt]
  &= g_k^{(0)}(\fF_R)
  + \sum_{\emptyset\subset S\subset R}[g_k^{(1)}(\fF_S) + g_{k-1}^{(0)}(\fF_S)]\\[2pt]
  &= h_k(\fF_R) + \sum_{\emptyset\subset S\subset R}h_k(\fF_S)\\[2pt]
  &= \sum_{\emptyset\subset S\subseteq{}R}h_k(\fF_S).\qedhere
\end{align*}
\end{proof}

The following two lemmas are essential in the proof of the upcoming
recurrence relation in Lemma  \ref{lem:hkWrecur}.
\begin{lemma}
\label{lem:recur-relation-F3-wrt-K}
The following relation holds, for all $0\le{}k\le{}d+1$:
\begin{equation}\label{equ:recur-relation-F3-wrt-K}
(k+1) h_{k+1}(\fF_{[3]})+ (d+2-k)  h_k(\fF_{[3]}) =
\sum_{\emptyset \subset R \subseteq [3]} (-1)^{3-|R|}
\sum_{v\in \VV_R} g_k^{(3-|R|)}(\kK_{R} / v).  
\end{equation}
\end{lemma}
\begin{proof}[Sketch of proof]
  The complete proof may be found in Section~\ref{app:recur}
  of Appendix~\ref{app:omitted}.
  Our starting point is relation \eqref{equ:linksQ}. 
  We first substitute $ h_k( \bx{}Q) $ and $ h_{k+1} (\bx{}Q)$ on 
  the left-hand side of \eqref{equ:linksQ}
  with their relevant expressions from \eqref{r4}.  
  We then group the terms so that we get a sum of: 
  \begin{enumerate} 
  \item  the left-hand side of \eqref{equ:recur-relation-F3-wrt-K}, 
  \item  $(k+1) h_{k+1}(\fF_R) + (d+1-k) h_k(\fF_R), \, R \in \sset_2$   
  \item  $(k+1) h_{k+1}(\fF_R) + (d-k) h_k(\fF_R) $ and
    $k  h_{k}(\fF_R) + (d-k-1) h_{k-1}(\fF_R)$ with $ R \in \sset_1$ 
  \item  additional terms.   
  \end{enumerate} 
  As will be described below, the intuition behind this grouping is to
  substitute the terms in (ii) and (iii) by sums involving quantities
  of the form $g_k^{(m)}(\kK_S/v)$. These quantities will be grouped
  with the terms obtained from a similar expansion of the term
  $h_k(\bx\qQ/v)$ appearing in the right-hand side of
  \eqref{equ:linksQ}, yielding the right-hand side of
  \eqref{equ:recur-relation-F3-wrt-K}.

  In the proof of \cite[Lemma 3.2]{kt-mnfms-12}, the sum in item (ii) 
  above is shown to be equal\footnote{The expression in \cite{kt-mnfms-12}
  is written differently; it is equivalent, however, to the expression
  stated here.} to
  $\sum_{i\in R}\sum_{v\in\VV_i}[h_k(\kK_R/v)-g_k(\kK_{\{i\}}/v)]$.
  For (iii) we use \eqref{equ:links} combined with the fact
  that for any $R\in\sset_1$, $\fF_R\equiv\bx{}P_R$.
  On the right-hand side of  \eqref{equ:linksQ}  we substitute
  $h_k(\bx{}Q/v)$ and $h_k(\bx{}Q/y_R)$
  using the relations  in Lemma \ref{lem:h-vectors-Q-links}.
  Finally, we equate our expansions of the left- and right-hand side
  of \eqref{equ:linksQ} and notice that the terms in (iv) and the
  expressions for $h_k(\bx{}Q/y_R)$ cancel-out.
  Recalling that $g_k=h_k-h_{k-1}$ and $g_k^{(2)}=h_k-2h_{k-1}+h_{k-2}$, 
  we appropriately regroup the remaining  terms to obtain the desired
  expression.
\end{proof}

\begin{lemma}\label{lem:LinksNoLinks}
  The following relation holds, for all $0\le{}k\le{}d+1$:
  \begin{equation}\label{equ:LinksNoLinks}
    \sum_{\emptyset \subset R \subseteq [3]} (-1)^{3-|R|}\,
    \sum_{v\in \VV_R} g_k^{(3-|R|)}(\kK_{R} / v)
    \leq
    \sum_{\emptyset \subset R \subseteq [3]} (-1)^{3-|R|}\,
    \sum_{v\in \VV_R} g_k^{(3-|R|)}(\kK_{R} ).
  \end{equation}
\end{lemma}
\begin{proof}
  Let us first observe that, by rearranging terms, we can rewrite
  relation \eqref{equ:LinksNoLinks} as follows:
  \begin{equation}\label{equ:LinksNoLinks-alt}
    \sum_{i=1}^{3}\sum_{v\in\VV_i}
    \sum_{\{i\}\subseteq{}R\subseteq{}[3]}(-1)^{3-|R|}\,g_k^{(3-|R|)}(\kK_R/v)
    \le
    \sum_{i=1}^{3}\sum_{v\in\VV_i}
    \sum_{\{i\}\subseteq{}R\subseteq{}[3]}(-1)^{3-|R|}\,g_k^{(3-|R|)}(\kK_R).
  \end{equation}
  Clearly, to show that relation \eqref{equ:LinksNoLinks-alt} holds, it
  suffices to prove that:
  \begin{equation}\label{equ:LinksNoLinks-sufficient}
    \sum_{\{i\}\subseteq{}R\subseteq{}[3]}(-1)^{3-|R|}\,g_k^{(3-|R|)}(\kK_R/v)
    \le
    \sum_{\{i\}\subseteq{}R\subseteq{}[3]}(-1)^{3-|R|}\,g_k^{(3-|R|)}(\kK_R),
    \quad v\in\VV_i,\quad i\in[3].
  \end{equation}
  In the rest of the proof we shall prove relation
  \eqref{equ:LinksNoLinks-sufficient} for $i=1$ and for any
  $v\in\VV_1$. The cases $i=2$ and $i=3$ are entirely similar.

  \newcommand{\sS}{\mathcal{S}}
  \newcommand{\qQp}{\bx\qQ'}
  \newcommand{\zZp}{\zZ'}

  \begin{figure*}[t]
    \begin{center}
      \includegraphics[width=0.475\textwidth,page=3]{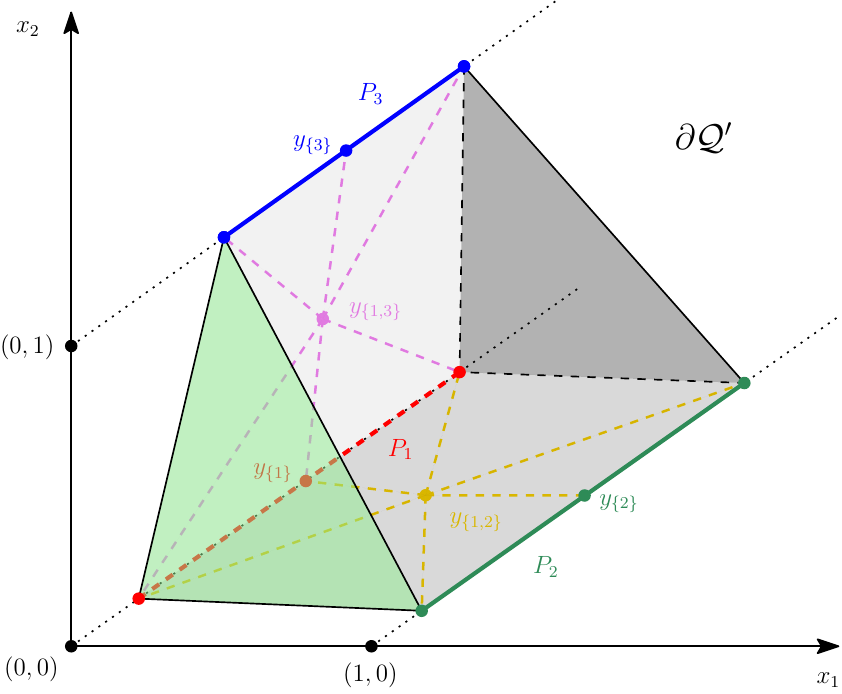}\hfill
      \includegraphics[width=0.475\textwidth,page=4]{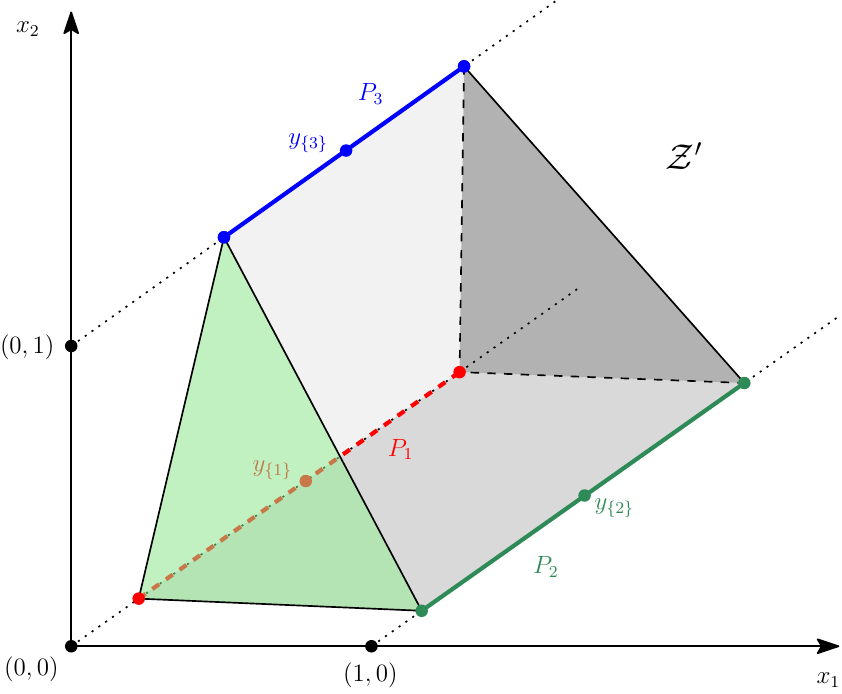}
    \end{center}
    \caption{Left: the $(d+1)$-complex $\qQp$ that we get from $\bx\qQ$
      be removing all faces incident to $y_{\{2,3\}}$. Right: the
      $(d+1)$-complex $\zZp$ that we get from the Cayley polytope
      $\cC_{[3]}$ of $P_1,P_2$ and $P_3$, after we: (i) have performed
      stellar subdivisions using the vertices $y_{\{1\}},y_{\{2\}}$
      and $y_{\{3\}}$ (which yields the $(d+1)$-polytope $\zZ$), and
      (ii) have removed the facet $\qQ_{\{2,3\}}$ from $\zZ$.}
    \label{fig:qQp+Z}
  \end{figure*}

  Fix a vertex $v\in\VV_1$. Let $\qQp$ be the polytopal
  $(d+1)$-complex that we get by removing from $\bx\qQ$ the faces that
  are incident to $y_{\{2,3\}}$ (see Fig.~\ref{fig:qQp+Z}(left)).
  It is straightforward to see that:
  (1) the stars of $v$ in $\qQ$ and $\qQp$ coincide (the faces
  incident to $y_{\{2,3\}}$ contain vertices from
  $\VV_{\{2,3\}}\cup\{y_{\{2\}},y_{\{3\}}\}$ only), and 
  (2) $\qQp$ is shellable.
  To verify the latter consider a shelling $\Sl{\bx\qQ}$ of
  $\bx\qQ$ that shells the star of $y_{\{2,3\}}$ in $\bx\qQ$ last; the
  shelling order that we get by removing from $\Sl{\bx\qQ}$ the facets
  that are incident to $y_{\{2,3\}}$ is clearly a shelling order for
  $\qQp$. Let $\sS_R$, $R\in\{\{1,2\},\{1,3\}\}$, be the star of $y_R$
  in $\qQp$
  (which actually coincides with the star of $y_R$ in $\bx\qQ$).
  Let $\xX$ denote the set of faces of $\qQp$ that are either faces in
  $\sS_{\{1,2\}}$ or faces in $\sS_{\{1,3\}}$, and let $\gG$ denote the
  set of faces of $\qQp$ that are either faces in $\fF_{[3]}$ or faces
  in $\fF_{\{2,3\}}$. Notice that the sets $\xX$ and $\gG$ form a
  disjoint union of the faces in $\qQp$, which implies that:
  \begin{equation}\label{equ:Qp-fk-decomp}
    f_k(\qQp)=f_k(\xX)+f_k(\gG),\qquad -1\le{}k\le{}d+1.
  \end{equation}
  Notice that $\xX$ is a $(d+1)$-complex, whereas $\gG$ is a set of
  faces with maximal dimension $d+1$.
  By applying the summation operator $\SSS_k(\cdot;d+1,0)$ to
  \eqref{equ:Qp-fk-decomp}, we immediately get the
  corresponding $h$-vector relation:
  \begin{equation}\label{equ:Qp-hk-decomp}
    h_k(\qQp)=h_k(\xX)+h_k(\gG),\qquad 0\le{}k\le{}d+2.
  \end{equation}

  We claim that there exists a specific shelling $\Sl{\qQp}$ of
  $\qQp$, which actually is an initial segment of a shelling of 
  $\bx{}\qQ$ that shells the star of $y_{\{2,3\}}$ last, 
  with the property that the corresponding shelling order
  has the facets in $\xX$ before the facets in $\gG$. We will postpone
  the proof of this claim, and we will assume for now that the claim
  holds true.
  Consider the specific shelling of $\qQp$ just mentioned, and notice
  that the facets in $\gG$ are actually the facets in
  $\fF_{[3]}$. The existence of this particular shelling $\Sl{\qQp}$
  also implies that $\xX$ is shellable, and the shelling $\Sl{\xX}$ of
  $\xX$ induced by $\Sl{\qQp}$ coincides with $\Sl{\qQp}$ as long as
  it visits the facets of $\xX$.
  As a result of this, and as long as we shell $\xX$, we get a
  contribution of $+1$ to both $h_k(\qQp)$ and $h_k(\xX)$ for every
  restriction of $\Sl{\qQp}$ of size $k$. After the shelling $\Sl{\qQp}$
  has left $\xX$, a restriction of size $k$ for $\Sl{\qQp}$ contributes
  $+1$ to $h_k(\qQp)$, does not contribute to $h_k(\xX)$ ($\xX$ has
  already been fully constructed), and, thus, by relation
  \eqref{equ:Qp-hk-decomp}, contributes $+1$ to $h_k(\gG)$.
  In other words, for this particular shelling
  $\Sl{\qQp}$ of $\qQp$, $h_k(\gG)$ counts the number of restrictions
  of size $k$ that correspond to the facets of $\qQp$ that are also
  facets of $\gG$ (and, of course, of $\fF_{[3]}$).

  The same argumentation can be applied to the links of vertices
  $v\in\VV_1$: $\qQp/v$ can be seen as the disjoint union of the sets
  $\xX/v$ and $\gG/v$, while the particular shelling $\Sl{\qQp}$ of $\qQp$
  that shells $\xX$ first, induces a particular shelling $\Sl{\qQp/v}$
  for $\qQp/v$ that shells the facets of $\bx{\qQ'/v}$ in $\xX/v$
  first. From these
  observations we immediately arrive at the following $h$-vector
  relation for $\qQp/v$, $\xX/v$ and $\gG/v$:
  \begin{equation}\label{equ:Qp-linkv-hk-decomp}
    h_k(\qQp/v)=h_k(\xX/v)+h_k(\gG/v),\qquad 0\le{}k\le{}d+1,
  \end{equation}
  from which we argue, as above, that $h_k(\gG/v)$ counts the number
  of restrictions of size $k$ for $\Sl{\qQp/v}$ that correspond to the
  facets of $\qQp/v$ that are also facets of $\gG/v$ (or
  $\fF_{[3]}/v$).

  Let us now consider the dual graph $G^\Delta(\bx\qQ)$ of $\bx\qQ$,
  oriented according to the shelling $\Sl{\bx\qQ}$, as well as the
  dual graph $G^\Delta(\bx\qQ/v)$ of $\bx\qQ/v$, also oriented
  according to the shelling $\Sl{\bx\qQ/v}$.
  We will denote by $\VV^\Delta(\yY)$ the subset of vertices of
  $G^\Delta(\bx\qQ)$ that are the duals of the facets in $\bx\qQ$ that
  belong to $\yY$, where $\yY$ stands for a subset of the set of faces
  of $\bx\qQ$.

  Since $\Sl{\bx\qQ/v}$ is induced from $\Sl{\bx\qQ}$,
  $G^\Delta(\bx\qQ/v)$ is isomorphic to the subgraph of
  $G^\Delta(\bx\qQ)$ defined over $\VV^\Delta(\str(v,\bx\qQ))$.
  Moreover, $h_k(\bx\qQ)$ counts the number of vertices of
  $\VV^\Delta(\bx\qQ)$ with in-degree equal to $k$, while $h_k(\gG)$
  counts the number of vertices of $\VV^\Delta(\gG)$ of in-degree $k$
  in $G^\Delta(\bx\qQ)$ (for the particular shelling $\Sl{\bx\qQ}$ of
  $\bx\qQ$ that we have chosen).
  Consequently, $h_k(\gG)$ counts the number of vertices of
  $\VV^\Delta(\gG)$ of in-degree $k$ in $G^\Delta(\bx\qQ)$;
  in an analogous manner, we can conclude that $h_k(\gG/v)$
  counts the number of vertices of $\VV^\Delta(\str(v,\gG))$ with
  in-degree $k$ in $G^\Delta(\bx\qQ/v)$.
  Since, however, $G^\Delta(\bx\qQ/v)$ is the subgraph of
  $G^\Delta(\bx\qQ)$ that corresponds to the face $v^{\Delta}$ of
  $G^{\Delta}(\bx\qQ)$, the number of vertices of
  $\VV^\Delta(\str(v,\gG))$ with in-degree $k$ 
  cannot exceed the number of vertices of $\VV^{\Delta}(\gG)$ with
  in-degree $k$. Hence,
  \begin{equation}\label{equ:LinkNoLink-G}
    h_k(\gG/v)\le{}h_k(\gG),\qquad 0\le{}k\le{}d+2.
  \end{equation}
  
  On the other hand, recall that $\gG$ is the disjoint union of
  $\fF_{[3]}$ and $\fF_{\{2,3\}}$. Using expressions
  \eqref{eq:fk_KR}, in conjunction with the fact that
  $\fF_S\equiv\kK_S$ for $S\in\sset_1$, we have, for all
  $-1\le{}k\le{}d+1$:
  \begin{align}
    f_k(\gG)&=f_k(\fF_{[3]})+f_k(\fF_{\{2,3\}})\notag\\
    &=\overbrace{f_k(\kK_{[3]})-\sum_{R\in\sset_2}f_k(\fF_R)-\sum_{R\in\sset_1}f_k(\kK_R)}^{f_k(\fF_{[3]})}
    +\overbrace{f_k(\kK_{\{2,3\}})-\sum_{i\in\{2,3\}}f_k(\kK_{\{i\}})}^{f_k(\fF_{\{2,3\}}}\notag\\
    &=f_k(\kK_{[3]})
    -\sum_{R\in\sset_2}\left[f_k(\kK_R)-\sum_{i\in{}R}f_k(\kK_{\{i\}})\right]
    -\sum_{R\in\sset_1}f_k(\kK_R)
    +f_k(\kK_{\{2,3\}})-\sum_{i\in\{2,3\}}f_k(\kK_{\{i\}})\notag\\
    &=f_k(\kK_{[3]})-\sum_{R\in\sset_2}f_k(\kK_R)
    +2\sum_{R\in\sset_1}f_k(\kK_R)
    -\sum_{R\in\sset_1}f_k(\kK_R)
    +f_k(\kK_{\{2,3\}})-\sum_{i\in\{2,3\}}f_k(\kK_{\{i\}})\notag\\
    &=f_k(\kK_{[3]})-\sum_{\{1\}\subset{}R\subset[3]}f_k(\kK_R)
    +f_k(\kK_{\{1\}})\notag\\
    &=\sum_{\{1\}\subseteq{}R\subseteq[3]}(-1)^{3-|R|}\,f_k(\kK_R).\label{eq:fk-G-KR}
  \end{align}
  By a similar argument, we can arrive that the following expression
  for $f_k(\gG/v)$:
  \begin{equation}\label{eq:fk-Gv-KRv}
     f_k(\gG/v)=\sum_{\{1\}\subseteq{}R\subseteq[3]}(-1)^{3-|R|}\,f_k(\kK_R/v),
     \qquad -1\le{}k\le{}d.
  \end{equation}
  By applying the summation operators $\SSS_k(\cdot;d+1,0)$ and
  $\SSS_k(\cdot;d,0)$ to relations \eqref{eq:fk-G-KR} and
  \eqref{eq:fk-Gv-KRv}, respectively, we get the corresponding
  $h$-vector relations:
  \begin{equation}\label{eq:hk-G-KR+links}
    \begin{aligned}
      h_k(\gG)&=
      \sum_{\{1\}\subseteq{}R\subseteq[3]}(-1)^{3-|R|}\,g_k^{(3-|R|)}(\kK_R),
      \qquad 0\le{}k\le{}d+2,\\
      h_k(\gG/v)&=
      \sum_{\{1\}\subseteq{}R\subseteq[3]}(-1)^{3-|R|}\,g_k^{(3-|R|)}(\kK_R/v),
      \qquad 0\le{}k\le{}d+1.
    \end{aligned}
  \end{equation}
  Relation \eqref{equ:LinksNoLinks-sufficient} (for $i=1$) follows by
  substituting the expressions for $h_k(\gG)$ and $h_k(\gG/v)$ from
  \eqref{eq:hk-G-KR+links} in \eqref{equ:LinkNoLink-G}.

  To finish our proof, it remains to establish our claim that
  there exists a specific shelling $\Sl{\qQp}$ of
  $\qQp$ with the property that the facets of $\xX$ appear in the shelling
  before the facets of $\gG$.
  Let us start with some definitions:
  we denote by $\zZ$ the $(d+1)$-complex we
  get by performing the stellar subdivisions on $\cC_{[3]}$ using the
  vertices $y_R$, $R\in\sset_1$
  (see also Fig.~\ref{fig:qQp+Z}(right)),
  and by $\qQ_R$, $R\in\sset_2$ the
  $(d+1)$-complex that we get by performing stellar subdivisions on
  the non-simplicial proper faces of $\cC_R$, namely the faces
  $\cC_S$, $\emptyset{}\subset{}S\subset{}R$. Notice that $\qQ_R$,
  $R\in\sset_2$, is nothing but a facet of $\zZ$, while $\bx\qQ_R$ is
  actually the link of $y_R$ in $\bx\qQ$. In fact, we can
  separate the facets of $\zZ$ in two categories; they are either
  (1) facets of the form $\qQ_R$, $R\in\sset_2$, which are
  non-simplicial, or 
  (2) facets in $\gG$ (or $\fF_{[3]}$), which are simplicial.
  Moreover, notice that $\str(y_R,\zZ)$, $R\in\sset_1$, consists of
  the faces belonging to the two facets $\qQ_S$,
  $R\subset{}S\subset[3]$ of $\zZ$.
  Since stellar subdivisions produce polytopal complexes \cite{EwSh74},
  $\zZ$ is polytopal and, thus, shellable. In fact, there
  exists a particular (line) shelling $\Sl{\zZ}$ of $\zZ$ in which the
  facets of $\str(y_{\{1\}},\zZ)$ appear first, while
  $\qQ_{\{2,3\}}$ is the last facet in $\Sl{\zZ}$. More precisely,
  for this particular shelling of $\zZ$, the two facets
  $\qQ_{\{1,2\}}$ and $\qQ_{\{1,3\}}$ appear first, followed by the
  facets in $\gG$, which, in turn, are followed by the facet
  $\qQ_{\{2,3\}}$.

  Let us call $\zZp$ the $(d+1)$-complex we get by removing
  $\qQ_{\{2,3\}}$ from $\zZ$. The complex $\zZp$ is shellable (it
  follows from the fact that $\Sl{\zZ}$ has $\qQ_{\{2,3\}}$ as its
  last facet), while the particular line shelling $\Sl{\zZ}$ of
  $\zZ$ described above, yields a shelling $\Sl{\zZp}$ for $\zZp$ in
  which the facets $\qQ_{\{1,2\}}$ and $\qQ_{\{1,3\}}$ appear first,
  followed by the facets in $\gG$.
  Notice that if we perform stellar subdivisions on the
  two non-simplicial facets $\qQ_{\{1,2\}}$ and $\qQ_{\{1,3\}}$ of
  $\zZp$ (using the vertices $y_{\{1,2\}}$ and $y_{\{1,3\}}$), we
  arrive at the simplicial $(d+1)$-complex $\qQp$ described earlier.
  %
  Furthermore, from the particular shelling $\Sl{\zZp}$ of $\zZp$
  described above, we may obtain the sought-for shelling for $\qQp$
  that shells $\xX$ first and $\gG$ last. To see this, notice that
  given any shelling order for $\bx{}P_i$, $i=1,2,3$, we may construct
  a shelling for $\qQ_R$, $R\in\{\{1,2\},\{1,3\}\}$, that:
  (1) shells $\st(y_{\{1\}},\qQ_R)$ first,
  (2) shells $\st(y_{R\sm\{i\}},\qQ_R)$ last, and
  (3) the shelling order of the facets in both stars is the order
  implied by the shellings of the boundary complexes $\bx{}P_i$ and
  $\bx{}P_{R\sm\{i\}}$.
  This implies that if we choose shelling orders for $\bx\qQ_{\{1,2\}}$ and
  $\bx\qQ_{\{1,3\}}$ that respect a common shelling order for
  $\bx{}P_1$, we can replace the facets $\qQ_{\{1,2\}}$ and
  $\qQ_{\{1,3\}}$ in $\Sl{\zZp}$ by the facets in
  $\str(y_{\{1,2\}},\qQp)$ and $\str(y_{\{1,3\}},\qQp)$, respectively,
  (the shelling orders of $\bx\qQ_{\{1,2\}}$ and $\bx\qQ_{\{1,3\}}$ are
  ``inherited'' in the shelling orders for $\str(y_{\{1,2\}},\qQp)$ and
  $\str(y_{\{1,3\}},\qQp)$) and arrive at a shelling order for
  $\qQp$ with the desired property.
\end{proof}

\let\qQp\undefined
\let\zZp\undefined

Using inequality \eqref{equ:LinksNoLinks} in Lemma
\ref{lem:LinksNoLinks}, we arrive at the following recurrence relation
for the elements of $\mb{h}(\fF_{[3]})$; its proof may be found in Section
\ref{app:recur} in Appendix \ref{app:omitted}.

\begin{lemma}\label{lem:hkWrecur}
  For all $0\le{}k\le{}d+1$, we have:
  \begin{equation}\label{r8}
    h_{k+1}(\fF_{[3]}) \leq \frac{n_{[3]}-d-2+k}{k+1} h_k(\fF_{[3]})
    + \sum_{i=1}^{3}\frac{n_i}{k+1}g_k(\fF_{[3]\sm\{i\}}).
  \end{equation}
\end{lemma}
\begin{proof}[Sketch of proof]
Using Lemma~\ref{lem:LinksNoLinks}, we can bound the left hand side of relation
\eqref{equ:recur-relation-F3-wrt-K} by the right hand side of relation
\eqref{equ:LinksNoLinks}, which involves $g$-vectors, or various
orders, of the complexes $\kK_R$, where $\emptyset\subset R\subseteq [3]$.
These can be substituted by their equal values from relation \eqref{equ:hkKR}
with $R=[3]$ and for all $R\in\sset_2$.
This gives an inequality involving $h$-vectors and $g$-vectors of $\fF_{[3]}$ 
and $\fF_R,\,R\in\sset_2$, which simplifies to relation \eqref{r8}.
\end{proof}
\section{Upper bounds}
\label{sec:ub}
In this section we establish upper bounds for the number of
$(k+2)$-faces of $\fF_{[3]}$, $0\le{}k\le{}d-1$, which immediately
yield upper bounds for the number of $k$-faces of $P_1+P_2+P_3$. Our
starting point is the recurrence relation \eqref{r8}.
We shall first prove a few lemmas that establish bounds for
the $g$-vector of $\fF_R$, $R\in\sset_2$, and the $h$-vectors of $\fF_{[3]}$
and $\kK_{[3]}$.

\begin{lemma}\label{lem:gkFbound}
  Let $R$ be a non-empty subset of $[3]$ of cardinality $2$.
  Then, for all $0\le{}k\le{}d+2$, we have:
  \begin{equation}\label{equ:gkFbound}
    g_k(\fF_{R})\le{}
    \sum_{\emptyset\subset{}S\subseteq{}R}(-1)^{|S|}\binom{n_S-d-3+k}{k}.
  \end{equation}
  Equality holds for some $k$, where $0\le{}k\le{}\lexp{d+1}$, if and only if
  $f_{l-1}(\fF_R)=\sum_{\emptyset\subset{}S\subseteq{}R}
  (-1)^{|S|}\binom{n_S}{l}$, for all $0\le{}l\le{}k$.
\end{lemma}

\begin{proof}
  The bound clearly holds, as equality, for $k=0$.
  For $k\ge{}1$, from \cite[Lemma 3.2]{kt-mnfms-12} we have:
  \begin{equation}\label{equ:hkFrec}
    h_k(\fF_R)\le{}\tfrac{n_R-d-2+k}{k}h_{k-1}(\fF_R)
    +\sum_{\emptyset\subset{}S\subset{}R}\tfrac{n_{R\sm{}S}}{k}g_{k-1}(\fF_S).
  \end{equation}
  Subtracting $h_{k-1}(\fF_R)$ from both sides of \eqref{equ:hkFrec} we
  get:
  \begin{equation}\label{equ:gkFrecur}
    g_k(\fF_R)\le{}\tfrac{n_R-d-2}{k}h_{k-1}(\fF_R)
    +\sum_{\emptyset\subset{}S\subset{}R}\tfrac{n_{R\sm{}S}}{k}g_{k-1}(\fF_S).
  \end{equation}
  Using now the upper bounds for $h_{k-1}(\fF_R)$, $g_{k-1}(\fF_S)$,
  $\emptyset\subset{}S\subset{}R$, and noting that
  $n_R-d-2\ge{}2(d+1)-d-2=d>0$, we deduce, for any $k\ge{}1$:
  \begin{align*}
    g_k(\fF_R)&\le{}\tfrac{n_R-d-2}{k}
    \sum_{\emptyset\subset{}S\subseteq{}R}(-1)^{|S|}\tbinom{n_S-d-3+k}{k-1}
    +\sum_{\emptyset\subset{}S\subset{}R}\tfrac{n_{R\sm{}S}}{k}
    \tbinom{n_S-d-3+k}{k-1}\\
    &=\tfrac{n_R-d-2}{k}\tbinom{n_R-d-3+k}{k-1}
    -\sum_{\emptyset\subset{}S\subset{}R}\tfrac{n_R-d-2}{k}
    \tbinom{n_S-d-3+k}{k-1}
    +\sum_{\emptyset\subset{}S\subset{}R}\tfrac{n_{R\sm{}S}}{k}
    \tbinom{n_S-d-3+k}{k-1}\\
    &=\tfrac{n_R-d-2+k}{k}\tbinom{n_R-d-3+k}{k-1}-\tbinom{n_R-d-3+k}{k-1}
    -\sum_{\emptyset\subset{}S\subset{}R}\tfrac{n_R-d-2-n_{R\sm{}S}}{k}
    \tbinom{n_S-d-3+k}{k-1}\\
    &=\tbinom{n_R-d-2+k}{k}-\tbinom{n_R-d-3+k}{k-1}
    -\sum_{\emptyset\subset{}S\subset{}R}\tfrac{n_S-d-2}{k}
    \tbinom{n_S-d-3+k}{k-1}\\
    &=\tbinom{n_R-d-3+k}{k}
    -\sum_{\emptyset\subset{}S\subset{}R}\left[\tfrac{n_S-d-2+k}{k}
    \tbinom{n_S-d-3+k}{k-1}-\tbinom{n_S-d-3+k}{k-1}\right]\\
    &=\tbinom{n_R-d-3+k}{k}
    -\sum_{\emptyset\subset{}S\subset{}R}
    \left[\tbinom{n_S-d-2+k}{k}-\tbinom{n_S-d-3+k}{k-1}\right]\\
    &=\tbinom{n_R-d-3+k}{k}
    -\sum_{\emptyset\subset{}S\subset{}R}\tbinom{n_S-d-3+k}{k}\\
    &=\sum_{\emptyset\subset{}S\subseteq{}R}(-1)^{|S|}\tbinom{n_S-d-3+k}{k}.
  \end{align*}
  
  We focus now on the equality claim. Suppose first that
  $f_{l-1}(\fF_R)=\sum_{\emptyset\subset{}S\subseteq{}R}
  (-1)^{|S|}\binom{n_S}{l}$, for all $0\le{}l\le{}k$. Then, by
  \cite[Lemma 3.3]{kt-mnfms-12},
  $h_\lambda(\fF_R)=\sum_{\emptyset\subset{}S\subseteq{}R}
  (-1)^{|S|}\binom{n_S-d-2+\lambda}{\lambda}$,
  for $\lambda=k-1,k$, which gives:
  \begin{align*}
    g_k(\fF_R)&=h_k(\fF_R)-h_{k-1}(\fF_R)=
    \sum_{\emptyset\subset{}S\subseteq{}R}(-1)^{|S|}\tbinom{n_S-d-2+k}{k}
    -\sum_{\emptyset\subset{}S\subseteq{}R}(-1)^{|S|}\tbinom{n_S-d-2+k-1}{k-1}\\
    &=\sum_{\emptyset\subset{}S\subseteq{}R}(-1)^{|S|}
    \left[\tbinom{n_S-d-2+k}{k}-\tbinom{n_S-d-2+k-1}{k-1}\right]
    =\sum_{\emptyset\subset{}S\subseteq{}R}(-1)^{|S|}\tbinom{n_S-d-3+k}{k}.
  \end{align*}

  Suppose now that $g_k(\fF_R)=
  \sum_{\emptyset\subset{}S\subseteq{}R}(-1)^{|S|}\tbinom{n_S-d-3+k}{k}$.
  By relation \eqref{equ:gkFrecur}, we conclude that
  $h_{k-1}(\fF_{R})$ must be equal to its upper bound (cf. \cite[Lemma
  3.3]{kt-mnfms-12}), since, otherwise,
  $g_k(\fF_{R})$ would not be maximal, which contradicts our assumption
  on the value of $g_k(\fF_{R})$. This gives:
  \begin{align*}
    h_k(\fF_R)&=g_k(\fF_R)+h_{k-1}(\fF_R)
    =\sum_{\emptyset\subset{}S\subseteq{}R}(-1)^{|S|}\tbinom{n_S-d-3+k}{k}
    +\sum_{\emptyset\subset{}S\subseteq{}R}(-1)^{|S|}\tbinom{n_S-d-2+k-1}{k-1}\\
    &=\sum_{\emptyset\subset{}S\subseteq{}R}(-1)^{|S|}
    \left[\tbinom{n_S-d-2+k-1}{k}+\tbinom{n_S-d-2+k-1}{k-1}\right]
    =\sum_{\emptyset\subset{}S\subseteq{}R}(-1)^{|S|}\tbinom{n_S-d-2+k}{k}.
  \end{align*}
  Now the fact that $h_k(\fF_{R})$ is maximal, implies that
  $h_l(\fF_{R})$ must be equal to its maximal value for all $0\le{}l<k$.
  To see this suppose that $h_l(\fF_{R})$ is not maximal for some
  $l$, with $0\le{}l<k$, and among all such $l$ choose the largest
  one. Then, Lemmas 3.2 and 3.3 in \cite{kt-mnfms-12} imply that
  $h_{l+1}(\fF_{R})$ cannot be maximal, which contradicts the
  maximality of $l$.
  Summarizing, we deduce that if $g_k(\fF_{R})$ is equal to its upper
  bound in \eqref{equ:gkFbound}, so is $h_l(\fF_{R})$ for all
  $0\le{}l\le{}k$. By Lemma 3.3 in \cite{kt-mnfms-12}, this implies
  that $f_{l-1}(\fF_R)=\sum_{\emptyset\subset{}S\subseteq{}R}
  (-1)^{|S|}\binom{n_S}{l}$, for all $0\le{}l\le{}k$.
\end{proof}

\begin{lemma}\label{lem:hkWbound}
  For all $0\le{}k\le{}d+2$, we have:
  \begin{equation}\label{equ:hkWbound}
    h_k(\fF_{[3]})\le{}\sum_{\emptyset\subset{}S\subseteq[3]}
    (-1)^{3-|S|}\binom{n_S-d-3+k}{k},\qquad{}n_S = \sum_{i\in{}S}n_i.
  \end{equation}
  Equality holds for some 
  $0\le{}k\le{} \lexp{d+2}$, if and only if
  $f_{l-1}(\fF_{[3]})=\sum_{\emptyset\subset{}S\subseteq{}[3]}
  (-1)^{3-|S|}\binom{n_S}{l}$,
  for all $0\le{}l\le{}k$.
\end{lemma}

\begin{proof}
  We are going to prove relation \eqref{equ:hkWbound} by induction on
  $k$. The result clearly holds for $k=0$, since
  \begin{equation*}
    h_0(\fF_{[3]})=1=1-3+3=\tbinom{n_{[3]}-d-3}{0}
    -\sum_{i=1}^3\tbinom{n_{[3]\sm\{i\}}-d-3}{0}+\sum_{i=1}^3\tbinom{n_i-d-3}{0}.
  \end{equation*}
  
  Suppose the bound holds for some $k\ge{}0$. We will show that it
  holds for $k+1$. Using relation \eqref{r8}, Lemma \ref{lem:gkFbound},
  and the fact that, for
  any $k\ge{}0$, $n_{[3]}-d-2+k\ge{}3(d+1)-d-2=2d+1>0$, we have:
  {\allowdisplaybreaks
    \begin{align*}
      h_{k+1}(\fF_{[3]})&\leq\tfrac{n_{[3]}-d-2+k}{k+1} h_k(\fF_{[3]})
      + \sum_{i=1}^3\tfrac{n_i}{k+1}g_k(\fF_{[3]\sm{}\{i\}})\\
      &\le\tfrac{n_{[3]}-d-2+k}{k+1}
      \sum_{\emptyset\subset{}S\subseteq{}[3]}(-1)^{3-|S|}\tbinom{n_{S}-d-3+k}{k}
      +\sum_{i=1}^3\tfrac{n_i}{k+1}
      \sum_{\emptyset\subset{}S\subseteq[3]\sm\{i\}}(-1)^{|S|}
      \tbinom{n_{S}-d-3+k}{k}\\
      &=\tfrac{n_{[3]}-d-2+k}{k+1}\tbinom{n_{[3]}-d-3+k}{k}
      -\sum_{i=1}^3
      \tfrac{n_{[3]}-d-2+k}{k+1}\tbinom{n_{[3]\sm\{i\}}-d-3+k}{k}
      +\sum_{i=1}^3
      \tfrac{n_{[3]}-d-2+k}{k+1}\tbinom{n_i-d-3+k}{k}\\
      &\qquad
      +\sum_{i=1}^3\tfrac{n_i}{k+1}
      \tbinom{n_{[3]\sm\{i\}}-d-3+k}{k}
      -\sum_{i=1}^3\tfrac{n_i}{k+1}\sum_{j\in[3]\sm\{i\}}
      \tbinom{n_{j}-d-3+k}{k}\\
      &=\tbinom{n_{[3]}-d-2+k}{k+1}
      -\sum_{i=1}^3\tfrac{n_{[3]}-d-2+k-n_i}{k+1}
      \tbinom{n_{[3]\sm\{i\}}-d-3+k}{k}
      +\sum_{i=1}^3\tfrac{n_{[3]}-d-2+k-n_{[3]\sm\{i\}}}{k+1}
      \tbinom{n_i-d-3+k}{k}\\
      &=\tbinom{n_{[3]}-d-2+k}{k+1}
      -\sum_{i=1}^3\tfrac{n_{[3]\sm\{i\}}-d-2+k}{k+1}
      \tbinom{n_{[3]\sm\{i\}}-d-3+k}{k}
      +\sum_{i=1}^3\tfrac{n_{i}-d-2+k}{k+1}
      \tbinom{n_i-d-3+k}{k}\\
      &=\tbinom{n_{[3]}-d-2+k}{k+1}
      -\sum_{i=1}^3\tbinom{n_{[3]\sm\{i\}}-d-2+k}{k+1}
      +\sum_{i=1}^3\tbinom{n_i-d-2+k}{k+1}\\
      &=\sum_{\emptyset\subset{}S\subseteq{}[3]}(-1)^{3-|S|}
      \tbinom{n_{S}-d-2+k}{k+1},
    \end{align*}}
  where we used the fact that:
  \begin{align*}
    \sum_{i=1}^3\tfrac{n_{[3]\sm\{i\}}}{k+1}\tbinom{n_i-d-3+k}{k}
    &=\sum_{i=1}^3\left(\sum_{j\in[3]\sm\{i\}}\tfrac{n_j}{k+1}\right)
    \tbinom{n_i-d-3+k}{k}
    =\sum_{i=1}^3\sum_{j\in[3]\sm\{i\}}\tfrac{n_j}{k+1}\tbinom{n_i-d-3+k}{k}\\
    &=\sum_{i=1}^3\sum_{j\in[3]\sm\{i\}}\tfrac{n_i}{k+1}\tbinom{n_j-d-3+k}{k}
    =\sum_{i=1}^3\tfrac{n_i}{k+1}\sum_{j\in[3]\sm\{i\}}\tbinom{n_j-d-3+k}{k}.
  \end{align*}

  The rest of the proof is concerned with the equality claim. Assume
  first that
  $f_{l-1}(\fF_{[3]})=\sum_{\emptyset\subset{}S\subseteq{}[3]}
  (-1)^{3-|S|}\binom{n_S}{l}$, for all $0\le{}l\le{}k$. Then we have:
  \begin{align*}
    h_k(\fF_{[3]})&=\sum_{i=0}^{d+2}(-1)^{k-i}\tbinom{d+2-i}{d+2-k}
    f_{i-1}(\fF_{[3]})
    =(-1)^k\sum_{i=0}^{d+2}(-1)^{i}\tbinom{d+2-i}{d+2-k}
    \sum_{\emptyset\subset{}S\subseteq{}[3]}(-1)^{3-|S|}\tbinom{n_S}{i}\\
    &=(-1)^k\sum_{\emptyset\subset{}S\subseteq{}[3]}(-1)^{3-|S|}
    \sum_{i=0}^{d+2}(-1)^{i}\tbinom{d+2-i}{d+2-k}
    \tbinom{n_S}{i}
    =\sum_{\emptyset\subset{}S\subseteq{}[3]}(-1)^{3-|S|}
    \tbinom{n_S-d-3+k}{k}.
  \end{align*}
  In the above relation we used the combinatorial identity
  (cf. \cite[eq.~(5.25)]{gkp-cm-89}):
  \begin{equation*}
    \sum_{0\le{}k\le{}l}\binom{l-k}{m}\binom{s}{k-n}(-1)^k
    =(-1)^{l+m}\binom{s-m-1}{l-m-n},
  \end{equation*}
  where $k\sub{}i$, $l\sub{}d+2$, $m\sub{}d+2-k$, $n\sub{}0$, and
  $s\sub{}n_S$.

  Suppose now that
  $h_k(\fF_{[3]})=\sum_{\emptyset\subset{}S\subseteq{}[3]}
  (-1)^{3-|S|}\binom{n_S-d-3+k}{k}$. Since relation \eqref{r8} holds
  for all $k\ge{}0$, we conclude that $h_l(\fF_{[3]})$ must be equal
  to its upper bound in \eqref{equ:hkWbound}, for all $0\le{}l<k$. 
  To see this suppose that \eqref{equ:hkWbound} is not tight for some
  $l$, with $0\le{}l<k$, and among all such $l$ choose the largest
  one. Then, relation \eqref{r8} implies that
  $h_{l+1}(\fF_{[3]})$ cannot be equal to its upper bound from
  \eqref{equ:hkWbound}, which contradicts the maximality of $l$.
  Hence, if $h_k(\fF_{[3]})$ is equal to its upper bound in
  \eqref{equ:hkWbound}, so is $h_l(\fF_{[3]})$ for all $0\le{}l<k$,
  which gives, for all $l$ with $0\le{}l\le{}k$:
  \begin{align}
    f_{l-1}(\fF_{[3]})&=\sum_{i=0}^{d+2}\tbinom{d+2-i}{l-i}h_i(\fF_{[3]})
    =\sum_{i=0}^{d+2}\tbinom{d+2-i}{l-i}\sum_{\emptyset\subset{}S\subseteq{}[3]}
    (-1)^{3-|S|}\tbinom{n_S-d-3+i}{i}\notag\\
    &=\sum_{\emptyset\subset{}S\subseteq{}[3]}(-1)^{3-|S|}
    \sum_{i=0}^{d+2}\tbinom{d+2-i}{l-i}\tbinom{n_S-d-3+i}{i}
    =\sum_{\emptyset\subset{}S\subseteq{}[3]}(-1)^{3-|S|}
    \sum_{i=0}^{d+2}\tbinom{d+2-i}{d+2-l}
    \tbinom{n_S-d-3+i}{n_S-d-3}\label{equ:s1}\\
    &=\sum_{\emptyset\subset{}S\subseteq{}[3]}(-1)^{3-|S|}\tbinom{n_S}{n_S-l}
    =\sum_{\emptyset\subset{}S\subseteq{}[3]}(-1)^{3-|S|}\tbinom{n_S}{l},
    \label{equ:s2}
  \end{align}
  where, in order to get from \eqref{equ:s1} to \eqref{equ:s2}, we
  used the combinatorial identity (cf. \cite[eq.~(5.26)]{gkp-cm-89}):
  \begin{equation*}
    \sum_{0\le{}k\le{}l}\binom{l-k}{m}\binom{q+k}{n}=\binom{l+q+1}{m+n+1},
  \end{equation*}
  with $k\sub{}i$, $l\sub{}d+2$, $m\sub{}d+2-l$,
  $q\sub{}n_S-d-3$, and $n\sub{}n_S-d-3$.
\end{proof}

We are now going to bound the elements of the $h$-vector of
$\kK_{[3]}$. More precisely:
\begin{lemma}\label{lem:hkKbound}
  For all $0\le{}k\le{}d+2$, we have:
  \begin{equation}\label{equ:hkKbound-smallk}
    h_k(\kK_{[3]})\le\binom{n_{[3]}-d-3+k}{k}.
  \end{equation}
  Furthermore, for $d\ge{}3$ and $d$ odd, we have:
  \begin{equation}\label{equ:hkKbound-midk}
    h_{\lexp{d}+1}(\kK_{[3]})\le\binom{n_{[3]}-\lexp{d}-3}{\lexp{d}+1}
    -\sum_{i=1}^{3}\binom{n_i-\lexp{d}-2}{\lexp{d}+1}.
  \end{equation}
  Equality holds for some $k$, where $0\le{}k\le\lexp{d+1}$, if and
  only if, for all $\emptyset\subset{}R\subseteq[3]$,
  $f_{l-1}(\fF_{R})=\sum_{\emptyset\subset{}S\subseteq{}R}
  (-1)^{|R|-|S|}\binom{n_S}{l}$, for all
  $0\le{}l\le\min\{k,\lexp{d+|R|-1}\}$.
\end{lemma}
\begin{proof}[Sketch of proof]
The complete proof can be found in Section~\ref{app:ub} of Appendix~\ref{app:omitted}.
To prove the upper bound for $h_k(\kK_{[3]})$, we distinguish between
two cases: (1) the case $k=0$, where the result follows by a
straightforward calculation from relation \eqref{equ:hkKR} with
$R=[3]$, and (2) the case
$k\geq 1$, where again we use \eqref{equ:hkKR} with $R=[3]$
and substitute $g_k(\fF_R)$ by its upper bound from relation~\eqref{equ:gkFrecur} in Lemma~\ref{lem:gkFbound}.
We, thus, obtain a bound for $h_k(\kK_{[3]})$ expressed in
terms of $h_k(\fF_{[3]})$, $h_{k-1}(\fF_R), R\in\sset_2$, and
$g_\lambda(\bx{}P_i)$, $\lambda=k,k-1$.
Combining the upper bounds from Lemma~\ref{lem:hkWbound},
Lemma 3.3 in \cite{kt-mnfms-12}, along with the upper
bounds for the $g$-vector of a $d$-polytope (cf.~\cite[Corollary
8.38]{z-lp-95}), respectively, gives the upper bound in the statement of the lemma.

For the equality claim we assume that $h_k(\kK_{[3]})$ attains its maximal value.
Then, the expression bounding $h_k(\kK_{[3]})$ used above,
in conjunction with Lemmas \ref{lem:hkWrecur}, \ref{lem:gkFbound},
\ref{lem:hkWbound},
and \cite[Lemma 3.3]{kt-mnfms-12}, yields the equality conditions
in the statement of the lemma. 
In the opposite direction, we assume that these conditions hold and,
using Lemma \ref{lem:hkWbound} and \cite[Lemma 3.3]{kt-mnfms-12}, 
we show that the quantities in the right hand side of relation~\eqref{equ:hkKR} with $R=[3]$, 
attain their maximal values. The conclusion then follows from an easy calculation. 
\end{proof}

We are now ready to state and prove the main theorem of the paper
concerning upper bounds on the number of $k$-faces of the
Minkowski sum of three convex $d$-polytopes.

\begin{theorem}\label{thm:ms3ub}
  Let $P_1$, $P_2$ and $P_3$ be three $d$-polytopes in $\reals^d$,
  $d\ge{}2$, with $n_i\ge{}d+1$ vertices, $1\le{}i\le{}3$. Then, for all
  $1\le{}k\le{}d$, we have:
  \begin{equation}\label{equ:ms3ub}
    \begin{aligned}
      f_{k-1}(P_1+P_2+P_3)&\le
      f_{k+1}(C_{d+2}(n_{[3]}))-\sum_{i=0}^{\lexp{d+2}}\binom{d+2-i}{k+2-i}
      \sum_{\emptyset\subset{}S\subset[3]}(-1)^{|S|}\binom{n_S-d-3+i}{i}\\
      &\quad-\delta\binom{\lexp{d}+1}{k-\lexp{d}}
      \sum_{i=1}^3\binom{n_i-\lexp{d}-2}{\lexp{d}+1},
    \end{aligned}
  \end{equation}
  where $\delta=d-2\lexp{d}$, and $n_S=\sum_{i\in{}S}n_i$.
  Equality holds for all $1\le{}k\le{}d$, if and only if
  \begin{equation}\label{equ:full-conditions}
    f_{l-1}(\fF_{R})=\sum_{\emptyset\subset{}S\subseteq{}R}
    (-1)^{|R|-|S|}\binom{n_S}{l},\qquad
    0\le{}l\le\ltexp{d+|R|-1},
    \qquad\emptyset\subset{}R\subseteq[3].
  \end{equation}
\end{theorem}

\begin{proof}
  If suffices to establish upper bounds for $f_k(\fF_{[3]})$ for all
  $0\le{}k\le{}d+1$. Indeed,
  writing the $f$-vector of $\fF_{[3]}$ in terms of its $h$-vector, and using
  relation \eqref{equ:DSW}, along with Lemmas \ref{lem:hkWbound} and
  \ref{lem:hkKbound} we get:
  \begin{align}
   \nonumber f_{k-1}(\fF_{[3]})&=\sum_{i=0}^{d+2}\tbinom{d+2-i}{k-i}h_i(\fF_{[3]})
    =\sum_{i=0}^{\lexp{d+2}}\tbinom{d+2-i}{k-i}h_i(\fF_{[3]})
    +\sum_{i=\lexp{d+2}+1}^{d+2}\tbinom{d+2-i}{k-i}h_i(\fF_{[3]})\\
   \nonumber &=\sum_{i=0}^{\lexp{d+2}}\tbinom{d+2-i}{k-i}h_i(\fF_{[3]})
    +\sum_{j=0}^{\lexp{d+1}}\tbinom{j}{k-d-2+j}h_{d+2-j}(\fF_{[3]})\\
    &=\sum_{i=0}^{\lexp{d+2}}\tbinom{d+2-i}{k-i}h_i(\fF_{[3]})
    +\sum_{j=0}^{\lexp{d+1}}\tbinom{j}{k-d-2+j}h_j(\kK_{[3]}).
    \label{equ:fkF3_as_sum}
  \end{align}
  From Lemma \ref{lem:hkWbound} we have:
  \begin{equation*}
    \sum_{i=0}^{\lexp{d+2}}\tbinom{d+2-i}{k-i}h_i(\fF_{[3]})
    \le\sum_{i=0}^{\lexp{d+2}}\tbinom{d+2-i}{k-i}
    \sum_{\emptyset\subset{}S\subseteq[3]}(-1)^{3-|S|}\tbinom{n_S-d-3+i}{i},
  \end{equation*}
  whereas from Lemma \ref{lem:hkKbound} we get
  \begin{equation*}
    \sum_{j=0}^{\lexp{d+1}}\tbinom{j}{k-d-2+j}h_j(\kK_{[3]})
    \le\sum_{i=0}^{\lexp{d+1}}\tbinom{n_{[3]}-d-3+j}{j}
    -\delta\tbinom{\lexp{d}+1}{k-\lexp{d}-2}
    \sum_{i=1}^3\tbinom{n_i-\lexp{d}-2}{\lexp{d}+1},
  \end{equation*}
  where $\delta=d-2\lexp{d}$. Hence:
  {\allowdisplaybreaks
  \begin{align*}
    f_{k-1}(\fF_{[3]})
    &\le\sum_{i=0}^{\lexp{d+2}}\tbinom{d+2-i}{k-i}
    \sum_{\emptyset\subset{}S\subseteq[3]}(-1)^{3-|S|}\tbinom{n_S-d-3+i}{i}
    +\sum_{j=0}^{\lexp{d+1}}\tbinom{j}{k-d-2+j}\tbinom{n_{[3]}-d-3+j}{j}\\
    &\qquad-\delta\tbinom{\lexp{d}+1}{k-\lexp{d}-2}
    \sum_{i=1}^3\tbinom{n_i-\lexp{d}-2}{\lexp{d}+1}\\
    &=\sum_{i=0}^{\lexp{d+2}}\tbinom{d+2-i}{k-i}\tbinom{n_{[3]}-d-3+i}{i}
    +\sum_{i=0}^{\lexp{d+1}}\tbinom{i}{k-d-2+i}\tbinom{n_{[3]}-d-3+i}{i}\\
    &\qquad
    -\sum_{i=0}^{\lexp{d+2}}\tbinom{d+2-i}{k-i}
    \sum_{\emptyset\subset{}S\subset[3]}(-1)^{|S|}\tbinom{n_S-d-3+i}{i}
    -\delta\tbinom{\lexp{d}+1}{k-\lexp{d}-2}
    \sum_{i=1}^3\tbinom{n_i-\lexp{d}-2}{\lexp{d}+1}\\
    &=\sideset{}{^{\,*}}{\sum}_{i=0}^{\frac{d+2}{2}}
    (\tbinom{d+2-i}{k-i}+\tbinom{i}{k-d-2+i})\tbinom{n_{[3]}-d-3+i}{i}
    -\sum_{i=0}^{\lexp{d+2}}\tbinom{d+2-i}{k-i}
    \sum_{\emptyset\subset{}S\subset[3]}(-1)^{|S|}\tbinom{n_S-d-3+i}{i}\\
    &\qquad-\delta\tbinom{\lexp{d}+1}{k-\lexp{d}-2}
    \sum_{i=1}^3\tbinom{n_i-\lexp{d}-2}{\lexp{d}+1}\\
    &=f_{k-1}(C_{d+2}(n_{[3]}))-\sum_{i=0}^{\lexp{d+2}}\tbinom{d+2-i}{k-i}
    \sum_{\emptyset\subset{}S\subset[3]}(-1)^{|S|}\tbinom{n_S-d-3+i}{i}\\
    &\qquad-\delta\tbinom{\lexp{d}+1}{k-\lexp{d}-2}
    \sum_{i=1}^3\tbinom{n_i-\lexp{d}-2}{\lexp{d}+1},
  \end{align*}}%
  where:
  \begin{equation*}
    \sideset{}{^{\,*}}{\sum}_{i=0}^{\frac{m}{2}}T_i=
    \sum_{i=0}^{\lexp{m}-1}T_i+
    \tfrac{1}{2}\left(1+m-2\ltexp{m}\right)
    T_{\lexp{m}}.
  \end{equation*}
  Our upper bounds follow from the fact that
  $f_{k-1}(P_1+P_2+P_3)=f_{k+1}(\fF_{[3]})$, $1\le{}k\le{}d$.

  In what follows we concentrate on the necessary and sufficient
  conditions for the upper bounds in \eqref{equ:ms3ub} to hold as
  equalities.
  From the derivation of the upper bounds above (see also
  relation~\eqref{equ:fkF3_as_sum}), it is clear that the
  bounds are tight if and only if:
  \begin{enumerate}
  \item[(1)]
    $h_k(\fF_{[3]})$ is maximal, for all
    $0\le{}k\le\lexp{d+2}$, and
  \item[(2)]
    $h_k(\kK_{[3]})$ is maximal, for all $0\le{}k\le\lexp{d+1}$.
  \end{enumerate}
  According to Lemma \ref{lem:hkWbound} and
  Lemma \ref{lem:hkKbound}, these conditions are, respectively,
  equivalent to requiring that:
  \begin{enumerate}
  \item[(1)]
    $f_{l-1}(\fF_{[3]})=\sum_{\emptyset\subset{}S\subseteq{}[3]}
    (-1)^{3-|S|}\binom{n_S}{l}$, for all $0\le{}l\le\lexp{d+2}$, and 
  \item[(2)]
    $f_{l-1}(\fF_{R})=\sum_{\emptyset\subset{}S\subseteq{}R}
    (-1)^{|R|-|S|}\binom{n_S}{l}$, for all
    $0\le{}l\le\min\{\lexp{d+1},\lexp{d+|R|-1}\}$, and for all
    $\emptyset\subset{}R\subseteq[3]$.
  \end{enumerate}
  For $R\equiv[3]$, condition (1) implies condition (2), while for
  $R\subset[3]$,
  $\min\{\lexp{d+1},\lexp{d+|R|-1}\}=\lexp{d+|R|-1}$.
  We, therefore, conclude that the bounds in \eqref{equ:ms3ub} are
  attained if and only if, conditions \eqref{equ:full-conditions} hold
  true for all $0\le{}k\le\lexp{d-|R|+1}$ and for all
  $\emptyset\subset{}R\subseteq[3]$.
\end{proof}

\section{Tightness of upper bounds}
\label{sec:lb}

In this section we show that the bounds in Theorem \ref{thm:ms3ub} are
tight. We distinguish between the cases $d=2$, $d=3$ and
$d\ge{}4$. For $d=2$, it is easy to verify that for $k=1,2$, the
right-hard side of inequality \eqref{equ:ms3ub} evaluates to
$n_1+n_2+n_3$, which is known to be tight.

\subsection{Three dimensions}\label{sec:lb:3d}
For $d=3$, the upper bounds in Theorem \ref{thm:ms3ub} are as follows:
\begin{equation}\label{equ:ms3-3polytopes-ubounds}
  \begin{aligned}
    f_0(P_1+P_2+P_3)&\le{}n_1 n_2+n_2 n_3+n_1 n_3-n_1-n_2-n_3+2,\\
    f_1(P_1+P_2+P_3)&\le{}2 n_1 n_2+2 n_2 n_3+2 n_1 n_3-n_1-n_2-n_3-6,\\
    f_2(P_1+P_2+P_3)&\le{}n_1 n_2+n_2 n_3+n_1 n_3-6.
  \end{aligned}
\end{equation}
In order to prove that these bounds are tight, we exploit two results:
one by Fukuda and Weibel \cite{fw-fmacp-07} and one by Weibel
\cite{w-mfmsl-12}.
Weibel \cite{w-mfmsl-12} has shown that the number of $k$-faces of the
Minkowski sum of $r$ $d$-polytopes $P_1,\ldots,P_r$ in $\reals^d$,
where $r\ge{}d$, is related to the number of $k$-faces of the
Minkowski sum of subsets of these polytopes of size at most $d-1$ as
follows:
\begin{equation}\label{equ:k-faces-large-r}
  f_k(\MS)-\alpha=\sum_{j=1}^{d-1}(-1)^{d-1-j}\binom{r-1-j}{d-1-j}
  \sum_{S\in\sset_j^r}(f_k(P_S)-\alpha),
\end{equation}
where $\sset_j^r$ is the family of subsets of $[r]$ of size $j$, $P_S$
is the Minkowski sum of the polytopes in $S$, and $\alpha=2$ if $k=0$
and $d$ is odd, and $\alpha=0$ otherwise. For $d=r=3$,
equation \eqref{equ:k-faces-large-r} simplifies to:
\begin{equation}\label{equ:ms3-3polytopes-subsets}
  \begin{aligned}
    f_k(P_1+P_2+P_3)&=\alpha+\sum_{j=1}^2(-1)^{2-j}\tbinom{2-j}{2-j}
    \sum_{S\in\sset_j^3}(f_k(P_S)-\alpha)\\
    &=\alpha-\sum_{i=1}^3(f_k(P_i)-\alpha)
    +\sum_{i=1}^3(f_k(P_{[3]\sm\{i\}})-\alpha)\\
    &=\alpha-\sum_{i=1}^{3}f_k(P_i)+3\alpha
    +\sum_{i=1}^{3}f_k(P_{[3]\sm\{i\}})-3\alpha\\
    &=\alpha+\sum_{1\le{}i<j\le{}3}f_k(P_i+P_j)-\sum_{i=1}^{3}f_k(P_i).
  \end{aligned}
\end{equation}

Besides relation \eqref{equ:k-faces-large-r}, Weibel \cite{w-mfmsl-12}
also presented a construction of $r$ simplicial $d$-polytopes, such
that any subset $S$ of these polytopes of size at most $d-1$ has the
maximum possible number of vertices, namely,
$f_0(P_S)=\prod_{i\in{}S}n_i$.
Specializing this construction in our case, \ie for $r=d=3$, we
deduce that it is possible to construct three simplicial $3$-polytopes
$P_1$, $P_2$, $P_3$ in $\reals^3$, such that $f_0(P_i)=n_i$, $1\le{}i\le{}3$, 
and $f_0(P_i+P_j)=n_{i}n_j$, $1\le{}i<j\le{}3$. Substituting in
\eqref{equ:ms3-3polytopes-subsets} for $k=0$, we get:
\begin{equation*}
  f_0(P_1+P_2+P_3)=2+\sum_{1\le{}i<j\le{}3}n_{i}n_j-\sum_{i=1}^{3}n_i
  =n_1 n_2+n_2 n_3+n_1 n_3-n_1-n_2-n_3+2,
\end{equation*}
\ie the upper bound in \eqref{equ:ms3-3polytopes-ubounds} is tight for
$k=0$\footnote{This is essentially the result of Theorem 3 in
  \cite{w-mfmsl-12} for $d=r=3$; however, we recapitulate this result
  in order to show that Weibel's construction yields tights bounds for
  $k=1,2$ also.}.
Since all $P_i$'s are simplicial, we have
\begin{equation}\label{equ:simplicial-3polytopes}
  f_1(P_i)=3 n_i-6,\quad
  f_2(P_i)=2 n_i-4,
  \qquad 1\le{}i\le{}3.
\end{equation}
On the other hand, since $f_0(P_i+P_j)$ is maximal, for all
$1\le{}i<j\le{}3$, we get, by \cite[Corollary 4]{fw-fmacp-07}, that
$f_k(P_i+P_j)$ is also maximal for $k=1,2$, and for all
$1\le{}i<j\le{}3$. Hence:
\begin{equation}\label{equ:ms2-3polytopes}
  f_1(P_i+P_j)=2 n_i n_j+n_i+n_j-8,\qquad
  f_2(P_i+P_j)=n_i n_j+n_i+n_j-6.
\end{equation}
Substituting from \eqref{equ:simplicial-3polytopes} and
\eqref{equ:ms2-3polytopes} in \eqref{equ:ms3-3polytopes-subsets},
and recalling that $\alpha=0$ for $k>0$, we get:
\begin{equation*}
  \begin{aligned}
    f_1(P_1+P_2+P_3)&=\sum_{1\le{}i<j\le{}3}(2 n_i n_j+n_i+n_j-8)
    -\sum_{i=1}^{3}(3 n_i-6)\\
    &=[2(n_1 n_2+n_2 n_3+n_1 n_3)+2(n_1+n_2+n_3)-24]
    -[3 (n_1+n_2+n_3)-18]\\
    &=2 n_1 n_2+2 n_2 n_3+2 n_1 n_3-n_1-n_2-n_3-6,
  \end{aligned}
\end{equation*}
and
\begin{equation*}
  \begin{aligned}
    f_2(P_1+P_2+P_3)&=\sum_{1\le{}i<j\le{}3}(n_i n_j+n_i+n_j-6)
    -\sum_{i=1}^{3}(2 n_i-4)\\
    &=[n_1 n_2+n_2 n_3+n_1 n_3+2(n_1+n_2+n_3)-18]
    -[2 (n_1+n_2+n_3)-12]\\
    &=n_1 n_2+n_2 n_3+n_1 n_3-6,
  \end{aligned}
\end{equation*}
\ie the upper bounds in \eqref{equ:ms3-3polytopes-ubounds} are tight
for $k=1,2$.

\subsection{Four or more dimensions}
We now focus on the case $d\ge{}4$. We shall construct three 
$d$-polytopes $P_1, P_2$ and $P_3$ in $\reals^d$, such that they satisfy the
conditions in relation~\eqref{equ:full-conditions}.
Consequently, as Theorem \ref{thm:ms3ub} asserts, these
polytopes attain the upper bounds in \eqref{equ:ms3ub}.

Consider the following $d$-dimensional moment-like curves in $\reals^d$:
\begin{align*}
  \mc_1(t)&=(t,\z t^2,\z t^3,t^4,t^5,\ldots,t^{d}),\\
  \mc_2(t)&=(\z t,t^2,\z t^3,t^4,t^5,\ldots,t^{d}),\\
  \mc_3(t)&=(\z t,\z t^2,t^3,t^4,t^5,\ldots,t^{d}),
\end{align*}
where $t>0$, and $\z \geq 0$. Let $\me_{1,1}=(0),\me_{1,2}=(1)$ be the
standard affine basis of $\reals$ and recall that 
$\me_{2,1}=(0,0),\me_{2,2}=(1,0),\me_{2,3}=(0,1)$ is the
standard affine basis of $\reals^2$.
We shall define three polytopes as the convex hulls of points,
chosen appropriately on each of these $d$-curves.
We then proceed to show that $\fF_R$, $R\in\sset_2$, and $\fF_{[3]}$,
have the following property:
every set of $k=\ltexp{d+1}$ vertices from $\fF_R$, or $k\leq
\ltexp{d+2}$ vertices from $\fF_{[3]}$,
defines a $(k-1)$-face of $\fF_R$ or $\fF_{[3]}$, respectively.
This property readily yields the necessary and sufficient conditions
establishing the tightness of the upper bounds (cf.~rel. \eqref{equ:full-conditions}).

Let $x_{i,j}$, $1\le{}j\le{}n_i$, $1\le{}i\le{}3$, be $n_{[3]}$
positive real numbers, such that $x_{i,j}<x_{i,j+1}$,
$1\le{}j\le{}n_i-1$, and let $\tau$ be a positive real parameter.
Let $x_{i,j}^\epsilon=x_{i,j}+\epsilon$, $t_{i,j}=x_{i,j}\tau^{\nu_i}$,
$t_{i,j}^\epsilon=x_{i,j}^\epsilon\tau^{\nu_i}$, where
$1\le{}j\le{}n_i$, $1\le{}i\le{}3$, $\epsilon>0$, and 
$\nu_i=3-i,\, 1\leq i\leq 3$. The value of
$\epsilon$ is chosen such that $x_{i,j}^\epsilon<x_{i,j+1}$,
for all $1\le{}j<n_i$, and for all $1\le{}i\le{}3$.
Finally, we set $\z=\tau^M$, where $M\geq d(d+1)$.
We are going to define three vertex sets $V_i$ as follows: 
\begin{equation}
  V_i=\{\mc_i(t_{i,1}),\mc_i(t_{i,2}),\ldots \mc_i(t_{i,n_i})\}
  \qquad 1\le{}i\le{}3.
\end{equation}
Call $P_i$ the $d$-polytope we get as
the convex hull of the vertices in $V_i$,
and let $\VV_i$ be the image of $V_i$ 
via the Cayley embedding.
As in Section \ref{sec:cayley}, call $\cC$ the Cayley polytope of the
$P_i$'s in $\reals^{d+2}$, and
$\fF_R$, $\emptyset\subset{}R\subseteq[3]$, the set of faces of $\cC$ 
with at least one vertex from each $\VV_i$, $i\in{}R$.
Note that, by construction, $P_i$ is a 
$\ltexp{d}$-neighborly polytope in $\reals^d$ with $n_i$ vertices,
which immediately implies that conditions
\eqref{equ:full-conditions} hold for $R\in\sset_1$ and for all
$0\le{}l\le{}\lexp{d}$. Hence, it suffices to show that:
\begin{equation}\label{equ:fkWR}
  f_{l-1}(\fF_{R})=\sum_{\emptyset\subset{}S\subseteq{}R}
  (-1)^{|R|-|S|}\tbinom{n_S}{l},
  \qquad 0\le{}l\le\ltexp{d+|R|-1},
  \qquad 2\le{}|R|\le{}3,
\end{equation}
which we will succeed by choosing a sufficiently small value for $\tau$.

To prove that the constructed polytopes have the desired properties
(see Lemmas~\ref{lem:fkW2-tau} and \ref{lem:fkW3-tau}, bellow),
we adopt the key idea used in the proofs of \cite[Theorem 0.7 \&
Corollary 0.8]{z-lp-95}
on basic properties of cyclic $d$-polytopes, and adapt this idea
to our setting, where we view the faces the Minkowski sum of the
polytopes $P_i$, $i\in{}R$, via the face set $\fF_R$
of their Cayley polytope, where $2\le|R|\le{}3$.

We start off with subsets $R$ of size two. To show that
$f_{k-1}(\fF_R)$ is according to relation \eqref{equ:fkWR}, 
recall (cf.~Section~\ref{sec:cayley}) that the polytope $\cC$ contains the Cayley polytope
$\cC_R$ of the polytopes in $R$ as a $d$-subcomplex embedded in
$\reals^{d+2}$. Thus, in order to prove relation \eqref{equ:fkWR}
for $\fF_R$, we may consider $\cC_R$ and $\fF_R$ independently of
$\cC$, \ie we can disassociate the polytopes $P_i$, $i\in{}R$, from the
Cayley polytope $\cC$. In other words, we think of the polytopes
$P_i$, $i\in{}R$, as $d$-polytopes in $\reals^d$, 
while their Cayley polytope $\cC_R$ is seen as a
$(d+1)$-polytope in $\reals^{d+1}$.
We exploit this observation in order to prove the following lemma.

\begin{lemma}\label{lem:fkW2-tau}
There exists a sufficiently small positive value $\hat\tau_R$ for
$\tau$ such that, for all $\tau\in(0,\hat\tau_R)$,
\begin{equation*}
    f_{k-1}(\fF_R)=\sum_{\emptyset\subset{}S\subseteq{}R}
    (-1)^{2-|S|}\tbinom{n_S}{k},  \qquad 2\le{}k\le\ltexp{d+1},
    \quad R\in \sset_2.
\end{equation*}
\end{lemma}
\begin{proof}
Without loss of generality let $R=\{1,3\}$. The rest of the cases
are analogous.
The condition in the statement of the lemma is equivalent to the requirement
that $\cC_{\{1,3\}}$ is a $(\VV_{1},\ltexp{d+1})$-bineighborly polytope 
(see~\cite{kt-mnfms-12} for definitions and details),
which in turn is equivalent to the requirement that
\begin{equation}\label{equ:bineighborliness-condition}
   f_{\lexp{d+1}-1}(\fF_{\{1,3\}})=\sum_{\emptyset\subset{}S\subseteq{} \{1,3\}}
   (-1)^{2-|S|}\tbinom{n_S}{\lexp{d+1}}.
\end{equation}
  
We shall prove that condition
\eqref{equ:bineighborliness-condition} holds true for the Cayley
polytope $\cC_{\{1,3\}}$ of the polytopes $P_1, P_3$, and for
sufficiently small values of $\tau$, as described in the statement
of the lemma.

Define $\delta:=d+1-2\ltexp{d+1}$. Let $X$ be a positive real
number such that $X>x_{3,n_3}^\epsilon$, and let\footnote{Although
we have set $\nu_3=0$, we keep $\nu_3$ as is in the
proof, so as to make more profound the analogy of the proof
presented here for $R=\{1,3\}$ with
the cases $R=\{1,2\}$ and $R=\{2,3\}$.} $T\allowbreak = X\tau^{\nu_3}$.
Choose a set $U$ of $k_{m}\neq 0$ vertices  
$\mc_m(t_{m,j_{m,1}}),\mc_m(t_{m,j_{m,2}}),\ldots,\mc_m(t_{m,j_{m,k_{m}}})$
from the set $V_{m}$, such that $j_{m,1}<j_{m,2}<\ldots<j_{m,k_{m}}$, for $m\in\{1,3\}$,
and $k_{1}+k_{3}=\ltexp{d+1}$. 
Let $\UU=\{\mcc_m(t_{m,j_{m,1}}),\mcc_m(t_{m,j_{m,2}}),\ldots,\allowbreak
\mcc_m(t_{m,j_{m,k_{m}}})\mid m\in\{1,3\}\}$, be the Cayley
embedding of $U$ in $\reals^{d+1}$ (using the
affine basis $\me_{1,1}, \me_{1,2} $).
For a vector $\mb{x}=(x_1,x_2,\ldots,x_{d+1})\in\reals^{d+1}$, we define the
$(d+2)\times(d+2)$ determinant $H_{\UU}(\mb{x})$ as follows: 
\begin{center}
   \includegraphics[width=\textwidth]{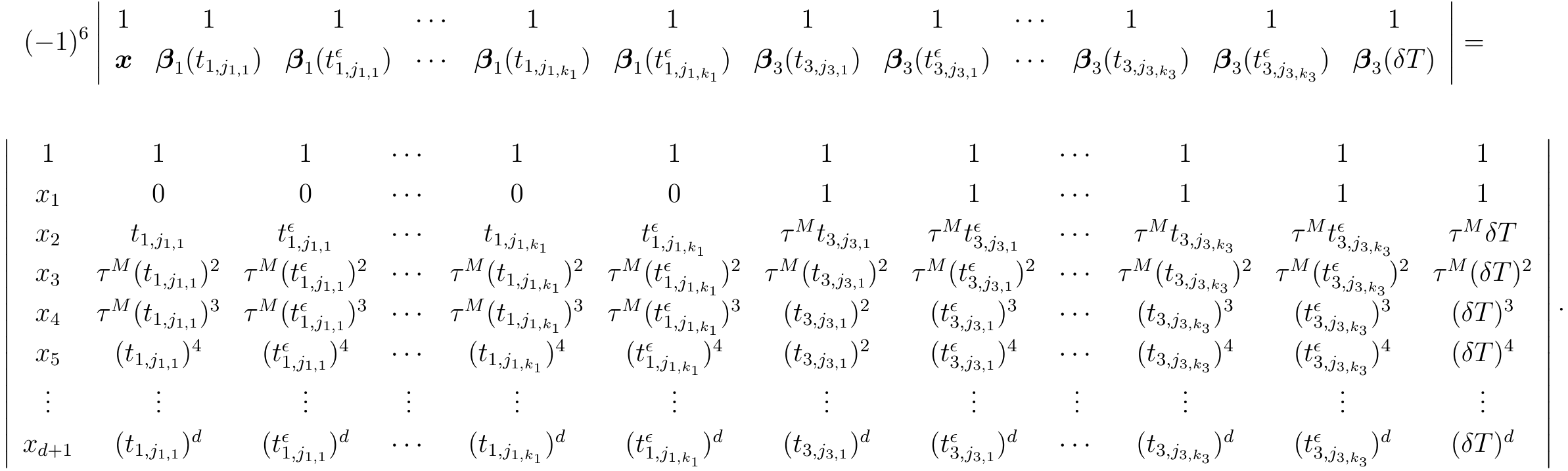}
\end{center}
Notice that for $d$ odd the last column $\binom{1}{\mcc_3(\delta T)}$  of $H_{\UU}(x)$ does not exist.
The equation $H_{\UU}(\mb{x})=0$ is the equation of a
hyperplane in $\reals^{d+1}$ that passes through the points in
$\UU$. We are going to show that, for any choice of $\UU$,
and for all vertices $\mbv$ in $\VV_{\{1,3\}}\sm{}\UU$, $\VV_{\{1,3\}}=\VV_1\cup \VV_3$,
we have $H_{\UU}(\mbv)>0$ for sufficiently small values of $\tau$.

Suppose we have some vertex $\mbv\in{}\VV_{\{1,3\}}\sm{}\UU$. 
Then, $\mbv=\mcc_s(t_{s,\lambda})$,
$t_{s,\lambda}=x_{s,\lambda}\tau^{\nu_{s}}$, 
where $1\leq\lambda\leq n_{s}$,  $s$ is either $1$ or $3$, and 
$\lambda \notin\{j_{s,1},j_{s,2},\ldots,j_{s,k_{s}}\}$.
We perform the following determinant transformations on $H_{\UU}(\mbv)$:
initially we subtract its second row from its first, and then
we shift its first column to the right via an even number of column swaps.
More precisely, we need to shift the first column of $H_{\UU}(\mbv)$ to
the right so that the values $t_{s,\lambda},t_{s,j_{s,1}},t_{s,j_{s,1}}^\epsilon,
t_{s,j_{s,2}},t_{s,j_{s,2}}^\epsilon,\ldots,t_{s,j_{s,k_s}},\allowbreak t_{s,j_{s,k_s}}^\epsilon$
appear consecutively in the columns of $H_{\UU}(\mbv)$ and in
increasing order. To do that we always need an even number of
column swaps, due to the way we have chosen $\epsilon$.

Consider the case where $s=1$ and suppose that all necessary operations on $H_{\UU}(\mbv)$ have been
performed. Then $H_{\UU}(\mbv)$ is in the form of the determinant
$D_{n,m}(\tau;I,J,\mb{\mu})$ of Lemma \ref{lem:det2}
(multiplied by $\tau^M$),
with $n\sub{}2k_{1}+1$, $m\sub{}2k_{3}$, $l\sub{}d+2$, 
$\mb{\mu}\sub{} (0,0,1,2,\ldots,d)$, $\alpha \sub \nu_{1}$, 
$\beta \sub \nu_{3}$, $I\sub 3$, and $J\sub 5$.
Note that the requirement for $M$ in Lemma \ref{lem:det2}
is satisfied by our choice of $M$.
According to Lemma \ref{lem:det2}, $H_{\UU}(\mbv)$ has the
following asymptotic expansion in terms of $\tau$:
\begin{equation}\label{equ:HU-tau-expansion}
  H_{\UU}(\mbv)
  =\tau^M(C\tau^{\xi}+\Theta(\tau^{\xi+1})),\qquad
  \xi = \nu_1(-2+\sum_{i=4}^{2k_1+3}(i-2))+\nu_3(3+\sum_{i=2k_1+4}^{d+2}(i-2)),
\end{equation}
where $C$ is a positive constant independent of $\tau$.
The asymptotic expansion in \eqref{equ:HU-tau-expansion} implies
that there exists a
positive value $\hat\tau_{\mbv,\UU}$ for $\tau$ such that for all
$\tau\in(0,\hat\tau_{\mbv,\UU})$, $H_{\UU}(\mbv)>0$.
The case $s=3$ is completely analogous.

Since the number of the subsets $\UU$ is finite, while for each
such subset $\UU$ we need to consider a finite number of
vertices in $\VV_{\{1,3\}}\sm{}\UU$, it suffices to consider a positive value
$\hat{\tau}_{\{1,3\}}$ for $\tau$ that is small enough, so that all 
possible determinants $H_{\UU}(\mbv)$ are strictly positive for any
$\tau\in(0,\hat\tau_{\{1,3\}})$.
For $\tau\in(0,\hat\tau_{\{1,3\}})$, our analysis above immediately implies that
\emph{for each} set $\UU$ the equation $H_{\UU}(\mb{x})=0$, $\mb{x}\in\reals^{d+1}$, is the
equation of a supporting hyperplane of $\cC_R$ passing through the 
vertices of $\UU$, and those only. In other words, every
set $\UU$, where $|\UU|=\lexp{d+1},\, |\UU\cap \VV_1|=k_1\neq 0$, and
$|\UU\cap \VV_3|=k_3\neq 0,$ 
defines a $(\lexp{d+1}-1)$-face of $\cC_R$.
Taking into account that the number of such subsets $\UU$ is
$\sum_{i=1}^{\lexp{d+1}-1}\tbinom{n_1}{i}\tbinom{n_3}{\lexp{d+1}-i}$,
we deduce that
\begin{align*}
  f_{\lexp{d+1}-1}(\fF_{\{1,3\}})
  &=\sum_{i=1}^{\lexp{d+1}-1}\tbinom{n_1}{i}\tbinom{n_3}{\lexp{d+1}-i}
  =\tbinom{n_1+n_3}{\lexp{d+1}}-\tbinom{n_1}{\lexp{d+1}}
  -\tbinom{n_3}{\lexp{d+1}}\\
  &=\sum_{\emptyset\subset{}S\subseteq{} \{1,3\}} (-1)^{2-|S|}\tbinom{n_S}{\lexp{d+1}}.
\end{align*}
Hence, condition \eqref{equ:bineighborliness-condition} is satisfied
for all $\tau\in(0,\hat\tau_{\{1,3\}})$.
\end{proof}

We now consider the case $R=[3]$. In this case we can show that:
%
\begin{lemma}\label{lem:fkW3-tau}
  There exists a sufficiently small positive value $\hat\tau_{[3]}$ for
  $\tau$ such that, for all $\tau\in(0,\hat\tau_{[3]})$,
  \begin{equation}
    f_{k-1}(\fF_{[3]})=\sum_{\emptyset\subset{}S\subseteq{}[3]}
    (-1)^{3-|S|}\tbinom{n_S}{k},  \qquad 3\le{}k\le\ltexp{d+2},
  \end{equation}
\end{lemma}
\begin{proof}
Define $\delta:=d+2-2k$ and let $T$ be a positive real number such that
$T>t_{3,n_{3}}^\epsilon(=x_{3,n_3}^\epsilon)$. 
Choose a set $U$ of $k_i\neq 0$ vertices from $V_i$, $1\leq i\leq 3$, 
such that $k_1+k_2+k_3=k$, and denote by $\UU$ the Cayley embedding of $U$
in $\reals^{d+2}$ (using the affine basis $\me_{2,i}$, $1\le i\le3$).
Let $\mc_i(t_{i,j_{i,1}})$,$\mc(t_{i,j_{i,2}})$,$\ldots$, $\mc_i(t_{i,j_{i,k_i}})$,
be the vertices in $U$, and
$\mcc_i(t_{i,j_{i,1}}),\mcc_i(t_{i,j_{i,2}}),\ldots,\mcc_i(t_{i,j_{i,k_i}})$,
be their corresponding vertices in $\UU$,
where $j_{i,1}<j_{i,2}<\ldots<j_{i,k_i}$ for all $1\le{}i\le{}3$.
Let $\mb{x}=(x_1,x_2,\ldots,x_{d+2})$ and define the
$(d+3)\times(d+3)$ determinant $H_{\UU}(\mb{x})$ as
follows: 
\begin{center}
 \includegraphics[width=1.\textwidth]{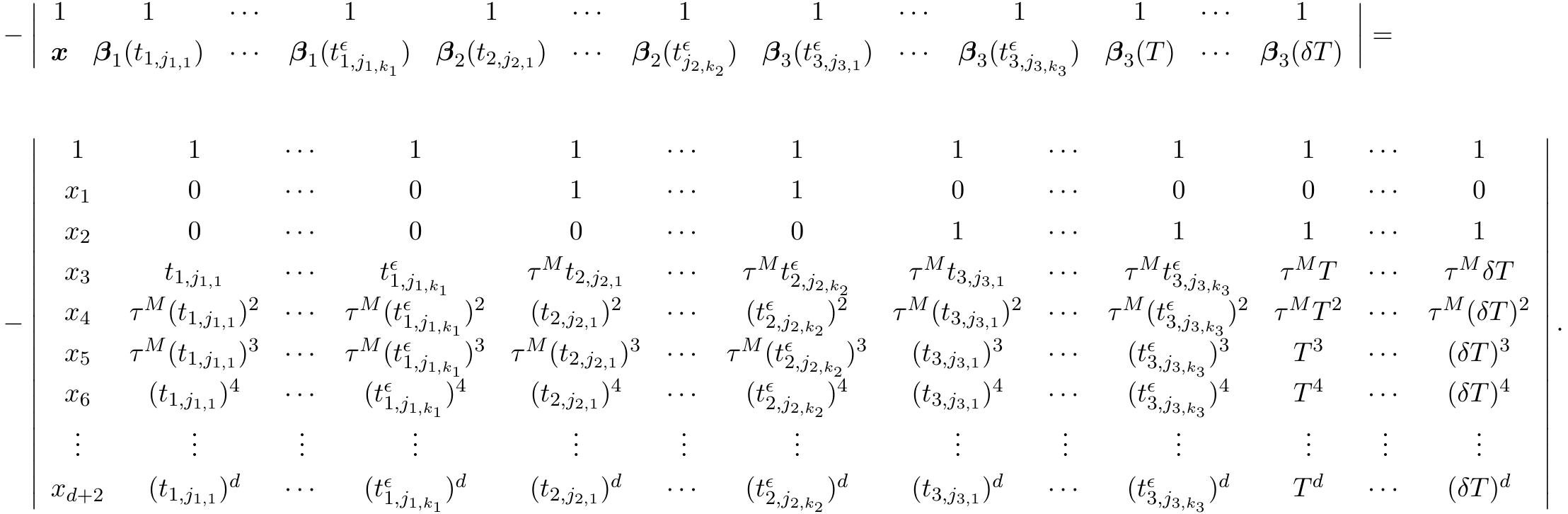}
\end{center}
We can alternatively describe $H_{\UU}(\mb{x})$ as follows:
  \begin{enumerate}
  \item
    The first column of $H_{\UU}(\mb{x})$ is $\binom{1}{\mb{x}}$.
  \item
    For $i$ ranging from $1$ to $3$, and for $\lambda$ ranging from $1$ to
    $k_i$, the next $k_i$ \emph{pairs of columns} of
    $H_{\UU}(\mb{x})$ are $\binom{1}{\mcc_i(t_{i,j_{i,\lambda}})}$ and
    $\binom{1}{\mcc_i(t_{i,j_{i,\lambda}}^\epsilon)}$.
  \item
    For $\lambda$ ranging from $1$ to $\delta$,
    the last $\delta$ columns of $H_{\UU}(\mb{x})$ are
    $\binom{1}{\mcc_3(\lambda T)}$. Notice that if $k=\lexp{d+2}$ and 
    $d$ is even, this category of
    columns of $H_{\UU}(\mb{x})$ does not exist.
  \end{enumerate}

The equation $H_{\UU}(\mb{x})=0$ is the equation of a
hyperplane in $\reals^{d+2}$ that passes through the points in
$\UU$. Recall that $\VV_{[3]}=\VV_1\cup \VV_2 \cup \VV_3$. 
We are going to show that, for any choice of $\UU$,
and for all vertices $\mbv$ in $\VV_{[3]}\sm{}\UU$, 
we have $H_{\UU}(\mbv)>0$ for sufficiently small $\tau$.
  
Suppose we have some vertex $\mbv\in{}\VV_{[3]}\sm{}\UU$. 
Then, $\mbv=\mcc_s(t_{s,\lambda})$,
$t_{s,\lambda}=x_{s,\lambda}\tau^{\nu_s}$, for some
$1\leq\lambda\leq n_s$ and $1\leq s \leq 3$, such that 
$\lambda \notin\{j_{s,1},j_{s,2},\ldots,j_{s,k_s}\}$.
Then we can transform $H_{\UU}(\mbv)$ in the form of the determinant 
$E_{n,m,k}(\tau;\mb{\mu})$ 
of Lemma \ref{lem:det3}, by subtracting the second and third 
row of $H_{\UU}(\mbv)$ from its first row and 
shifting the first column of $H_{\UU}(\mbv)$ to the right via
an even number of column swaps.
More precisely, we need to shift the first column of $H_{\UU}(\mbv)$ to
the right so that the values $t_{s,\lambda},t_{s,j_{s,1}},t_{s,j_{s,1}}^\epsilon,
t_{s,j_{s,2}},t_{s,j_{s,2}}^\epsilon,\ldots,t_{s,j_{s,k_s}},t_{s,j_{s,k_s}}^\epsilon$,
appear consecutively in the columns of $H_{\UU}(\mbv)$ and in
increasing order. To do that we always need an even number of
column swaps, due to the way we have chosen $\epsilon$.

Now, suppose that $\mbv \in \VV_1$.
Then $H_{\UU}(\mbv)$ is in the form of the determinant 
 $E_{n,m,k}(\tau;\mb{\mu})$ of Lemma \ref{lem:det3},
where $n\sub{}2k_1+1$, $m\sub{}2k_2$, $k\sub{}2k_3+\delta$, 
$l\sub{}d+3$, and $\mb{\mu}\sub{} (0,0,0,1,2,\ldots,d)$.
Obviously, $M\geq 2|\mb{\mu}|=d(d+1)$.
Applying now Lemma \ref{lem:det3}, we deduce that $H_{\UU}(\mbv)$ can be
written as:
\begin{equation*}
  H_{\UU}(\mbv)
 =C'\tau^{\xi}+
  \Theta(\tau^{\xi+1}),\qquad
  \xi = 4+2\sum_{i=7}^{2k_1+5}(i-3)+\sum_{i=2k_1+6}^{2k_1+2k_2+3}(i-3),
\end{equation*}
where $C'$ is a positive constant independent of $\tau$. The
asymptotic estimate above implies that $H_{\UU}(\mbv)>0$, for sufficiently
small $\tau$.

The remaining cases, \ie the cases $\mbv \in \VV_2$ and $\mbv \in \VV_3$,
are completely analogous and we omit them.
We thus conclude that, for any specific choice of $U$, and for
any specific vertex $\mbv\in{}\VV_{[3]}\sm{}\UU$,
there exists some $\tau_{\mbv,\UU}>0$ (cf. Lemma \ref{lem:det3})
that depends on $\mbv$ and $\UU$, such that for all $\tau\in
(0,\tau_{\mbv,\UU})$ we have $H_{\UU}(\mbv)>0$.
For each $k$ with $3\le{}k\le{}\lexp{d+2}$, the number
of the sets $\UU$ of size $k$ containing at least one vertex from each
$\VV_i$, $1\leq i \leq 3$, is
\begin{equation*}
  \tbinom{n_1+n_2+n_3}{k}-\tbinom{n_1+n_2}{k}-\tbinom{n_1+n_3}{k}-
  \tbinom{n_2+n_3}{k} +\tbinom{n_1}{k}+\tbinom{n_2}{k}+\tbinom{n_3}{k}=
  \sum_{\emptyset\subset{}S\subseteq{}[3]} (-1)^{3-|S|}\tbinom{n_S}{k}.
\end{equation*}
For each such subset $\UU$ we need
to consider the $(n_1+n_2+n_3-k)$ vertices in $\VV_{[3]}\sm{}\UU$, therefore it
suffices to consider a positive value $\hat{\tau}_{[3]}$ for $\tau$ that
is small enough, so that all
  \begin{equation*}
    \sum_{k=2}^{\lexp{d+1}}(n_1+n_2+n_3-k)\,\sum_{\emptyset\subset{}S\subseteq{}[3]} (-1)^{3-|S|}\tbinom{n_S}{k},
  \end{equation*}
possible determinants $H_{\UU}(\mbv)$ are strictly positive. 
For $\tau\sub\hat\tau_{[3]}$, our analysis above immediately implies that
\emph{for each} set $\UU$ the equation $H_{\UU}(\mb{x})=0$, $\mb{x}\in\reals^{d+2}$, is the
equation of a supporting hyperplane for $\cC$ passing through the 
vertices of $\UU$, and those only. In other words, every
set $\UU$, of $k$ vertices, for $3\le{}k\le{}\lexp{d+2}$,
with at least one vertex from each $\VV_i$, $1\leq i \leq 3$,
defines a $(k-1)$-face of $\cC$, which means that
\begin{equation*}
 f_{k-1}(\fF_{[3]})=\sum_{\emptyset\subset{}S\subseteq{}R} (-1)^{3-|S|}
 \tbinom{n_S}{k},\quad
\mbox{for all}~ 3\le{}k\le{}\ltexp{d+2}.\qedhere
\end{equation*}
\end{proof}

Relation \eqref{equ:fkWR} now immediately follows from Lemmas
\ref{lem:fkW2-tau} and \ref{lem:fkW3-tau}.
First choose a value $\tau^\star$ for $\tau$, smaller that
$\hat\tau_R$, for all $2\le |R|\le 3$. Then for this
value of $\tau$, the results of both Lemma \ref{lem:fkW2-tau} and
Lemma \ref{lem:fkW3-tau} hold true. Moreover, since $P_1$, $P_2$ and
$P_3$ are $\ltexp{d}$-neighborly for any $\tau>0$, and since 
$f_{-1}(\fF_R)=(-1)^{|R|-1}$, for all $\emptyset\subset R\subseteq [3]$,
while $f_{k-1}(\fF_R)=0$, for all $1\le k\le |R|$,
we conclude that, for
$\tau\equiv\tau^\star$, relations \eqref{equ:fkWR} hold.

Based on the analysis above, as well as the analysis in Section
\ref{sec:lb:3d}, we conclude that the upper bounds stated in
Theorem \ref{thm:ms3ub} are actually tight for any $d\ge{}2$.
We can, thus, restate Theorem \ref{thm:ms3ub} in its complete and
definitive form:

\begin{theorem}\label{thm:final-result}
  Let $P_1$, $P_2$ and $P_3$ be three $d$-polytopes in $\reals^d$,
  $d\ge{}2$, with $n_i\ge{}d+1$ vertices, $1\le{}i\le{}3$. Then, for all
  $1\le{}k\le{}d$, we have:
  \begin{equation*}
    \begin{aligned}
      f_{k-1}(P_1+P_2+P_3)&\le
      f_{k+1}(C_{d+2}(n_{[3]}))-\sum_{i=0}^{\lexp{d+2}}\binom{d+2-i}{k+2-i}
      \sum_{\emptyset\subset{}S\subset[3]}(-1)^{|S|}\binom{n_S-d-3+i}{i}\\
      &\quad-\delta\binom{\lexp{d}+1}{k-\lexp{d}}
      \sum_{i=1}^3\binom{n_i-\lexp{d}-2}{\lexp{d}+1},
    \end{aligned}
  \end{equation*}
  where $\delta=d-2\lexp{d}$, and $n_S=\sum_{i\in{}S}n_i$.
  Moreover, for any $d\ge{}2$, there exist three $d$-polytopes in $\reals^d$
  for which the bounds above are attained for all $1\le{}k\le{}d$.
\end{theorem}
\section{Summary and  open problems}
\label{sec:concl}

In this paper we have computed the maximum number of
$k$-faces, $f_k(P_1+P_2+P_3)$, $0\le{}k\le{}d-1$, of the Minkowski sum
of three $d$-polytopes $P_1,P_2$ and $P_3$ in $\reals^d$ as a function
of the number of their  vertices $n_1,n_2$ and $n_3.$ 
When $d=2$ our expressions reduce to known tight  bounds, while
for $d=3$ we show the tightness of our upper bounds by
exploiting results from \cite{fw-fmacp-07} and \cite{w-mfmsl-12}. 
In four or more dimensions we present a novel construction that
achieves the upper bounds:
we consider the $d$-dimensional moment-like curves 
$\mc_1(t)=(t,\z t^2,\z t^3,t^4,\ldots,t^{d})$,
$\mc_2(t)=(\z t,t^2,\z t^3,t^4,\ldots,t^{d})$, and 
$\mc_3(t)=(\z t,\z t^2,t^3,t^4,\ldots,t^{d})$, and we show that our 
maximal values are attained when $P_i$ is the $d$-polytope with
vertex set
\begin{equation*}
  V_i=\{\mc_i( x_{i,1}\tau^\star),\mc_i(x_{i,2}\tau^\star),\ldots,
  \mc_i(x_{i,n_i}\tau^\star)\},
  \qquad i=1,2,3,
\end{equation*}
with $0<x_{i,1}<x_{i,2}<\cdots<x_{i,n_i}$ and $\z=(\tau^\star)^M$. The
parameter value $\tau^\star$ is a sufficiently small positive number,
while $M$ is chosen sufficiently large.

Our ultimate goal is to extend our results for the Minkowski
sum of $r$ $d$-polytopes in $\mathbb R^d$, for $r\ge{} 4$ and $d\ge{}3$. 
Towards this direction, we can extend our methodology and tools so as
to prove relations for  $r$ polytopes that generalize 
certain relations that hold true for two or three polytopes. 
For example, relation~\eqref{equ:hkKR} in
Lemma~\ref{lem:hkK-hkQ-dehn-sommervile} generalizes to:
\begin{equation*}
  h_k(\kK_R)=\sum_{\emptyset\subset{}S\subseteq{}R}g_{k}^{(|R|-|S|)}(\fF_S),
  \qquad 0\le{}k\le{}d+|R|-1,
\end{equation*}
while the Dehn-Sommerville-like equations in the same lemma
(cf.~rel. \eqref{equ:DSW}),
generalize to:
\begin{equation}\label{eq:generalDS}
  h_{d+r-1-k}(\fF_{[r]})=h_k(\kK_{[r]}), \qquad 0\le{}k\le{}d+r-1,
\end{equation}
where $[r]=\{1,2,\ldots,r\}$, while $\fF_R$ and $\kK_R$,
  $\emptyset \subset R \subseteq [r]$, are defined
as in Section~\ref{sec:cayley}.
Notice that, since for $r=1$ we have
$\fF_{[1]}\equiv{}\kK_{[1]}\equiv\bx{}P_1$,
the equations in \eqref{eq:generalDS} reduce to the well-known 
Dehn-Sommerville equations for a simplicial $d$-polytope.
We can also obtain a generalization of relation \eqref{r4}.
Let $\qQ$ be the simplicial $(d+r-1)$-sphere we get by 
performing stellar subdivisions on the non-simplicial
faces of the Cayley polytope of the $r$ polytopes.
For all $0\le{}k\le{}d+r-1$, we can obtain the following two
expressions relating the $h$-vector elements of $\bx\qQ$ with those
of $\fF_S$ and $\kK_S$, $\emptyset{}\subset{}S\subseteq[r]$: 
\begin{align*}
  h_k(\bx{}Q)&=h_k(\fF_{[r]})+\sum_{\emptyset{}\subset{}S\subset{}[r]}
    \sum_{i=0}^{r-|S|-1}E_{r-|S|,i}\,h_{k-i}(\fF_S),\\
  h_k(\bx{}Q)&=h_k(\fF_{[r]})+\sum_{\emptyset{}\subset{}S\subset{}[r]}
  \sum_{i=0}^{r-|S|-1}E_{r-|S|,i}\,h_{k-1-i}(\kK_S),
\end{align*}
where $E_{m,k}$, $m\ge{}k+1>0$, are the Eulerian numbers
\cite{gkp-cm-89,oeis:A008292}:
\begin{equation*}
  E_{m,k}=\sum_{i=0}^{k}(-1)^{i}\binom{m+1}{i}(k+1-i)^m,
  \qquad m\ge{}k+1>0.
\end{equation*}
A recurrence relation similar to \eqref{r8} in Lemma~\ref{lem:hkWrecur}
is not as straightforward to obtain. However, we conjecture that
the following recurrence relation holds for all $0\le{}k\le{}d+r-2$:
\begin{equation*}
  h_{k+1}(\fF_{[r]})\le\frac{n_{[r]}-d-r+1+k}{k+1}\,h_k(\fF_{[r]})
  +\sum_{i=1}^{r}\frac{n_i}{k+1}\,g_k(\fF_{[r]\sm\{i\}}),
  \qquad n_{[r]}=\sum_{i=1}^{r}n_i.
\end{equation*}

The bounds presented in this paper refer to polytopes of the same
dimension. We would like to derive similar bounds for two or
more polytopes when the dimensions of these polytopes differ,
as well as in the special case of simple polytopes.
Finally, a similar problem is to express the number of $k$-faces of the
Minkowski sum of $r$ $d$-polytopes in terms of the number of facets of
these polytopes. Results in this direction are known for
$d=2$ and $d=3$ only.  We would like to
derive such expressions for any $d\ge{}4$ and any number, $r$, of
summands.

\section*{Acknowledgements}
The work in this paper has been partially supported by the
FP7-REGPOT-2009-1 project ``Ar\-chi\-me\-des Center for Modeling, Analysis
and Computation'' (under grant agreement n\textsuperscript{o} 245749),
and has been co-financed by the European Union (European
Social Fund -- ESF) and Greek national funds through the Operational
Program ``Education and Lifelong Learning'' of the National Strategic
Reference Framework (NSRF) -- Research Funding Program:
THALIS -- UOA (MIS 375891).

\phantomsection
\addcontentsline{toc}{section}{References}

\bibliographystyle{alpha}
\bibliography{msmax3p}

\begin{thebibliography}{dBvKOS00}

\bibitem[A00]{oeis:A008292}
The integer sequence {A008292} ({E}ulerian numbers).
\newblock The On-Line Encyclopedia of Integer Sequences.
\newblock \\{\scriptsize\url{http://oeis.org/A008292}}.

\bibitem[BM71]{bm-sdcs-71}
H.~Bruggesser and P.~Mani.
\newblock Shellable decompositions of cells and spheres.
\newblock {\em Math. Scand.}, 29:197--205, 1971.
\newblock \\{\scriptsize\url{http://www.mscand.dk/article.php?id=2034}}.

\bibitem[CLO05]{CLO2}
David~A. Cox, John Little, and Donal O'Shea.
\newblock {\em Using Algebraic Geometry}, volume 185 of {\em Graduate Texts in
  Mathematics}.
\newblock Springer, New York, 2nd edition, 2005.

\bibitem[dBvKOS00]{bkos-cgaa-00}
Mark de~Berg, Marc van Kreveld, Mark Overmars, and Otfried Schwarzkopf.
\newblock {\em Computational Geometry: Algorithms and Applications}.
\newblock Springer-Verlag, Berlin, Germany, 2nd edition, 2000.

\bibitem[ES74]{EwSh74}
G.~Ewald and G.~C. Shephard.
\newblock Stellar {S}ubdivisions of {B}oundary {C}omplexes of {C}onvex
  {P}olytopes.
\newblock {\em Mathematische Annalen}, 210:7--16, 1974.
\newblock \\{\scriptsize\url{http://dx.doi.org/10.1007/BF01344542}}.

\bibitem[FHW09]{fhw-emcms-09}
Efi Fogel, Dan Halperin, and Christophe Weibel.
\newblock On the {E}xact {M}aximum {C}omplexity of {M}inkowski {S}ums of
  {P}olytopes.
\newblock {\em Discrete Comput. Geom.}, 42:654--669, 2009.
\newblock \\{\scriptsize\url{http://dx.doi.org/10.1007/s00454-009-9159-1}}.

\bibitem[Fog08]{f-mscaa-08}
Efraim Fogel.
\newblock {\em Minkowski Sum Construction and other Applications of
  Arrangements of Geodesic Arcs on the Sphere}.
\newblock PhD thesis, Tel-Aviv University, October 2008.

\bibitem[Fuk04]{f-zcmacp-04}
Komei Fukuda.
\newblock From the zonotope construction to the {M}inkowski addition of convex
  polytopes.
\newblock {\em J. Symb. Comput.}, 38:1261--1272, 2004.
\newblock \\{\scriptsize\url{http://dx.doi.org/10.1016/j.jsc.2003.08.007}}.

\bibitem[FW07]{fw-fmacp-07}
Komei Fukuda and Christophe Weibel.
\newblock $f$-vectors of {M}inkowski {A}dditions of {C}onvex {P}olytopes.
\newblock {\em Discrete Comput. Geom.}, 37(4):503--516, 2007.
\newblock \\{\scriptsize\url{http://dx.doi.org/10.1007/s00454-007-1310-2}}.

\bibitem[Gan60]{g-tm-60}
F.~R. Gantmacher.
\newblock {\em The Theory of Matrices}, volume~I.
\newblock Chelsea Publishing Co., New York, 1960.

\bibitem[Gan05]{g-atm-05}
F.~R. Gantmacher.
\newblock {\em Applications of the Theory of Matrices}.
\newblock Dover, Mineola, New York, 2005.

\bibitem[GKP89]{gkp-cm-89}
R.~L. Graham, D.~E. Knuth, and O.~Patashnik.
\newblock {\em Concrete Mathematics}.
\newblock Addison-Wesley, Reading, MA, 1989.

\bibitem[GS93]{gs-mapcc-93}
Peter Gritzmann and Bernd Sturmfels.
\newblock Minkowski {A}ddition of {P}olytopes: Computational {C}omplexity and
  {A}pplications to {G}r\"obner bases.
\newblock {\em SIAM J. Disc. Math.}, 6(2):246--269, May 1993.
\newblock \\{\scriptsize\url{http://dx.doi.org/10.1137/0406019}}.

\bibitem[HK71]{hk-la-71}
Kenneth Hoffman and Ray Kunze.
\newblock {\em Linear Algebra}.
\newblock Prentice Hall, 2nd edition, 1971.

\bibitem[HRS00]{hrs-ctlsb-00}
Birkett Huber, J\"org Rambau, and Francisco Santos.
\newblock The {C}ayley {T}rick, lifting subdivisions and the {B}ohne-{D}ress
  theorem on zonotopal tilings.
\newblock {\em J. Eur. Math. Soc.}, 2(2):179--198, June 2000.
\newblock \\{\scriptsize\url{http://dx.doi.org/10.1007/s100970050003}}.

\bibitem[KT11]{kt-tlbnf-11}
Menelaos~I. Karavelas and Eleni Tzanaki.
\newblock Tight lower bounds on the number of faces of the {M}inkowski sum of
  convex polytopes via the {C}ayley trick, December 2011.
\newblock \href{http://arxiv.org/abs/1112.1535}{{\tt arXiv:1112.1535 [cs.CG]}}.

\bibitem[KT12]{kt-mnfms-12}
Menelaos~I. Karavelas and Eleni Tzanaki.
\newblock The maximum number of faces of the {M}inkowski sum of two convex
  polytopes.
\newblock In {\em Proceedings of the 23rd ACM-SIAM Symposium on Discrete
  Algorithms (SODA'12)}, pages 11--28, Kyoto, Japan, January 17--19, 2012.
\newblock \\{\scriptsize\url{http://doi.acm.org/10.1145/2095116.2095118}}.

\bibitem[Lat91]{Lat91}
Jean-Claude Latombe.
\newblock {\em Robot Motion Planning}.
\newblock Kluwer Academic Publishers, Norwell, Massachusetts, USA, 1991.

\bibitem[LM04]{lm-cpq-03}
Ming~C. Lin and Dinesh Manocha.
\newblock Collision and proximity queries.
\newblock In Jacob~E. Goodman and Joseph O'Rourke, editors, {\em Handbook of
  Discrete and Computational Geometry}, chapter~35, pages 787--808. CRC Press,
  Boca Raton, Florida, 2nd edition, 2004.

\bibitem[Mat02]{Mat02}
Jiri Matousek.
\newblock {\em Lectures on Discrete Geometry}.
\newblock Graduate Texts in Mathematics. Springer-Verlag New York, Inc., New
  York, 2002.

\bibitem[McM70]{m-mnfcp-70}
P.~McMullen.
\newblock The maximum numbers of faces of a convex polytope.
\newblock {\em Mathematika}, 17:179--184, 1970.
\newblock \\{\scriptsize\url{http://dx.doi.org/10.1112/S0025579300002850}}.

\bibitem[San09]{s-tovnm-09}
Raman Sanyal.
\newblock Topological obstructions for vertex numbers of {M}inkowski sums.
\newblock {\em J. Comb. Theory, Ser. A}, 116(1):168--179, 2009.
\newblock \\{\scriptsize\url{http://dx.doi.org/10.1016/j.jcta.2008.05.009}}.

\bibitem[Stu96]{s-gbcp-96}
Bernd Sturmfels.
\newblock {\em Gr\"obner {B}ases and {C}onvex {P}olytopes}, volume~8 of {\em
  Univ. Lectures Series}.
\newblock American Mathematical Society, Providence, Rhode Island, 1996.

\bibitem[TRH00]{trhh-smcpma-00}
Alexander~V. Tuzikov, Jos~B.T.M. Roerdink, and Henk~J.A.M. Heijmans.
\newblock Similarity measures for convex polyhedra based on {M}inkowski
  addition.
\newblock {\em Pattern Recognition}, 33(6):979--995, 2000.
\newblock
  \\{\footnotesize\url{http://dx.doi.org/10.1016/S0031-3203(99)00159-4}}.

\bibitem[Wei07]{w-mspcc-07}
Christophe Weibel.
\newblock {\em Minkowski Sums of Polytopes: Combinatorics and Computation}.
\newblock PhD thesis, \'Ecole Polytechnique F\'ed\'erale de Lausanne, 2007.

\bibitem[Wei12]{w-mfmsl-12}
Christophe Weibel.
\newblock Maximal f-vectors of {M}inkowski {S}ums of {L}arge {N}umbers of
  {P}olytopes.
\newblock {\em Discrete Comput. Geom.}, 47(3):519--537, April 2012.
\newblock \\{\scriptsize\url{http://dx.doi.org/10.1007/s00454-011-9385-1}}.

\bibitem[Zie95]{z-lp-95}
G\"unter~M. Ziegler.
\newblock {\em Lectures on Polytopes}, volume 152 of {\em Graduate Texts in
  Mathematics}.
\newblock Springer-Verlag, New York, 1995.

\end{thebibliography}

\phantomsection
\addcontentsline{toc}{section}{Appendix}
\appendix
\allowdisplaybreaks

\newcommand{\proofseparator}{\smallskip\hrule\smallskip}

\section{Omitted \& full proofs}
\label{app:omitted}

\subsection{Omitted \& full proofs of Section~\ref{sec:ds}}
\label{app:ds}

\begin{proof}[Proof of Lemma~\ref{lem:gvec_as_sum}]
  The result clearly holds for $m=0$, since:
  \begin{equation*}
    g_k^{(0)}(\yY)=h_k(\yY)=\sum_{i=0}^{0}(-1)^i\tbinom{0}{i}h_{k-i}(\yY).
  \end{equation*}
  Suppose the relation holds for some $m\ge{}0$. We will show it
  holds for $m+1$. Indeed:
  \begin{align*}
    g_k^{(m+1)}(\yY)&=g_k^{(m)}(\yY)-g_{k-1}^{(m)}(\yY)\\
    &=\sum_{i=0}^m(-1)^i\tbinom{m}{i}h_{k-i}(\yY)
    -\sum_{i=0}^m(-1)^i\tbinom{m}{i}h_{k-1-i}(\yY)\\
    &=\sum_{i=0}^{m+1}(-1)^i\tbinom{m}{i}h_{k-i}(\yY)
    -\sum_{j=1}^{m+1}(-1)^{j-1}\tbinom{m}{j-1}h_{k-j}(\yY)\\
    &=\sum_{i=0}^{m+1}(-1)^i\tbinom{m}{i}h_{k-i}(\yY)
    -\sum_{j=0}^{m+1}(-1)^{j-1}\tbinom{m}{j-1}h_{k-j}(\yY)\\
    &=\sum_{i=0}^{m+1}(-1)^i\tbinom{m}{i}h_{k-i}(\yY)
    +\sum_{i=0}^{m+1}(-1)^i\tbinom{m}{i-1}h_{k-i}(\yY)\\
    &=\sum_{i=0}^{m+1}(-1)^i\left[\tbinom{m}{i}+\tbinom{m}{i-1}\right]
    h_{k-i}(\yY)\\
    &=\sum_{i=0}^{m+1}(-1)^i\tbinom{m+1}{i}h_{k-i}(\yY).\qedhere
  \end{align*}
\end{proof}

\proofseparator

\begin{proof}[Proof of Lemma~\ref{lem:sumoperator_gvec}]
  By replacing $h_{k-\nu-j}(\yY)$ from its defining equation, we get:
  \begin{align}
    g^{(D-\delta-\nu)}_{k-\nu}(\yY)
    &=\sum_{j=0}^{D-\delta-\nu}(-1)^j\tbinom{D-\delta-\nu}{j}h_{k-\nu-j}(\yY)
    \notag\\
    &=\sum_{j=0}^{D-\delta-\nu} (-1)^j\tbinom{D-\delta-\nu}{j}
    \sum_{i=0}^{\delta+1}(-1)^{k-\nu-j-i}\tbinom{\delta+1-i}{\delta+1-k+\nu+j}
    f_{i-1}(\yY)\label{rel1}\\
    &=\sum_{j=0}^{D-\delta-\nu} (-1)^j\tbinom{D-\delta-\nu}{j}
    \sum_{i=0}^{D+1}(-1)^{k-\nu-j-i}\tbinom{\delta+1-i}{\delta+1-k+\nu+j}
    f_{i-1}(\yY)\label{rel2}\\
    &=\sum_{i=0}^{D+1}(-1)^{k-\nu-i}f_{i-1}(\yY)\sum_{j=0}^{D-\delta-\nu}
    \tbinom{D-\delta-\nu}{j}\tbinom{\delta+1-i}{k-\nu-i-j}\label{rel3}\\
    &=\sum_{i=0}^{D+1}(-1)^{k-\nu-i}\tbinom{D+1-\nu-i}{k-\nu-i}f_{i-1}(\yY)
    \label{rel4}\\
    &=\sum_{i=0}^{k-\nu} (-1)^{k-\nu-i}\tbinom{D+1-\nu-i}{D+1-k}
    f_{i-1}(\yY)
    \label{rel5}\\
    &=\sum_{j=\nu}^{k} (-1)^{k-j}\tbinom{D+1-j}{D+1-k}f_{j-\nu-1}(\yY)
    \label{rel6}\\
    &=\sum_{j=0}^{D+1}(-1)^{k-j}\tbinom{D+1-j}{D+1-k}f_{j-\nu-1}(\yY)
    \label{rel7}\\
    &=\SSS_k(\yY; D,\nu),\notag
  \end{align}
  where:
  \begin{itemize}
  \item
    in order to go from \eqref{rel1} to \eqref{rel2}, we used that
    $\binom{\delta+1-i}{\delta+1-k+\nu+j}=0$ for 
    $i>\delta+1$,
  \item
    in order to go from \eqref{rel3} to \eqref{rel4}, we used the
    combinatorial identity:
    \begin{equation*}
      \sum_{i=0}^n\tbinom{n}{i}\tbinom{m}{k-i}=\tbinom{n+m}{k}
      =\sum_{i=0}^k\tbinom{n}{i}\tbinom{m}{k-i}=\tbinom{n+m}{k},
    \end{equation*}
    with $n\sub{}D-\delta-\nu$, $m\sub{}\delta+1-i$, $i\sub{}j$,
    $k\sub{}k-\nu-i$,
  \item
    in order to go from \eqref{rel4} to \eqref{rel5}, we used that
    $\binom{D+1-\nu-i}{k-\nu-i}=0$ for $i>k-\nu$, and that
    $\binom{D+1-\nu-i}{k-\nu-i}=\binom{D+1-\nu-i}{(D+1-\nu-i)-(k-\nu-i)}
    =\binom{D+1-\nu-i}{D+1-k}$, and, finally,
  \item
    in order to go from \eqref{rel6} to \eqref{rel7}, we used that
    $f_{j-\nu-1}(\yY)=0$ for $j<\nu$ (\ie for $j-\nu-1<-1$), and that
    $\binom{D+1-j}{D+1-k}=0$ for $j>k$.\qedhere
  \end{itemize}
\end{proof}

\subsection{Omitted \& full proofs of Section~\ref{sec:recur}}
\label{app:recur}

\begin{proof}[Proof of Lemma~\ref{lem:recur-relation-F3-wrt-K}]
Using relation \eqref{r4}, and after rearranging the terms, the left
hand side of relation \eqref{equ:linksQ} becomes:
\begin{equation}\label{equ:full-expansion}
  \begin{aligned}
    &\overbrace{(k+1) h_{k+1}(\fF_{[3]})+ (d+2-k)  h_k(\fF_{[3]})}^{T_1}
    +\overbrace{\sum_{R\in\sset_2}[(k+1) h_{k+1}(\fF_R)+(d+2-k) h_k(\fF_R)]}^{T_2}\\
    &+\overbrace{\sum_{R\in\sset_1}[(k+1) h_{k+1}(\fF_R)+(d+2-k) h_k(\fF_R)]}^{T_3}
    +\overbrace{\sum_{R\in\sset_1}[(k+1) h_{k}(\fF_R)+(d+2-k) h_{k-1}(\fF_R)]}^{T_4}.
  \end{aligned}
\end{equation}
We are going to analyze each term in the expression above separately.
For any $R\in\sset_2$:
(i) the relation at the top of page 18 in \cite[Lemma 3.2]{kt-mnfms-12},
(ii) relations \eqref{equ:hkKR}, with $R\in\sset_2$, and
(iii) relation (3.9) in \cite{kt-mnfms-12},
give: 
\begin{align*}
  \nonumber (k+1) h_{k+1}(\fF_R) + (d+1-k) h_k(\fF_R) &= 
  \sum_{i\in{}R}\sum_{v\in \VV_i} [ h_{k}(\kK_{R} / v) - g_k(\kK_{\{i\}} /v)]\\
  &=\sum_{v\in\VV_R} h_{k}(\kK_{R} / v)
  - \sum_{\emptyset \subset S \subset R} \sum_{v\in \VV_S}
  g_k(\kK_{S} /v).
\end{align*}
Hence term $T_2$ can be rewritten as:
\begin{equation}\label{r34b-ms2p}
  \begin{aligned}
    T_2&=\sum_{R\in\sset_2}h_k(\fF_R)
    +\sum_{R\in\sset_2}\sum_{v\in\VV_R}h_{k}(\kK_{R}/v)
    -\sum_{R\in\sset_2}\sum_{\emptyset\subset{}S\subset{}R}\sum_{v\in \VV_S}
    g_k(\kK_{S} /v)\\
    &=\overbrace{\sum_{R\in\sset_2}h_k(\fF_R)}^{T_5}
    +\overbrace{\sum_{R\in\sset_2}\sum_{v\in\VV_R}h_{k}(\kK_{R}/v)}^{T_6}
    -\overbrace{2\sum_{R\in\sset_1}\sum_{v\in \VV_R}g_k(\kK_{R}/v)}^{T_7}.
  \end{aligned}
\end{equation}

Applying relation \eqref{equ:links} to the $(d-1)$-complex
$\fF_R$, $R\in \sset_1$, and using the identity $\fF_R \equiv \kK_R (\equiv
\bx{}P_R)$, we derive the following expressions:
\begin{align*}
  (k+1) h_{k+1}(\fF_R) + (d-k) h_k(\fF_R)
  &=\sum_{v\in \VV_R} h_k (\kK_R / v),\\
  k h_{k}(\fF_R) + (d-(k-1)) h_{k-1}(\fF_R)
  &= \sum_{v\in \VV_R} h_{k-1} (\kK_R / v),
\end{align*}
which, in turn yield the following expansions for $T_3$ and $T_4$:
\begin{align}
  \label{r5left2} 
  T_3&=\overbrace{\sum_{R\in\sset_1} \sum_{v\in \VV_R} h_k (\kK_R / v)}^{T_8}
  +\overbrace{2\sum_{R\in\sset_1}h_k(\fF_R)}^{T_9},\\[4pt]
  \label{r5left3} 
  T_4&= \overbrace{\sum_{R\in\sset_1} \sum_{v\in \VV_R} h_{k-1} (\kK_R / v)}^{T_{10}}
  +\overbrace{\sum_{R\in\sset_1}[h_{k}(\fF_R)+h_{k-1}(\fF_R)]}^{T_{11}}.
\end{align}

On the other hand, utilizing the expressions in Lemma
\ref{lem:h-vectors-Q-links}, we 
arrive at the following expansion for the right-hand side of
\eqref{equ:linksQ}:
\begin{equation}\label{equ:h-vectors-Q-links-RHS}
  \begin{aligned}
    &\sum_{i=1}^{3}\sum_{v\in\VV_i}\left[h_k(\kK_{[3]}/v)
      +\sum_{\{i\}\subseteq{}R\subset[3]}h_{k-1}(\kK_{R}/v)
      +h_{k-2}(\kK_{\{i\}}/v)\right]\\
    &+\sum_{R\in\sset_1}[h_k(\fF_R) + h_{k-1}(\fF_R)]
    +\sum_{R\in\sset_2}\sum_{\emptyset\subset{}S\subseteq{}R}h_{k}(\fF_S).
  \end{aligned}
\end{equation}
Since 
\begin{align*}
  \sum_{i=1}^{3}\sum_{v\in\VV_i}h_k(\kK_{[3]}/v)
  &=\sum_{v\in\VV_{[3]}}h_k(\kK_{[3]}/v),\\
  \sum_{i=1}^{3}\sum_{v\in\VV_i}\sum_{\{i\}\subseteq{}R\subset[3]}h_{k-1}(\kK_{R}/v)
  &=\sum_{R\in\sset_2}\sum_{v\in\VV_R}h_{k-1}(\kK_{R}/v)
  +\sum_{R\in\sset_1}\sum_{v\in\VV_R}h_{k-1}(\kK_{R}/v),\\
  \sum_{i=1}^{3}\sum_{v\in\VV_i}h_{k-2}(\kK_{\{i\}}/v)
  &=\sum_{R\in\sset_1}\sum_{v\in\VV_R}h_{k-2}(\kK_{R}/v),
\end{align*}
and
\begin{equation*}
  \sum_{R\in\sset_2}\sum_{\emptyset\subset{}S\subseteq{}R}h_{k}(\fF_S)
  =\sum_{R\in\sset_2}h_k(\fF_R)+2\sum_{R\in\sset_1}h_k(\fF_R),
\end{equation*}
the expression in \eqref{equ:h-vectors-Q-links-RHS} can be rewritten
in the following more convenient form:
\begin{align*}
  &\overbrace{\sum_{v\in\VV_{[3]}}h_k(\kK_{[3]}/v)}^{T_{12}}
  +\overbrace{\sum_{R\in\sset_2}\sum_{v\in\VV_R}h_{k-1}(\kK_{R}/v)}^{T_{13}}
  +\overbrace{\sum_{R\in\sset_1}\sum_{v\in\VV_R}[h_{k-1}(\kK_{R}/v)+h_{k-2}(\kK_{R}/v)]}^{T_{14}}\\
  &+\overbrace{\sum_{R\in\sset_1}[h_k(\fF_R) + h_{k-1}(\fF_R)]
  +\sum_{R\in\sset_2}h_k(\fF_R)+2\sum_{R\in\sset_1}h_k(\fF_R)}^{T_{15}}.
\end{align*}

Solving relation \eqref{equ:linksQ} in terms of the term $T_1$, we get:
\begin{align*}
  T_1&=T_{12}+T_{13}+T_{14}+T_{15}-(T_2+T_3+T_4)\\
  &=T_{12}+T_{13}+T_{14}+T_{15}-[(T_5+T_6-T_7)+(T_8+T_9)+(T_{10}+T_{11})]\\
  &=T_{12}+(T_{13}-T_{6})+(T_{14}+T_7-T_8-T_{10})+(T_{15}-T_5-T_{9}-T_{11})\\
  &=T_{12}+(T_{13}-T_{6})+(T_{14}+T_7-T_8-T_{10}),
\end{align*}
where we used the fact that the terms $T_5$, $T_9$ and $T_{11}$
cancel-out with the term $T_{15}$.
Observe now that:
\begin{equation*}
  T_{13}-T_6
  =\sum_{R\in\sset_2}\sum_{v\in\VV_R}h_{k-1}(\kK_{R}/v)
  -\sum_{R\in\sset_2}\sum_{v\in\VV_R}h_{k}(\kK_{R}/v)
  =-\sum_{R\in\sset_2}\sum_{v\in\VV_R}g_k(\kK_{R}/v),
\end{equation*}
while
\begin{align*}
  T_{14}+T_7-T_8-T_{10}
  &=\sum_{R\in\sset_1}\sum_{v\in\VV_R}[h_{k-1}(\kK_{R}/v)+h_{k-2}(\kK_{R}/v)]
  +2\sum_{R\in\sset_1}\sum_{v\in \VV_R}g_k(\kK_{R}/v)\\
  &\quad
  -\sum_{R\in\sset_1} \sum_{v\in \VV_R} h_k (\kK_R / v)
  -\sum_{R\in\sset_1} \sum_{v\in \VV_R} h_{k-1} (\kK_R / v)\\
  &=\sum_{R\in\sset_1}\sum_{v\in\VV_R}\{h_{k-1}(\kK_{R}/v)+h_{k-2}(\kK_{R}/v)
  +2[h_k(\kK_{R}/v)-h_{k-1}(\kK_{R}/v)]\\
  &\qquad\qquad\qquad-h_k(\kK_R/v)-h_{k-1}(\kK_R/v)\}\\
  &=\sum_{R\in\sset_1}\sum_{v\in\VV_R}[h_k(\kK_{R}/v)
  -2h_{k-1}(\kK_{R}/v)+h_{k-2}(\kK_{R}/v)]\\
  &=\sum_{R\in\sset_1}\sum_{v\in\VV_R}g^{(2)}_k(\kK_{R}/v).
\end{align*}
Hence,
\begin{align*}
  T_1&=\sum_{v\in\VV_{[3]}}h_k(\kK_{[3]}/v)
  -\sum_{R\in\sset_2}\sum_{v\in\VV_R}g_k(\kK_{R}/v)
  +\sum_{R\in\sset_1}\sum_{v\in\VV_R}g^{(2)}_k(\kK_{R}/v)\\
  &=\sum_{\emptyset{}\subset{}R\subseteq[3]}(-1)^{3-|R|}
  \sum_{v\in\VV_R}g^{(3-|R|)}_k(\kK_{R}/v).\qedhere
\end{align*}
\end{proof}

\proofseparator

\begin{proof}[Proof of Lemma~\ref{lem:hkWrecur}]
  By Lemma~\ref{lem:LinksNoLinks}, relation
  \eqref{equ:recur-relation-F3-wrt-K} yields:
  \begin{align}
    \nonumber (k+1) h_{k+1}(\fF_{[3]})+ (d+2-k)  h_k(\fF_{[3]}) &\leq 
    \sum_{\emptyset \subset R \subseteq [3]} (-1)^{3-|R|}
    \sum_{v\in \VV_R}  g_k^{(3-|R|)}(\kK_{R} ) \\[1em]
    &=  n_{[3]}  h_k(\kK_{[3]}) - \sum_{R\in \sset_2} n_R g_k(\kK_{R})
    + \sum_{R\in \sset_1} n_R g_k^{(2)}(\kK_{R}) 
    \label{r8a}
  \end{align}

  By relation \eqref{equ:hkKR} with $R\equiv{}[3]$, we can write
  $h_k(\kK_{[3]})$ as:
  \begin{equation}\label{equ:hkKexpansion2}
    h_k(\kK_{[3]})=h_k(\fF_{[3]}) + \sum_{R \in \sset_2}g_{k}(\fF_R)
    + \sum_{R\in \sset_1}g_{k}^{(2)}(\fF_R),
  \end{equation}
  whereas from relation \eqref{equ:hkKR} for all $R\in \sset_2$ we easily get:
  \begin{equation}
    g_{k}(\kK_R)= g_{k}(\fF_R)
    + \sum_{\emptyset \subset S \subset R} g^{(2)}_k(\fF_S ).
    \label{equ:g_kF_R} 
  \end{equation} 
  
  Since $\kK_R\equiv\fF_R$, for any $R\in\sset_1$,
  we can employ relations \eqref{equ:hkKexpansion2} and
  \eqref{equ:g_kF_R} to rewrite the right hand side of
  \eqref{r8a} as follows:
  \begin{align*}
    n_{[3]} h_k(\kK_{[3]})-\sum_{R\in \sset_2} n_R&\,g_k(\kK_{R})
    + \sum_{R\in \sset_1} n_R g_k^{(2)}(\kK_{R})\\
    &=n_{[3]} h_k(\fF_{[3]}) +  n_{[3]}\sum_{R\in \sset_2}
    g_k(\fF_R)+ n_{[3]}\sum_{R\in\sset_1} g_k^{(2)}(\fF_R)\\
    &\quad  -\left[ \sum_{R\in \sset_2} n_R g_{k}(\fF_R)  + \sum_{R\in \sset_2}
      n_R \sum_{\emptyset \subset S \subset R} g^{(2)}_k(\fF_S)\right]
    + \sum_{R\in\sset_1} n_R g_k^{(2)} (\fF_R)\\[1em]
    & = n_{[3]} h_k(\fF_{[3]})+\sum_{R\in\sset_2}(n_{[3]}-n_R)g_k(\fF_R)\\
    &\quad+\overbrace{\left[n_{[3]} \sum_{i=1}^{3} g_k^{(2)}(\fF_{\{i\}}) -
        \sum_{R\in \sset_2} n_R \sum_{\emptyset\subset{}S\subset{}R}g^{(2)}_k(\fF_S )
        +\sum_{i=1}^{3} n_{i} g_k^{(2)}(\fF_{\{i\}})\right]}^{T}.
  \end{align*}
  Using the identity:
  \begin{equation*}
    \sum_{R\in \sset_2} n_R \sum_{\emptyset \subset S \subset R}g^{(2)}_k(\fF_S ) 
    = 2\sum_{i=1}^3n_i g_k^{(2)}(\fF_{\{i\}})
    + \sum_{i=1}^3 n_{[3]\sm \{i\}}\, g_k^{(2)}(\fF_{\{i\}}),
  \end{equation*}
  we see that the last term (term $T$) in the relation above vanishes:
  \begin{align*}
    n_{[3]} &\sum_{i=1}^{3} g_k^{(2)}(\fF_{\{i\}}) -
    \sum_{R\in \sset_2} n_R \sum_{\emptyset\subset{}S\subset{}R}g^{(2)}_k(\fF_S )
    +\sum_{i=1}^{3} n_{\{i\}} g_k^{(2)}(\fF_{\{i\}})\\
    &=
    n_{[3]} \sum_{i=1}^{3} g_k^{(2)}(\fF_{\{i\}})
    -\left[2\sum_{i=1}^3n_i g_k^{(2)}(\fF_{\{i\}})
      + \sum_{i=1}^3 n_{[3]\sm \{i\}}\, g_k^{(2)}(\fF_{\{i\}})\right]
    +\sum_{i=1}^{3} n_{i} g_k^{(2)}(\fF_{\{i\}})\\
    &=\sum_{i=1}^{3}(n_{[3]}-2n_i-n_{[3]\sm
      \{i\}}+n_i)\,g_k^{(2)}(\fF_{\{i\}})
    =0.
  \end{align*}
  Hence, relation \eqref{r8a} simplifies to:
  \begin{align*}
    (k+1) h_{k+1}(\fF_{[3]})&+ (d+2-k)  h_k(\fF_{[3]}) \leq
    n_{[3]} h_k(\fF_{[3]})+\sum_{R\in\sset_2}(n_{[3]}-n_R)g_k(\fF_R)\\
    &=n_{[3]} h_k(\fF_{[3]})+\sum_{R\in\sset_2}n_{[3]\sm{}R}g_k(\fF_R)
    = n_{[3]} h_k(\fF_{[3]}) +  \sum_{i=1}^3 n_{i}g_k(\fF_{[3]\sm \{i\}}),
  \end{align*}
  from which we obtain the relation in the statement of the lemma.
\end{proof}

\subsection{Omitted \& full proofs of Section~\ref{sec:ub}}
\label{app:ub}

\begin{proof}[Proof of Lemma~\ref{lem:hkKbound}]
The bound for $h_k(\kK_{[3]})$ holds as equality for $k=0$, since
by relation~\eqref{equ:hkKR} with $R=[3]$, (see also \eqref{equ:hkKexpansion2}), we have
  \begin{align*}
    h_0(\kK_{[3]})&=h_0(\fF_{[3]})+\sum_{R\in\sset_2}g_0(\fF_{R})
    +\sum_{i=1}^3g_0^{(2)}(\bx{}P_i)\\
    &=1+\sum_{R\in\sset_2}[h_0(\fF_{R})-h_{-1}(\fF_{R})]
    +\sum_{i=1}^3[h_0(\bx{}P_i)-2h_{-1}(\bx{}P_i)+h_{-2}(\bx{}P_i)]\\
    &=1+\sum_{R\in\sset_2}[(-1)-0]
    +\sum_{i=1}^3[1-2\cdot{}0+0]=1.
  \end{align*}
  Suppose now that $k\ge{}1$. Then, using relation
  \eqref{equ:gkFrecur}, we get, for $k\ge{}1$:
  \begin{align*}
    h_k(\kK_{[3]})
    &=h_k(\fF_{[3]})+\sum_{R\in\sset_2}g_k(\fF_R)
    +\sum_{R\in\sset_1}g_k^{(2)}(\fF_R)\\
    &\le{}h_k(\fF_{[3]})
    +\sum_{R\in\sset_2}\left[\tfrac{n_R-d-3}{k}h_{k-1}(\fF_{R})
    +\sum_{i\in{}R}\tfrac{n_{R\sm\{i\}}}{k}g_{k-1}(\bx{}P_i)\right]
    +\sum_{i=1}^{3}g_k^{(2)}(\bx{}P_i)\\
    &=h_k(\fF_{[3]})
    +\sum_{R\in\sset_2}\tfrac{n_R-d-3}{k}h_{k-1}(\fF_{R})
    +\sum_{i=1}^3\left[\tfrac{n_{[3]\sm\{i\}}}{k}g_{k-1}(\bx{}P_i)
        +g_k(\bx{}P_i)-g_{k-1}(\bx{}P_i)\right],
  \end{align*}
  which finally yields:
  \begin{equation}\label{equ:hkKrecur}
    h_k(\kK_{[3]})\le{}h_k(\fF_{[3]})
    +\sum_{R\in\sset_2}\tfrac{n_R-d-3}{k}h_{k-1}(\fF_{R})
    +\sum_{i=1}^3\left[\tfrac{n_{[3]\sm\{i\}}-k}{k}g_{k-1}(\bx{}P_i)
        +g_k(\bx{}P_i)\right].
  \end{equation}
  Since $n_R-d-3\ge{}2(d+1)-d-3=d-1>0$, for $R\in\sset_2$, and
  $n_{[3]\sm\{i\}}-k\ge{}2(d+1)-(d+2)=d>0$ for any
  $0\le{}k\le{}d+2$, we can use the upper bounds for
  $h_k(\fF_{[3]})$ and $h_{k-1}(\fF_{R})$, $R\in \sset_2$ from Lemma
  \ref{lem:hkWbound} and \cite[Lemma 3.3]{kt-mnfms-12}, respectively, in
  conjunction with the known upper bounds for the elements of the
  $g$-vector of a $d$-polytope (cf.~\cite[Corollary
  8.38]{z-lp-95}). More precisely:
  \begin{align*}
    h_k(\kK_{[3]})&\le{}\sum_{\emptyset\subset{}S\subseteq{}[3]}(-1)^{3-|S|}
    \tbinom{n_S-d-3+k}{k}
    +\sum_{R\in\sset_2}\tfrac{n_R-d-3}{k}
    \left[\tbinom{n_R-d-2+k-1}{k-1}-
      \sum_{i\in{}R}\tbinom{n_i-d-2+k-1}{k-1}\right]\\
    &\qquad
    +\sum_{i=1}^3\left[\tfrac{n_{[3]\sm\{i\}}-k}{k}\tbinom{n_i-d-2+k-1}{k-1}
      +g_k(\bx{}P_i)\right]\\
    &=\sum_{\emptyset\subset{}S\subseteq{}[3]}(-1)^{3-|S|}
    \tbinom{n_S-d-3+k}{k}
    +\sum_{R\in\rR_2}\tfrac{n_R-d-3}{k}
    \left[\tbinom{n_R-d-3+k}{k-1}-\sum_{i\in{}R}\tbinom{n_i-d-3+k}{k-1}\right]\\
    &\qquad
    +\sum_{i=1}^3\left[\tfrac{n_{[3]\sm\{i\}}}{k}\tbinom{n_i-d-3+k}{k-1}
      -\tbinom{n_i-d-3+k}{k-1}
      +\tbinom{n_i-d-2+k}{k}+g_k(\bx{}P_i)-\tbinom{n_i-d-2+k}{k}\right]\\
    &=\tbinom{n_{[3]}-d-3+k}{k}-\sum_{i=1}^3\tbinom{n_{[3]\sm\{i\}}-d-3+k}{k}
    +\sum_{i=1}^3\tbinom{n_i-d-3+k}{k}\\
    &\qquad+\sum_{R\in\sset_2}\tfrac{n_R-d-3}{k}
    \left[\tbinom{n_R-d-3+k}{k-1}
    -\sum_{i\in{}R}\tbinom{n_i-d-3+k}{k-1}\right]\\
    &\qquad
    +\sum_{i=1}^3\tfrac{n_{[3]\sm\{i\}}}{k}\tbinom{n_i-d-3+k}{k-1}
    +\sum_{i=1}^3\tbinom{n_i-d-3+k}{k}
    +\sum_{i=1}^3\left[g_k(\bx{}P_i)-\tbinom{n_i-d-2+k}{k}\right].
  \end{align*}
  From the proof of Lemma \ref{lem:gkFbound} it is easy to see that:
  \begin{align*}
    &\sum_{R\in\sset_2}\tfrac{n_R-d-3}{k}\left[\tbinom{n_R-d-3+k}{k-1}
      -\sum_{i\in{}R}\tbinom{n_i-d-3+k}{k-1}\right]
    +\sum_{i=1}^3\tfrac{n_{[3]\sm\{i\}}}{k}\tbinom{n_i-d-3+k}{k-1}\\
    &=\sum_{R\in\sset_2}\left[\tbinom{n_R-d-3+k}{k}
      -\sum_{i\in{}R}\tbinom{n_i-d-3+k}{k}\right]
    =\sum_{i=1}^3\tbinom{n_{[3]\sm{}\{i\}}-d-3+k}{k}
    -2\sum_{i=1}^3\tbinom{n_i-d-3+k}{k}
  \end{align*}
  Hence we have:
  \begin{align*}
    h_k(\kK_{[3]})&\le\tbinom{n_{[3]}-d-3+k}{k}
    -\sum_{i=1}^3\tbinom{n_{[3]\sm\{i\}}-d-3+k}{k}
    +\sum_{i=1}^3\tbinom{n_i-d-3+k}{k}
    +\sum_{i=1}^3\tbinom{n_{[3]\sm{}\{i\}}-d-3+k}{k}\\
    &\quad-2\sum_{i=1}^3\tbinom{n_i-d-3+k}{k}+\sum_{i=1}^3\tbinom{n_i-d-3+k}{k}
    +\sum_{i=1}^3\left[g_k(\bx{}P_i)-\tbinom{n_i-d-2+k}{k}\right]\\
    &=\tbinom{n_{[3]}-d-3+k}{k}
    +\sum_{i=1}^3\left[g_k(\bx{}P_i)-\tbinom{n_i-d-2+k}{k}\right].
  \end{align*}
  Since $g_k(\bx{}P_i)-\binom{n_i-d-2+k}{k}\le{}0$, for all $k\ge{}0$,
  we get the sought-for bound in \eqref{equ:hkKbound-smallk} for
  $0\le{}k\le{}d+2$. Furthermore, for $d$ odd and $k=\lexp{d}+1$,
  we have $g_k(\bx{}P_i)=0$, which yields the bound in
  \eqref{equ:hkKbound-midk}.

  To prove the equality claim, we distinguish between the cases
  $k\le\lexp{d}$, and $k=\lexp{d}+1$ with $d$ odd.
  Consider the case $k\le\lexp{d}$ first, and assume that
  $h_k(\kK_{[3]})=\binom{n_{[3]}-d-3+k}{k}$.
  From relation \eqref{equ:hkKrecur} we deduce that both
  $h_k(\fF_{[3]})$ and $g_k(\bx{}P_i)$, $1\le{}i\le{}3$, must be equal
  to their maximum values, since otherwise we would have that
  $h_k(\kK_{[3]})<\binom{n_{[3]}-d-3+k}{k}$. In view of Lemma
  \ref{lem:hkWbound}, the maximality of $h_k(\fF_{[3]})$ implies that
  $f_{l-1}(\fF_{[3]})=\sum_{\emptyset\subset{}S\subseteq{}[3]}
  (-1)^{3-|S|}\binom{n_S}{l}$, for all $0\le{}l\le{}k$, whereas the
  maximality of $g_k(\bx{}P_i)$ implies that $P_i$ is $k$-neighborly,
  for all $1\le{}i\le{}3$, \ie for all $1\le{}i\le{}3$,
  $f_{l-1}(\bx{}P_i)=f_{l-1}(\fF_{\{i\}})=\binom{n_i}{l}$, for all
  $0\le{}l\le{}k$. But then we also have that
  $g_{k-1}(\bx{}P_i)=\binom{n_i-d-2+k-1}{k-1}$, which gives:
  \begin{equation}\label{equ:gk2calc}
    g_k^{(2)}(\bx{}P_i)=g_k(\bx{}P_i)-g_{k-1}(\bx{}P_i)
    =\tbinom{n_i-d-2+k}{k}-\tbinom{n_i-d-2-k-1}{k-1}
    =\tbinom{n_i-d-3+k}{k}.
  \end{equation}
  By relation \eqref{r8}, the maximality of
  $h_k(\fF_{[3]})$ implies that $g_{k-1}(\fF_{[3]\sm\{i\}})$ attains its maximum
  value for all $1\le{}i\le{}3$. By following the argumentation in the
  proof of Lemma \ref{lem:gkFbound}, the maximality of
  $g_{k-1}(\fF_{[3]\sm\{i\}})$ further implies that
  $h_l(\fF_{[3]\sm\{i\}})$ is maximal, for all $0\le{}l\le{}k-1$.
  Solving, now, equation
  \eqref{equ:hkKR} (for $R\equiv{}[3]$) in terms of the sum of the
  $h_k(\fF_{[3]\sm\{i\}})$'s we get:
  \begin{equation*}
    \sum_{i=1}^3h_k(\fF_{[3]\sm\{i\}})=h_k(\kK_{[3]})-h_k(\fF_{[3]})
    +\sum_{i=1}^3h_{k-1}(\fF_{[3]\sm\{i\}})-\sum_{i=1}^3g_k^{(2)}(\bx{}P_i).
  \end{equation*}
  Substituting in the above equation the values for $h_k(\kK_{[3]})$,
  $h_k(\fF_{[3]})$, $h_{k-1}(\fF_{[3]\sm\{i\}})$ and
  $g_k^{(2)}(\bx{}P_i)$, it is easy to verify that
  \begin{equation*}
    \sum_{i=1}^3h_k(\fF_{[3]\sm\{i\}})=
    \sum_{i=1}^3\left[\tbinom{n_{[3]\sm\{i\}}-d-2+k}{k}
    -\sum_{j\in[3]\sm\{i\}}\tbinom{n_j-d-2+k}{k}\right].
  \end{equation*}
  In other words, the sum of the $h_k(\fF_{[3]\sm\{i\}})$'s attains its
  maximum value, which implies that each of the summands attains its
  maximum value. We thus conclude that $h_l(\fF_{[3]\sm\{i\}})$ is
  maximal, for all $0\le{}l\le{}k$, which, by \cite[Lemma
  3.3]{kt-mnfms-12}, implies that, for all $R\in\sset_2$,
  $f_{l-1}(\fF_{R})=\sum_{\emptyset\subset{}S\subseteq{}R}(-1)^{2-|S|}\binom{n_S}{l}$,
  for all $0\le{}l\le{}k$.

  Let us now consider the reverse direction
  and assume that for all $\emptyset\subset{}R\subseteq[3]$,
  $f_{l-1}(\fF_{R})=\sum_{\emptyset\subset{}S\subseteq{}R}
  (-1)^{|R|-|S|}\binom{n_S}{l}$, for all $0\le{}l\le{}k$ (for
  $k\le\lexp{d}$, $\min\{k,\lexp{d+|R|-1}\}=k$).
  Using Lemma \ref{lem:hkWbound}, the condition above, for $R=[3]$,
  implies that $h_l(\fF_{[3]})$ attains its upper bound value for all
  $0\le{}l\le{}k$. Using \cite[Lemma 3.3]{kt-mnfms-12}, the condition
  above, for $R\in\sset_2$, implies that $h_l(\fF_R)$ attains its upper
  bound value for all $0\le{}l\le{}k$, and thus $g_k(\fF_R)$ attains
  its upper bound value. Finally, the condition above,
  for $1\le{}i\le{}3$, implies that $P_i$ is $k$-neighborly, which means
  that $g_l(\bx{}P_i)=g_l(\fF_{\{i\}})=\binom{n_i-d-2+l}{l}$, for all
  $0\le{}l\le{}k$, and thus (cf. \eqref{equ:gk2calc})
  $g_k^{(2)}(\bx{}P_i)=g_k^{(2)}(\fF_{\{i\}})=\binom{n_i-d-3+k}{k}$.
  Appealing now to relation \eqref{equ:hkKR} for $R\equiv{}[3]$, it is easy to
  verify that $h_k(\kK_{[3]})=\binom{n_{[3]}-d-3+k}{k}$.

  We end the equality claim proof by considering the case
  $k=\lexp{d}+1$, for $d$ odd. Since for $d$ odd,
  $g_{\lexp{d}+1}(\bx{}P_i)=0$, relation \eqref{equ:hkKrecur},
  simplifies to:
  \begin{equation}\label{equ:hkKrecur-dodd}
    h_{\lexp{d}+1}(\kK_{[3]})\le{}h_{\lexp{d}+1}(\fF_{[3]})
    +\sum_{R\in\sset_2}\tfrac{n_R-d-3}{\lexp{d}+1}h_{\lexp{d}}(\fF_{R})
    +\sum_{i=1}^3\tfrac{n_{[3]\sm\{i\}}-\lexp{d}-1}{\lexp{d}+1}
    g_{\lexp{d}}(\bx{}P_i),
  \end{equation}
  while relation \eqref{equ:hkKR} (with $R\equiv{}[3]$) simplifies to:
  \begin{equation}\label{equ:hkKexpansion-dodd}
    h_{\lexp{d}+1}(\kK_{[3]})=h_{\lexp{d}+1}(\fF_{[3]})
    +\sum_{i=1}^3g_{\lexp{d}+1}(\fF_{[3]\sm\{i\}})
    -\sum_{i=1}^3g_{\lexp{d}}(\bx{}P_i).
  \end{equation}
  The argument in this case is essentially the same as
  before. Assuming that $h_{\lexp{d}+1}(\kK_{[3]})$ is maximal, we deduce,
  from \eqref{equ:hkKrecur-dodd}, that both
  $h_{\lexp{d}+1}(\fF_{[3]})$ and $g_{\lexp{d}}(\bx{}P_i)$ are
  maximal, which, imply, respectively, that
  $f_{l-1}(\fF_{[3]})=\sum_{\emptyset\subset{}S\subseteq{}[3]}
  (-1)^{3-|S|}\binom{n_S}{l}$, for all
  $0\le{}l\le{}\lexp{d}+1=\lexp{d+2}$, and that, for all $1\le{}i\le{}3$,
  $f_{l-1}(\fF_{\{i\}})=\binom{n_i}{l}$, for all $0\le{}l\le\lexp{d}$.
  The maximality of $h_{\lexp{d}+1}(\fF_{[3]})$ implies also the
  maximality of $g_l(\fF_{R})$, for all $R\in\sset_2$, and for all
  $0\le{}l\le\lexp{d}$, and thus
  the maximality of $h_l(\fF_{R})$, for all $R\in\sset_2$, and for all
  $0\le{}l\le\lexp{d}$. By solving equation
  \eqref{equ:hkKexpansion-dodd} in terms of the sum of the
  $h_{\lexp{d}+1}(\fF_R)$'s, we also deduce that
  $h_{\lexp{d}+1}(\fF_R)$ is maximal, for all $R\in\sset_2$. Hence, we
  have that $h_l(\fF_R)$ is maximal, for all $R\in\sset_2$,
  and for all $0\le{}l\le{}\lexp{d}+1=\lexp{d+1}$, which, by
  \cite[Lemma 3.3]{kt-mnfms-12}, gives that 
  $f_{l-1}(\fF_{R})=\sum_{\emptyset\subset{}S\subseteq{}R}
  (-1)^{2-|S|}\binom{n_S}{l}$, for all $R\in\sset_2$, and
  for all $0\le{}l\le{}\lexp{d+1}$.\\
  Assuming now that, for all $\emptyset\subset{}R\subseteq[3]$,
  $f_{l-1}(\fF_{R})=\sum_{\emptyset\subset{}S\subseteq{}R}
  (-1)^{|R|-|S|}\binom{n_S}{l}$, for all
  $0\le{}l\le\min\{\lexp{d}+1,\lexp{d+|R|-1}\}$,
  we deduce, from Lemma \ref{lem:hkWbound}, that
  that $h_{\lexp{d}+1}(\fF_{[3]})$ attains its upper bound value for all
  $0\le{}l\le{}\lexp{d+2}=\lexp{d}+1$.
  Furthermore, Lemma 3.3 in \cite{kt-mnfms-12}, 
  implies that, for all $R\in\sset_2$, $h_l(\fF_R)$ attains its upper
  bound value for all $0\le{}l\le{}\lexp{d+1}=\lexp{d}+1$, which means
  that $g_{\lexp{d+1}}(\fF_R)$ attains its upper bound value, for all
  $R\in\sset_2$. Finally, our assumption above, implies that, 
  for all $1\le{}i\le{}3$, $P_i$ is neighborly, which means
  that $g_{\lexp{d}}(\bx{}P_i)=\binom{n_i-\lexp{d}-3}{\lexp{d}}$.
  Appealing to relation \eqref{equ:hkKexpansion-dodd} above, it is
  easy to verify that $h_{\lexp{d}+1}(\kK_{[3]})$ attains its upper bound in
  \eqref{equ:hkKbound-midk}.
\end{proof}

\section{Asymptotic analysis of Vandermonde-like determinants}

We start by introducing what is known as
\emph{Laplace's Expansion Theorem} for determinants
(see \cite{g-tm-60,hk-la-71} for details and proofs).
Consider a $n\times{}n$ matrix $A$.
Let $\mb{r}=(r_1,r_2,\ldots,r_k)$, be a vector of $k$ row indices for
$A$, where $1\le{}k<n$ and $1\le{}r_1<r_2<\ldots<r_k\le{}n$.
Let $\mb{c}=(c_1,c_2,\ldots,c_k)$ be a vector of $k$ column indices for
$A$, where $1\le{}k<n$ and $1\le{}c_1<c_2<\ldots<c_k\le{}n$. We
denote by $S(A;\mb{r},\mb{c})$ the $k\times{}k$ submatrix of $A$
constructed by keeping the entries of $A$ that belong to a row in
$\mb{r}$ and a column in $\mb{c}$.
The \emph{complementary submatrix} for $S(A;\mb{r},\mb{c})$, denoted
by $\compl(A;\mb{r},\mb{c})$, is the $(n-k)\times(n-k)$ submatrix of
$A$ constructed by removing the rows and columns of $A$ in $\mb{r}$
and $\mb{c}$, respectively. Then, the determinant of $A$ can be
computed by expanding in terms of the $k$ columns of $A$ in
$\mb{c}$ according to the following theorem.

\begin{theorem}[\textbf{Laplace's Expansion Theorem}]\label{thm:LET}
  Let $A$ be a $n\times{}n$ matrix. Let
  $\mb{c}=(c_1,c_2,\ldots,\allowbreak{}c_k)$ be a vector of $k$ column
  indices for
  $A$, where $1\le{}k<n$ and $1\le{}c_1<c_2<\ldots<c_k\le{}n$. Then:
  \begin{equation}
    \det(A)=\sum_{\mb{r}}(-1)^{|\mb{r}|+|\mb{c}|}
    \det(S(A;\mb{r},\mb{c}))\det(\compl(A;\mb{r},\mb{c})),
  \end{equation}
  where $|\mb{r}|=r_1+r_2+\ldots+r_k$,
  $|\mb{c}|=c_1+c_2+\ldots+c_k$, and the summation is taken over all
  row vectors $\mb{r}=(r_1,r_2,\ldots,r_k)$ of $k$ row indices for
  $A$, where $1\le{}r_1<r_2<\ldots<r_k\le{}n$.
\end{theorem}

In what follows  we recall some facts concerning   generalized
Vandermonde determinants that will be in use  to us later. 
Let $ n \geq 2,$\;   $ \mb{x}=(x_1,\ldots,x_n) $  
and  $\mb{\mu}=(\mu_1,\mu_2,\ldots,\mu_n)$, where
we require that $0\le{}\mu_1<\mu_2<\ldots<\mu_n.$ The 
\emph{ generalized Vandermonde determinant},  
denoted by $ \GVD(\mb{x};\mb{\mu}), $ is the  $n\times n$ 
determinant whose  $i$-th row is  the vector $\mb{x}$  with  all its entries raised to $ \mu_i.$ 
 While there is no general formula for the generalized Vandermonde determinant, 
it is a well-known fact that, if the elements of $\mb{x}$ are in
strictly increasing order, then $\GVD(\mb{x};\mb{\mu})>0$ (for
example, see \cite{g-atm-05} for a proof of this fact).

\medskip

In the remainder of this section we consider two determinants
that are parameterized by a positive parameter $\tau$, and we study
their asymptotic behavior with respect to $\tau$. These determinants
are generalizations of the determinants that arise in the proofs of
Lemmas \ref{lem:fkW2-tau} and \ref{lem:fkW3-tau} in Section
\ref{sec:lb}, and are directly associated with the equations of
some appropriately defined supporting hyperplanes for the faces of
$\fF_R$ where $R\in\sset_2$ or $R\equiv[3]$ (recall that $\fF_R$
stands for the set of faces of the Cayley polytope of $|R|$ polytopes
$P_i$, $i\in{}R$, with the property that each face in $\fF_R$ has at
least one vertex from each polytope $P_i$).
The two determinants that we study are generalized-Vandermonde-like
determinants that are polynomial functions of $\tau$, and
correspond, respectively, to the two cases $R\in\sset_2$ and
$R\equiv[3]$ mentioned above.
Since in Section \ref{sec:lb} we are interested in small values of 
$\tau$, our asymptotic analysis in the two lemmas below is
targeted towards revealing the term of $\tau$ of minimal
exponent.

We start-off with the generalized version of the determinant that
arises in the upper bound tightness construction in Section
\ref{sec:lb} when $R\in\sset_2$.

\begin{lemma}
\label{lem:det2}
Fix two integers $m\ge{}2$ and $n\ge{}2$, with $n+m\ge{}5$.
Let $D_{n,m}(\tau;I,J,\mb{\mu})$ be the $(n+m)\times(n+m)$ determinant:
\begin{equation*}
  (-1)^{J+1}
  \left|
    \begin{array}{cccccc}
      (x_1\tau^{\alpha})^{\mu_1}
      &\cdots &(x_n\tau^{\alpha})^{\mu_1}&0 & \cdots & 0 \\[2pt]
      0 & \cdots & 0 
      &(y_1\tau^{\beta})^{\mu_2}
      &\cdots& (y_m\tau^{\beta})^{\mu_2}\\[3pt]
      f_3(\tau)(x_1\tau^{\alpha})^{\mu_3}
      &\cdots&f_3(\tau)(x_n\tau^{\alpha})^{\mu_3}
      &g_3(\tau)(y_1\tau^{\beta})^{\mu_3}
      &\cdots&g_3(\tau)(y_m\tau^{\beta})^{\mu_3}\\[3pt]
      f_4(\tau)(x_1\tau^{\alpha})^{\mu_4}
      &\cdots&f_4(\tau)(x_n\tau^{\alpha})^{\mu_4}
      &g_4(\tau)(y_1\tau^{\beta})^{\mu_4}
      &\cdots&g_4(\tau)(y_m\tau^{\beta})^{\mu_4}\\[3pt]
      f_5(\tau)(x_1\tau^{\alpha})^{\mu_5}
      &\cdots&f_5(\tau)(x_n\tau^{\alpha})^{\mu_5}
      &g_5(\tau)(y_1\tau^{\beta})^{\mu_5}
      &\cdots&g_5(\tau)(y_m\tau^{\beta})^{\mu_3}\\[3pt]
      (x_1\tau^{\alpha})^{\mu_6}
      &\cdots&(x_n\tau^{\alpha})^{\mu_6}&(y_1\tau^{\beta})^{\mu_6}
      &\cdots&(y_m\tau^{\beta})^{\mu_6}\\[3pt]
      (x_1\tau^{\alpha})^{\mu_7}
      &\cdots&(x_n\tau^{\alpha})^{\mu_7}&(y_1\tau^{\beta})^{\mu_7}
      &\cdots&(y_m\tau^{\beta})^{\mu_7}\\[3pt]
      \vdots   & \ddots  & \vdots 
      & \vdots  & \ddots  & \vdots  \\[3pt]
      (x_1\tau^{\alpha})^{\mu_\ell}
      &\cdots&(x_n\tau^{\alpha})^{\mu_\ell}&(y_1\tau^{\beta})^{\mu_\ell}
      &\cdots&(y_m\tau^{\beta})^{\mu_\ell}\\[3pt]
    \end{array}
  \right|
\end{equation*}
where $0<x_1<x_2<\ldots<x_n$, $0<y_1<y_2<\ldots<y_m$, $\ell=n+m$,
$\mb{\mu}=(\mu_1,\ldots,\mu_\ell)$,
with $0\le\mu_1\le\mu_2<\mu_3<\ldots<\mu_\ell$,
$(I,J)\in\{(3,4),(3,5),(4,5)\}$,
$f_{I}(\tau)=g_{J}(\tau)=1$, 
$f_i(\tau)=g_j(\tau)=\tau^M$, for $i\ne{}I$ and $j\ne{}J$, 
$\alpha>\beta\ge{}0$, $M\ge\alpha|\mb{\mu}|$ and
$\tau>0$. Then:
\begin{equation*}
  D_{n,m}(\tau;I,J,\mb{\mu})=C\tau^{\xi}+\Theta(\tau^{\xi+1}),\quad
  \xi = \alpha\left(\mu_1+\mu_3+\sum_{i=4}^{n+2}\mu_i-\mu_J\right)
  +\beta\left(\mu_2+\mu_J+\sum_{i=n+3}^{\ell}\mu_i\right),
\end{equation*}
where $C$ is a positive constant independent of $\tau$.
\end{lemma}

\begin{proof}
  For simplicity, we write $D_{n,m}(\tau)$ instead of
  $D_{n,m}(\tau;I,J,\mb{\mu})$, suppressing $I$, $J$ and $\mb{\mu}$ in
  the notation.
  We denote by $\Delta_{n,m}(\tau)$ the matrix corresponding to the
  determinant $(-1)^{J+1}D_{n,m}(\tau)$.
  If we apply Laplace's expansion with respect to the first $n$
  columns, \ie when $\mb{c} = (1,2,\ldots,n)$, we get: 
  \begin{align}
    D_{n,m}(\tau)&=(-1)^{J+1}
    \sum_{ \substack{\mb{r}=(r_1,r_2, \ldots, r_n) \\
        1 \leq r_1  < r_2 < \cdots <  r_n \leq   n+m  } }
    (-1)^{|\mb{r}| + |\mb{c}|}\,
    \det(S(\Delta_{n,m}(\tau);\mb{r},\mb{c}))
    \det(\compl(\Delta_{n,m}(\tau);\mb{r},\mb{c}))\notag\\
    &=\sum_{ \substack{\mb{r}=(r_1,r_2,\ldots,r_n)\\
        1 \leq r_1 < r_2 <  \cdots <  r_n \leq   n+m  } }
    \!\!\!
    (-1)^{|\mb{r}|+\frac{n(n+1)}{2}+J+1}\,
    \det(S(\Delta_{n,m}(\tau);\mb{r},\mb{c}))
    \det(\compl(\Delta_{n,m}(\tau);\mb{r},\mb{c})).
    \label{expansion1}
  \end{align}

  The above sum consists of $ \binom{n+m}{n}$ terms. Among these terms:
 \begin{enumerate}
  \item
    all those for which $\mb{r}$ contains the second row
    vanish (in this case the corresponding row of
    $S(\Delta_{n,m}(\tau);\mb{r},\mb{c})$
    consists of zeros), and
  \item
    all those for which $\mb{r}$ does not contain the first
    row vanish (in this case at least two rows of
    $\compl(\Delta_{n,m}(\tau);\mb{r},\mb{c})$ consist of zeros).
  \end{enumerate}
  The remaining terms of the expansion are the $\binom{n+m-2}{n-1}$
  terms for which $\mb{r}$ contains $1$ but not $2$, \ie
  $\mb{r}=(1,r_2,r_3,\ldots,r_n)$, with $3\le{}r_2<r_3<\ldots<r_n\le{}n+m$.
  For any given $\mb{r}$, we denote by $\cb{r}$ the vector
  of the $m$, among the $n+m$, row indices for $\Delta_{n,m}(\tau)$ that do
  not belong to $\mb{r}$ (recall that $2$ always belongs to
  $\mb{\rr}$). Notice that the elements of the 
  $k$-th row of $\Delta_{n,m}(\tau)$ have exponent $\mu_{k}$.
  Denoting by $\mb{\mu_r}$ the vector the $i$-th element of
  which is $\mu_{r_i}$, we have that:
  \begin{enumerate}
  \item
    $\det(S(\Delta_{n,m}(\tau);\mb{r},\mb{c}))$ is the $n\times{}n$
    generalized Vandermonde determinant
    $\GVD(\tau^{\alpha}\mb{x};\mb{\mu_\mb{r}})$, multiplied
    by $\tau^{M}$ if $J\in\mb{r}$.
  \item
    $\det(\compl(\Delta_{n,m}(\tau);\mb{r},\mb{c}))$ is the $m\times{}m$
    generalized Vandermonde determinant
    $\GVD(\tau^{\beta}\mb{y};\mb{\mu_\mb{\rr}})$, multiplied
    by $\tau^{M}$ if $I\in\mb{\rr}$.
  \end{enumerate}
  We can, thus, simplify the expansion in \eqref{expansion1} to get:
  \begin{equation}\label{expansion2}
    \begin{aligned}
      D_{n,m}(\tau)
      &=\sum_{\{\mb{r}\mid{}1\in\mb{r},2\nin\mb{r}\}}
      (-1)^{|\mb{r}|+\frac{n(n+1)}{2}+J+1}\,h(\mb{r},\tau;I,J)
      \GVD(\tau^{\alpha}\mb{x};\mb{\mu_r} )
      \GVD(\tau^{\beta}\mb{y}; \mb{\mu_{\rr}})\\
      &=\sum_{\{\mb{r}\mid{}1\in{}\mb{r},2\nin\mb{r}\}}
      (-1)^{|\mb{r}|+\frac{n(n+1)}{2}+J+1}\,h(\mb{r},\tau;I,J)
      \tau^{\alpha|\mb{\mu_r}|+\beta|\mb{\mu_\rr}|}\,
      \GVD( \mb{x};\mb{\mu_r} )  \GVD( \mb{y}; \mb{\mu_{\rr}}),
    \end{aligned}
  \end{equation}
  where
  \begin{equation*}
    h(\mb{r},\tau;I,J)=\begin{cases}
      1,&I\in\mb{r}\text{ and } J\nin\mb{r},\\
      \tau^{2M},&I\nin\mb{r}\text{ and } J\in\mb{r},\\
      \tau^{M},&\text{otherwise}.
    \end{cases}
  \end{equation*}

  In the remainder of the proof we seek to find the unique term in the
  expansion \eqref{expansion2} that corresponds to the minimum order
  of $\tau$, or, equivalently, the minimum exponent for $\tau$.
  Since $\alpha>\beta\ge{}0$, for any $\mb{r}$, with $I\in\mb{r}$ and
  $J\nin\mb{r}$, the exponent of $\tau$ is:
  \begin{equation*}
    \alpha |\mb{\mu_r}|+\beta |\mb{\mu_{\rr}}|
    <\alpha |\mb{\mu_r}|+\alpha |\mb{\mu_{\rr}}|
    =\alpha|\mb{\mu}|\le{}M,
  \end{equation*}
  where we used the fact that:
  \begin{equation*}
    |\mb{\mu_r}|+|\mb{\mu_\rr}|
    =\sum_{i=1}^{n}\mu_{r_i}+\sum_{i=1}^{m}\mu_{\rr_i}
    =\sum_{i\in\mb{r}}\mu_i+\sum_{i\in\cb{r}}\mu_i
    =\sum_{i=1}^\ell\mu_i=|\mb{\mu}|.
  \end{equation*}
  This implies that the terms in \eqref{expansion2} that
  correspond to the row vectors $\mb{r}$ that contain $J$ cannot be
  the terms of minimal order of $\tau$, since for these terms the
  exponent of $\tau$ is at least
  \begin{equation*}
    \alpha |\mb{\mu_r}|+\beta |\mb{\mu_{\rr}}|+M
    >\beta |\mb{\mu_r}|+\beta |\mb{\mu_{\rr}}|+M
    =\beta |\mb{\mu}|+M
    \ge{}M.
  \end{equation*}

  For the remaining terms, \ie for those $\mb{r}$ that do not
  contain $J$, we have $h(\mb{r},\tau;I,J)=1$.
  For these terms the exponent of $\tau$ is $\alpha |\mb{\mu_r}|+\beta
  | \mb{\mu_{\rr}}|$.
  Since $\alpha>\beta$, we may write $\alpha=\beta+\theta$ for some
  $\theta>0$. This gives:
  \begin{equation*}
    \alpha |\mb{\mu_r}|+\beta |\mb{\mu_{\rr}}|
    =(\beta+\theta) |\mb{\mu_r}|+\beta |\mb{\mu_{\rr}}|
    =\beta |\mb{\mu}| + \theta |\mb{\mu_r}|.
  \end{equation*}
  Clearly, in this case, the quantity
  $\alpha |\mb{\mu_r}|+\beta |\mb{\mu_{\rr}}|$ 
  attains its minimum when  $|\mb{\mu_r}|$ is minimal.
  We distinguish between the following cases:
  \begin{itemize}
  \item
    $(I,J)=(3,4)$.
    In this case $|\mb{\mu_r}|$ attains its minimal value if and only
    if $\mb{r}$ is equal to $\mb{\rho}=(1,3,5,6,\ldots,n+2)$.
    Furthermore,
    \begin{align*}
      |\mb{\mu_{\rho}}|&=\mu_1+\mu_3+\sum_{i=5}^{n+2}\mu_i
      =\mu_1+\mu_3+\sum_{i=4}^{n+2}\mu_i-\mu_J,\\
      |\mb{\mu_{\bar{\rho}}}|&=\mu_2+\mu_4+\sum_{i=n+3}^{\ell}\mu_i
      =\mu_2+\mu_J+\sum_{i=n+3}^{\ell}\mu_i
    \end{align*}
    and
    \begin{align*}
      |\mb{\rho}|+\frac{n(n+1)}{2}+J+1
      &=\sum_{i=1}^{n+2}{i}-(2+4)+\frac{n(n+1)}{2}+4+1\\
      &=\frac{(n+2)(n+3)}{2}+\frac{n(n+1)}{2}-1\\
      &=n^2+3n+3-1\\
      &=(n+1)(n+2),
    \end{align*}
    which is even for any $n\ge{}2$.
  \item $I\in\{3,4\}$ and $J=5$.
    In this case $|\mb{\mu_r}|$ attains its minimal value if and only
    if $\mb{r}$ is equal to $\mb{\rho}=(1,3,4,6,\ldots,n+2)$.
    Furthermore,
    \begin{align*}
      |\mb{\mu_{\rho}}|&=\mu_1+\mu_3+\mu_4+\sum_{i=6}^{n+2}
      =\mu_1+\mu_3+\sum_{i=4}^{n+2}-\mu_J,\\
      |\mb{\mu_{\bar{\rho}}}|&=\mu_2+\mu_5+\sum_{i=n+3}^{\ell}
      =\mu_2+\mu_J+\sum_{i=n+3}^{\ell}\mu_i,
    \end{align*}
    and
    \begin{align*}
      |\mb{\rho}|+\frac{n(n+1)}{2}+J+1
      &=\sum_{i=1}^{n+1}{i}-(2+5)+\frac{n(n+1)}{2}+5+1\\
      &=\frac{(n+2)(n+3)}{2}+\frac{n(n+1)}{2}-1\\
      &=n^2+3n+3-1\\
      &=(n+1)(n+2),
    \end{align*}
    which is again even for any $n\ge{}2$.
  \end{itemize}
  We can thus rewrite \eqref{expansion2} in the following form:
  \begin{equation*}
    D_{n,m}(\tau)=
    \tau^{\alpha|\mb{\mu_\rho}|+\beta|\mb{\mu_{\bar\rho}}|}\,
    \GVD( \mb{x};\mb{\mu_{\rho}})\GVD(\mb{y};\mb{\mu_{\bar\rho}})
    +\Theta(\tau^{\alpha|\mb{\mu_\rho}|+\beta|\mb{\mu_{\bar\rho}}|+1}).
  \end{equation*}
  The lemma immediately follows from the positivity of the generalized
  Vandermonde determinants $\GVD(\mb{x};\mb{\mu_{\rho}})$ and
  $\GVD(\mb{y};\mb{\mu_{\bar\rho}})$.
\end{proof}

We end with the following lemma, where we perform the asymptotic
analysis of the generalized version of the determinant that arises in
the upper bound tightness construction in Section \ref{sec:lb} when
$R\equiv[3]$.

\newcommand{\uu}{{\bar{u}}}
\newcommand{\vv}{{\bar{v}}}
\begin{lemma}\label{lem:det3}
  Fix three integers $m\ge{}2$, $n\ge{}2$ and $k\ge{}2$, with
  $n+m+k\ge{}7$.
  Let $E_{n,m,k}(\tau;\mb{\mu})$ be the $(n+m+k)\times(n+m+k)$ determinant:
  \begin{equation*}
    -\left|
      \begin{array}{ccccccccc}
        (x_1\tau^2)^{\mu_1}&\cdots &(x_n\tau^2)^{\mu_1} 
        &0&\cdots&0&0&\cdots&0\\[2pt]
        0&\cdots&0
        &(y_1\tau)^{\mu_2}&\cdots &(y_m\tau)^{\mu_2} 
        &0&\cdots&0\\[2pt]
        0&\cdots&0&0&\cdots&0
        &z_1^{\mu_3}&\cdots &z_k^{\mu_3}\\[2pt]
        (x_1\tau^2)^{\mu_4}&\cdots &(x_n\tau^2)^{\mu_4} 
        &\tau^{M}(y_1\tau)^{\mu_4}&\cdots&\tau^{M}(y_n\tau)^{\mu_4} 
        &\tau^{M}z_1^{\mu_4}&\cdots&\tau^{M}z_n^{\mu_4}\\[3pt]
        \tau^{M}(x_1\tau^2)^{\mu_5}&\cdots
        &\tau^{M}(x_n\tau^2)^{\mu_5} 
        &(y_1\tau)^{\mu_5}&\cdots&(y_m\tau)^{\mu_5}
        &\tau^{M}z_1^{\mu_5}&\cdots
        &\tau^{M}z_n^{\mu_5}\\[3pt]
        \tau^{M}(x_1\tau^2)^{\mu_6}&\cdots
        &\tau^{M}(x_n\tau^2)^{\mu_6} 
        &\tau^{M}(y_1\tau)^{\mu_6}&\cdots
        &\tau^{M}(y_m\tau)^{\mu_6}
        &z_1^{\mu_6}&\cdots&z_m^{\mu_6}\\[3pt]
        (x_1\tau^2)^{\mu_7}  &\cdots &(x_n\tau^2)^{\mu_7} 
        &(y_1\tau)^{\mu_7}&\cdots&(y_m\tau)^{\mu_7}        
        &z_1^{\mu_7}&\cdots&z_k^{\mu_7}\\[3pt]
        \vdots  & \ddots  & \vdots 
        &\vdots  & \ddots  & \vdots 
        & \vdots  & \ddots  & \vdots  \\[3pt]
        (x_1\tau^2)^{\mu_\ell}  &
        \cdots &(x_n\tau^2)^{\mu_\ell}
        &(y_1\tau)^{\mu_\ell}
        &\cdots&(y_m\tau)^{\mu_\ell}
        &z_1^{\mu_\ell}
        &\cdots&z_k^{\mu_\ell}\\[3pt]
      \end{array}
    \right|
  \end{equation*}
  where
  $0<x_1<x_2<\ldots<x_n$, $0<y_1<y_2<\ldots<y_m$, $0<z_1<z_2<\ldots<z_k$,
  $\ell=n+m+k$, $\mb{\mu}=(\mu_1,\mu_2,\ldots,\mu_\ell)$, with
  $0\le\mu_1\le\mu_2\le\mu_3<\mu_4<\mu_5<\ldots<\mu_\ell$,
  $M\ge{}2|\mb{\mu}|$ and $\tau>0$. Then,
  \begin{equation*}
    E_{n,m,k}(\tau;\mb{\mu})=C'\tau^\xi+\Theta(\tau^{\xi+1}),\quad
    \xi=2\left(\mu_1+\mu_4+\sum_{i=7}^{n+4}\mu_i\right)
    +\mu_2+\mu_5+\sum_{i=n+5}^{n+m+2}\mu_i,
  \end{equation*}
  where $C'$ is a positive constant independent of $\tau$.
\end{lemma}


\begin{proof} 
  We write $E_{n,m,k}(\tau)$ instead of $E_{n,m,k}(\tau;\mb{\mu})$,
  suppressing $\mb{\mu}$ in the notation.
  We denote by $\eE_{n,m,k}(\tau)$ the matrix corresponding to
  the determinant $-E_{n,m,k}(\tau)$.
  If we apply Laplace's expansion theorem with respect to the first $n$
  columns, \ie when $\mb{c}=(1,2,\ldots,n)$, we get:
  \begin{align}
    E_{n,m,k}(\tau)
    &=-\sum_{\mb{r}}(-1)^{|\mb{r}| + |\mb{c}|}\,
    \det(S(\eE_{n,m,k}(\tau);\mb{r},\mb{c}))
    \det(\compl(\eE_{n,m,k}(\tau);\mb{r},\mb{c}))\notag\\
    &=\sum_{\mb{r}}
    (-1)^{|\mb{r}|+\frac{n(n+1)}{2}+1}\,
    \det(S(\eE_{n,m,k}(\tau);\mb{r},\mb{c}))
    \det(\compl(\eE_{n,m,k}(\tau);\mb{r},\mb{c})).
    \label{exp1} 
  \end{align}

  The above sum consists of $ \binom{n+m+k}{n}$ terms. Among these terms: 
  \begin{enumerate}
  \item
    all those for which $\mb{r}$ contains the second or third
    row vanish (the corresponding row of
    $S(\eE_{n,m,k}(\tau);\mb{r},\mb{c})$ consists of zeros), and
  \item
    all those for which $\mb{r}$ does not contain the first row vanish 
    (in this case there exists a row of
    $\compl(\eE_{n,m,k}(\tau);\mb{r},\mb{c})$ that consists of
    zeros).
  \end{enumerate}
  The remaining terms of the expansion are the $\binom{n+m+k-3}{n-1}$
  terms for which $\mb{r}=(1,r_2,r_3,\ldots,r_n)$, with
  $4\le{}r_2<r_3<\ldots<r_n\le{}n+m+k$. As a result, the expansion in
  \eqref{exp1} simplifies to:
  \begin{equation}\label{exp2}
    E_{n,m,k}(\tau)
    =\sum_{\{\mb{r}\mid{}1\in\mb{r},2,3\nin\mb{r}\}}\,
    (-1)^{|\mb{r}|+\frac{n(n+1)}{2}+1}\,
    \det(S(\eE_{n,m,k}(\tau);\mb{r},\mb{c}))
    \det(\compl(\eE_{n,m,k}(\tau);\mb{r},\mb{c})).
  \end{equation}

  For any given $\mb{r}$, we denote by $\cb{r}$ the vector of the
  $m+k$ row indices for $\eE_{n,m,k}(\tau)$ that do not belong
  to $\mb{r}$. Moreover, $\mb{\mu_r}$ is the vector the $i$-th element
  of which is $\mu_{r_i}$.
  As in the proof of Lemma \ref{lem:det2}, we seek to find the
  unique minimum term in the expansion \eqref{exp2} that corresponds
  to the minimum order of $\tau$, or, equivalently, the minimum
  exponent for $\tau$.

  Let us denote by $\rR$ the set of row vectors
  $\rR=\{\mb{r}\mid{}1,4\in\mb{r}\text{ and }2,3,5,6\nin\mb{r}\}$.
  For any $\mb{r}\in\rR$, observe that:
  \begin{enumerate}
  \item
    $\det(S(\eE_{n,m,k}(\tau);\mb{r},\mb{c}))$ is the $n\times{}n$
    generalized Vandermonde determinant $\GVD(\tau^{2}\mb{x};\mb{\mu_r})$.
  \item
    $\det(\compl(\eE_{n,m,k}(\tau);\mb{r},\mb{c}))$ is the
    $(m+k)\times{}(m+k)$ determinant $D_{m,k}(\tau;3,4,\mb{\mu_\rr})$
    of Lemma \ref{lem:det2} multiplied by $(-1)^{4+1}=-1$, with
    $\mb{x}\sub\mb{y}$,
    $\mb{y}\sub\mb{z}$,
    $\mb{\mu}\sub\mb{\mu_\rr}$,
    $(I,J)=(3,4)$, $\alpha\sub{}1$, $\beta\sub{}0$ and $M\sub{}M$
    (since $M\ge{}2|\mb{\mu}|>|\mb{\mu_\rr}|$,
    the condition for $M$ in Lemma \ref{lem:det2} is satisfied).
  \end{enumerate}
  We can, thus, rewrite the expansion in \eqref{exp2} to get:
  \begin{equation}\label{exp3}
    \begin{aligned}
    E_{n,m,k}(\tau)
    &=\sum_{\mb{r}\in\rR}
    (-1)^{|\mb{r}|+\frac{n(n+1)}{2}+1}\,
    \GVD(\tau^2\mb{x};\mb{\mu_r})\,
    (-D_{m,k}(\tau;3,4,\mb{\mu_{\rr}}))\\
    &\quad+\sum_{\mb{r}\nin\rR}
    (-1)^{|\mb{r}|+\frac{n(n+1)}{2}+1}\,
    \det(S(\eE_{n,m,k}(\tau);\mb{r},\mb{c}))
    \det(\compl(\eE_{n,m,k}(\tau);\mb{r},\mb{c}))\\
    &=\sum_{\mb{r}\in\rR}
    (-1)^{|\mb{r}|+\frac{n(n+1)}{2}}\,\tau^{2|\mb{\mu_r}|}\,
    \GVD(\mb{x};\mb{\mu_r}) D_{m,k}(\tau;3,4,\mb{\mu_\rr})\\
    &\quad+\sum_{\mb{r}\nin\rR}
    (-1)^{|\mb{r}|+\frac{n(n+1)}{2}+1}\,
    \det(S(\eE_{n,m,k}(\tau);\mb{r},\mb{c}))
    \det(\compl(\eE_{n,m,k}(\tau);\mb{r},\mb{c}))
  \end{aligned}
  \end{equation}
  By Lemma \ref{lem:det2} we have:
  \begin{align*}
    D_{m,k}(\tau;3,4,\mb{\mu_\rr})
    =C_{\mb{r}}\,\tau^{1\cdot|\mb{\mu_\uu}|+0\cdot|\mb{\mu_\vv|}}
    +\Theta(\tau^{1\cdot|\mb{\mu_\uu}|+0\cdot|\mb{\mu_\vv|}+1})
    =C_{\mb{r}}\,\tau^{|\mb{\mu_\uu}|}+\Theta(\tau^{|\mb{\mu_\uu}|+1}),
  \end{align*}
  where $\mb{\uu}=(2,5,\rr_5,\ldots,\rr_{m+2})$,
  $\mb{\vv}=(3,6,\rr_{m+3},\ldots,\rr_{m+k})$ and $C_{\mb{r}}>0$.
  Hence, for any $\mb{r}\in\rR$, the term in the expansion of
  $E_{n,m,k}(\tau)$ that corresponds to $\mb{r}$ becomes:
  \begin{equation*}
    (-1)^{|\mb{r}|+\frac{n(n+1)}{2}}\,C_{\mb{r}}\,
    \tau^{2|\mb{\mu_r}|+|\mb{\mu_\uu}|}\,
    \GVD(\mb{x};\mb{\mu_r})
    +\Theta(\tau^{2|\mb{\mu_r}|+|\mb{\mu_\uu}|+1}).
  \end{equation*}  
  From this expression we deduce that the minimum exponent of $\tau$
  for any specific $\mb{r}\in\rR$ is:
  \begin{equation*}
    2 |\mb{\mu_r}|+|\mb{\mu_{\uu}}|
    <2 |\mb{\mu_r}|+|\mb{\mu_{\rr}}|
    <2 |\mb{\mu_r}|+2|\mb{\mu_{\rr}}|
    =2|\mb{\mu}|\le{}M,
  \end{equation*}
  where we used the fact that:
  \begin{equation*}
    |\mb{\mu_r}|+|\mb{\mu_\rr}|
    =\sum_{i=1}^{n}\mu_{r_i}+\sum_{i=1}^{m+k}\mu_{\rr_i}
    =\sum_{i\in\mb{r}}\mu_i+\sum_{i\in\cb{r}}\mu_i
    =\sum_{i=1}^\ell\mu_i=|\mb{\mu}|.
  \end{equation*}
  On the other hand, the terms in \eqref{exp3} that
  correspond to the row vectors $\mb{r}\nin\rR$ cannot be the terms of
  minimal order of $\tau$, since for these terms the exponent of
  $\tau$ is greater than $M$.
  We can thus restrict our attention to the terms for which
  $\mb{r}\in\rR$, and rewrite \eqref{exp3} as:
  \begin{align*}
    E_{n,m,k}(\tau)
    =\sum_{\mb{r}\in\rR}
    \left( (-1)^{|\mb{r}|+\frac{n(n+1)}{2}}\,C_{\mb{r}}\,
    \tau^{2|\mb{\mu_r}|+|\mb{\mu_\uu}|}\,
    \GVD(\mb{x};\mb{\mu_r})
    +\Theta(\tau^{2|\mb{\mu_r}|+|\mb{\mu_\uu}|+1})\right)+\Omega(\tau^{M}).
  \end{align*}
  From the expression above, we infer that the term of
  $E_{n,m,k}(\tau)$ for which the exponent of 
  $\tau$ is minimal is the term for which the quantity
  $2|\mb{\mu_r}|+|\mb{\mu_\uu}|$ is minimized.
  However, we have that:
  \begin{equation*}
    2|\mb{\mu_r}|+|\mb{\mu_\uu}|
    =2|\mb{\mu_r}|+|\mb{\mu_\uu}|+|\mb{\mu_\vv}|-|\mb{\mu_\vv}|
    =|\mb{\mu_r}|+|\mb{\mu_r}|+|\mb{\mu_\rr}|-|\mb{\mu_\vv}|
    =|\mb{\mu_r}|+|\mb{\mu}|-|\mb{\mu_\vv}|.
  \end{equation*}
  So, minimizing $2|\mb{\mu_r}|+|\mb{\mu_\uu}|$ amounts to determining
  the vectors $\mb{r}$ and $\mb{\vv}$ for which the difference
  $|\mb{\mu_r}|-|\mb{\mu_\vv}|$ becomes minimal.
  Let $\mb{\rho}=(1,4,7,8,\ldots,n+4)$,
  $\mb{\bar{\rho}'}=(2,5,n+5,n+6,\ldots,n+m+2)$ and
  $\mb{\bar{\rho}''}=(3,6,n+m+3,n+m+4,\ldots,\ell)$. It is trivial to
  verify that 
  \begin{itemize}
  \item
    $|\mb{\mu_r}|>|\mb{\mu_\rho}|$, for all $\mb{r}\ne\mb{\rho}$, and
  \item
    $|\mb{\mu_\vv}|<|\mb{\mu_{\bar{\rho}''}}|$, for all
    $\mb{\vv}\ne\mb{\bar{\rho}''}$.
  \end{itemize}
  From this observation we deduce that the unique
  minimal value for $2|\mb{\mu_r}|+|\mb{\mu_\uu}|$ is attained when
  $\mb{r}$, $\mb{\uu}$ and $\mb{\vv}$ are equal to $\mb{\rho}$,
  $\mb{\bar{\rho}'}$ and $\mb{\bar{\rho}''}$, respectively.
  Moreover,
  \begin{align*}
    |\mb{\rho}|+\frac{n(n+1)}{2}
    &=\sum_{i=1}^{n+4}{i}-(2+3+5+6)+\frac{n(n+1)}{2}
    =\frac{(n+4)(n+5)}{2}-16+\frac{n(n+1)}{2}\\[5pt]
    &=\frac{(n^2+9n+20)+(n^2+n)}{2}-16
    =(n^2+5n+10)-16=(n-1)(n+6).
  \end{align*}
  Since $(n-1)(n+6)$ is even for any $n$, the term in the expansion of
  $E_{n,m,k}(\tau)$ corresponding to the minimum exponent for $\tau$
  becomes
  $C_{\mb{\rho}}\,\tau^{2|\mb{\mu_\rho}|+|\mb{\mu_{\bar{\rho}'}}|}\,
  \GVD(\mb{x};\mb{\mu_\rho})$.
  The claim in the statement of the lemma immediately follows from the
  positivity of $C_{\mb{\rho}}$ and $\GVD(\mb{x};\mb{\mu_\rho})$, and
  by observing that $2|\mb{\mu_\rho}|+|\mb{\mu_{\bar{\rho}'}}|$ equals
  $\xi$.
\end{proof} 

\end{document}